\documentclass[11pt]{article}
\usepackage{mathrsfs}
\usepackage{latexsym,amssymb,amsmath,amscd,amscd,amsthm,amsxtra,cancel,xypic}
\usepackage[utf8]{inputenc}
\usepackage[T1]{fontenc}
\usepackage{amssymb}
\usepackage{nicefrac,mathtools,enumitem}
\usepackage[colorinlistoftodos]{todonotes}
\usepackage{soul}
\usepackage{tikz}
\usepackage{xcolor}
\usepackage{geometry}
\usepackage[colorlinks=true, pdfstartview=FitV, linkcolor=blue, citecolor=blue, urlcolor=blue]{hyperref}

\def \le{\left}
\def\ri{\right}
\def \d{{\rm d}}

\DeclareMathOperator*{\Res}{Res}

\def\dn{\mathrm{dn}}
\def\res{\mathrm{res}}

\usepackage{times}
\usepackage{amsthm}
\usepackage{amsmath}
\usepackage{amsfonts}
\usepackage{amssymb}
\usepackage{color}
\usepackage{mathrsfs}
\usepackage{stmaryrd}
\SetSymbolFont{stmry}{bold}{U}{stmry}{m}{n}
\usepackage{pifont}
\usepackage{bm}

\usepackage{tikz-cd}
\usepackage{tikz}
\usepackage{graphicx}
\usepackage[percent]{overpic}
\usepackage{subfigure}

\newcommand{\be }{\begin{equation}}
   \newcommand{\ee }{\end{equation}}

\newcommand{\R}{\mathbb R}
\newcommand{\C}{\mathbb C}

\newcommand{\norm}[1]{\left\Vert#1\right\Vert}

\newcommand{\K}{\mathbb{K}}

\newcommand {\wh}{\widehat}

\usepackage{thmtools}

\DeclareMathOperator{\re}{Re}
\DeclareMathOperator{\im}{Im}

\declaretheoremstyle[spaceabove=0.25cm,spacebelow=0.25cm,notefont=\normalfont\bfseries, notebraces={(}{)}]{theorem}
\declaretheoremstyle[spaceabove=0.25cm,spacebelow=0.25cm,bodyfont=\normalfont,notefont=\normalfont\bfseries, notebraces={(}{)}]{noital}
\declaretheoremstyle[spaceabove=0.25cm,spacebelow=0.25cm,bodyfont=\normalfont\color{darkgreen},notefont=\normalfont\bfseries, notebraces={(}{)}]{green}
\declaretheoremstyle[spaceabove=0.25cm,spacebelow=0.25cm,bodyfont=\normalfont,notefont=\normalfont\bfseries,qed=$\qedsymbol$,notebraces={(}{)}]{proofstyle}

\declaretheorem[name=Theorem,numberwithin=section,style=theorem]{thm}

\declaretheorem[name=Proposition,sibling=thm,style=theorem]{pro}

\declaretheorem[name=Lemma,sibling=thm,style=theorem]{lem}
\declaretheorem[name=Definition,sibling=thm,style=noital]{defi}

\declaretheorem[name=Remark,sibling=thm,style=theorem]{rmk}
\declaretheorem[name=Riemann-Hilbert Problem,sibling=thm,style=theorem]{RHP}

\definecolor{lightblue}{rgb}{0.68, 0.85, 0.9}
\definecolor{lightred}{rgb}{1.0, 0.8, 0.8}
\definecolor{lightgreen}{rgb}{0.8, 1.0, 0.8}
\definecolor{darkgreen}{rgb}{0.0, 0.5, 0.0}
\definecolor{color12lines}{rgb}{0.5, 0.0, 0.0}
\definecolor{color13lines}{rgb}{0.0, 0.5, 0.0}
\definecolor{color23lines}{rgb}{0.0, 0.0, 0.6}
\definecolor{darkgreen}{rgb}{0.0, 0.39, 0.0}
\definecolor{lightgray}{rgb}{0.75, 0.75, 0.75}

\tikzset{
   branchpoint/.style={orange, thick, mark=x, mark options={orange, line width=1.25pt}}
}

\tikzset{
   singularpoint/.style={blue, mark=*, mark options={blue, mark size=1.25pt}}
}

\tikzset{
   stokeslabel/.style={gray!95,font=\tiny}
}

\tikzset{
   cutlabel/.style={orange,font=\tiny}
}

\tikzset{
   sheetlabel/.style={font=\small}
}

\tikzset{
   branchcut/.style={orange,dashed,semithick}
}

\tikzset{
   disccolor/.style={gray!12}
}

\tikzset{
   wall/.style={black,thick}
}

\tikzset{
   path/.style={thick,rounded corners}
}

\tikzset{
   witharrow/.style={
      decoration={markings, mark=at position #1 with {\arrow[sloped]{Latex}}},
      postaction={decorate}
   }
}

\tikzset{
   antistokesmark/.style={semithick, black, radius=1.5pt, fill=yellow}
}

\tikzset{
   withbackgroundrectangle/.style={show background rectangle, background rectangle/.style={fill=gray!7}}
}

\numberwithin{equation}{section}

\title{\bf Arbitrary-genus dark soliton gases in the defocusing nonlinear Schr\"{o}dinger hydrodynamics}

\date{}

\begin{document}
   
   \maketitle

\begin{center}
  
Marco Bertola$^{\dagger}$ \footnote{marco.bertola@concordia.ca}
Deng-Shan Wang $^{\ddagger}$ \footnote{dswang@bnu.edu.cn}
Peng Yan$^{\dagger}$ \footnote{peng.yan@mail.concordia.ca}
Dinghao Zhu$^{\ddagger, \S}$ \footnote{dhzhu@mail.bnu.edu.cn}

\bigskip
\begin{minipage}{0.7\textwidth}
\begin{small}
\begin{enumerate}
\item[$^{\dagger}$]  {\it   Department of Mathematics and Statistics, Concordia University\\ 1455 de Maisonneuve W., Montr\'eal, Qu\'ebec,
Canada H3G 1M8} 
\item[$^{\ddagger}$] {\it School of Mathematical Sciences, Beijing Normal University,\\ Beijing 100875, China}
\item[$^{\S}$] {\it Institut de Recherche en Math\'ematique et Physique, UCLouvain\\ Chemin du Cyclotron2,
B-1348 Louvain-La-Neuve, Belgium}

\end{enumerate}
\end{small}
\end{minipage}
\end{center}
   \begin{abstract}
{The defocusing nonlinear Schr\"{o}dinger hydrodynamics supports exact dark solitons under finite density boundary conditions. However, the dark soliton gas, an interacting ensemble of dark solitons, has not yet been studied. In this work, we introduce an arbitrary-genus potential of dark soliton gases by considering the limit of the $\mathcal{N}$-dark soliton as $\mathcal{N}\to \infty$. The large-space asymptotics and long-time evolution of this dark soliton gas potential are analytically investigated through Deift-Zhou nonlinear steepest descent approach. The genus-$N$ dark soliton gas potential approaches the genus-$N$ finite-gap solution as $x \to -\infty$ and the background $1$ as $x \to +\infty$. In the  long-time evolution, as the self-similar variable $\xi=x/t$ increases, the gas configuration exhibits a cascade of behaviours, passing from unmodulated and modulated genus-$N$ regions and progressively reducing the genus down to the planar region (unmodulated genus-$0$ region). Notably, the  evolution of lower-genus soliton gases can be embedded within that of higher-genus gases, exhibiting identical dynamics within specific regimes. This phenomenon is encoded by the underlying spectra. We also include numerical validations, in perfect agreement with the theoretical predictions.}

   \end{abstract}
   
   \tableofcontents

\section{Introduction}
   We investigate the dark soliton gas of the defocusing nonlinear Schr\"{o}dinger (NLS) equation
   \begin{equation}\label{dNLS}
      \mathrm{i} q_t+ q_{xx}- 2(|q|^2-1)q=0,  
   \end{equation}
   where $q=q(x,t)$ is a complex field.
   The usual form of the NLS equation is
   \begin{equation}
      \mathrm{i} u_t+ u_{xx}+2\nu|u|^2u=0,  \label{NLS}
   \end{equation}
   with $\nu=\pm 1$, and under the transformation $u(x,t)=q(x,t) e^{-2it}$ and taking $\nu=-1$, the equation \eqref{NLS} is converted into \eqref{dNLS}. The sign of the nonlinearity coefficient $\nu$ dictates its fundamental character: for $\nu=1$, the equation is focusing, promoting the coalescence of waves; for
   $\nu=-1$, it is defocusing (or repulsive), leading to wave dispersion. This work concerns the latter case, the defocusing nonlinear Schr\"{o}dinger (dNLS) equation.
   \par
   Physically, the dNLS equation emerges as a fundamental model in mean-field theory. Its most celebrated application is in the description of Bose-Einstein condensates with repulsive interatomic interactions, where it manifests as the Gross-Pitaevskii equation \cite{Pitaevskii-Stringari-2003,Kevrekidis-2008,Dodson2019}. In this case, $q(x,t)$ represents the macroscopic wave function of the condensate, and $|q(x,t)|^2$ corresponds to the density of the quantum fluid. The dNLS equation also governs the evolution of the envelope of electromagnetic waves in defocusing optical media and optical fibers with normal dispersion \cite{Kivshar-2003}, and describes the dynamics of deep-water waves in certain regimes \cite{Zakharov-1968}. A key feature of these systems is the existence of a stable, uniform background state $q(x,t)=q_0$ with complex constant $|q_0|=1$. The most striking nonlinear excitations supported on this background are not bright, localized humps, but rather dark solitons-persistent dips or ``phase shadows'' in the background density.
   \par
   A dark soliton is a traveling wave solution of the dNLS equation characterized by a localized {\it density depression} and an accompanying phase jump across its center. In the simplest form, the single dark soliton solution for a background density $\rho=1$ is given by
   \begin{equation}
      q(x,t) =  i v + \sqrt{1-v^2} \tanh\Big( \sqrt{1-v^2} (x - v t) \Big),
      \label{Dark-Soliton}
   \end{equation}
   where $v$ is the velocity of the dark soliton with $|v|<1$ \cite{Tsuzuki-1971}. The depth of the density minimum is proportional to $\sqrt{1-v^2}$, making stationary soliton with velocity $v=0$ called ``black soliton'' (zero minimum density) and moving ones with $v\neq 0$ called ``grey soliton'', which is a traveling, phase-shifting, non-vanishing dip on a uniform background. These structures are robust, particle-like entities that preserve their form upon mutual interaction, exhibiting only a phase shift, a hallmark of the integrability of the dNLS equation \eqref{dNLS}. The generalization to multi-soliton states, known as $\mathcal{N}$-dark soliton solutions describing the coherent, interacting dynamics of multiple dark solitons on a finite background \cite{Faddev,Zakharov-Shabat-1973}, can be derived by inverse scattering transform \cite{Faddev}. The interactions of dark solitons are repulsive in nature: two dark solitons never cross. Instead, they experience an effective acceleration/deceleration and a permanent displacement from their free trajectories. The study of $\mathcal{N}$-dark soliton collisions, their bound states in confined geometries, and their stability under perturbations has been a central topic in nonlinear wave theory for decades, with direct implications for experiments in Bose-Einstein condensates and nonlinear optics where such multi-soliton trains can be engineered \cite{Burger-PRL-1999,Frantzeskakis-JPA-2010}. Moreover, the asymptotic stability of $\mathcal{N}$-dark soliton solutions of the dNLS equation \cite{Jenkins2016} has been investigated based on the $\overline{\partial}$-generalization of the Deift-Zhou nonlinear steepest descent technique \cite{Deift-Zhou1993}.
   \par
   In parallel, significant advancements have been made in the analysis of dNLS hydrodynamics. These include rigorous asymptotic results on the long-time behaviours and Whitham modulation theory for dispersive shock waves (or undular bores) developing from problems of initial discontinuities \cite{El-Phys.D1995,Jenkins2015,Wang-Yan-2025}, the analysis of the stability of dark solitons and their dispersive generalizations \cite{Chiron-2012}, and hydrodynamic optical soliton tunneling \cite{Hoefer-El-2018}. Additionally, in real experiments, Trillo and his team also investigated the dam-break Riemann problem \cite{Trillo2017} and the piston Riemann problem \cite{Trillo2022}, which were examined within the framework of quantum hydrodynamics described by the dNLS equation.
   \par
   Over the past years, the study of soliton gases has been attracting growing interest. The notion of a soliton gas was first introduced by Zakharov in 1971 \cite{Zakharov-1971} for the case of the KdV equation by deriving a spectral kinetic equation that captures velocity and phase shifts from soliton collisions. Then this concept has been extended to the NLS equation by El and his collaborators \cite{El-AMK-2005,El-Tovbis-2020}. A {\it dark soliton gas} consists of a large collection of interacting dark solitons with a prescribed spectral distribution, where the individual solitons are uncorrelated in phase but interact via their asymptotic pairwise phase shifts. Recently, Tovbis and Wang \cite{Tovbis-Wang-2025} derived the nonlinear dispersion relations of dark soliton gas in dNLS equation by using the idea of thermodynamic limit of quasi-momentum and quasi-energy differentials on the underlying family of Riemann surfaces, constructing a kinetic equation for such gases. 
   The aforementioned studies primarily focus on the kinetic theory of soliton gases. To rigorously characterize their configurations, Girotti et al. \cite{CMP2021gas}, inspired by the primitive potentials introduced by Dyachenko et al. \cite{DZZp}, pioneered the use of the Riemann-Hilbert method to derive the asymptotics of KdV soliton gases. This framework was subsequently extended to the mKdV equation, including the interaction of a trial soliton with the soliton gas \cite{CPAM2023gas}. Further developments include the study of the focusing NLS soliton gas by Orsatti \cite{Fnlsgas}, the generalization to the genus-2 KdV soliton gas by Wang et al. \cite{KDVg2}, and the Camassa-Holm gas by Geng et al \cite{CHgas}.
   The above results have been confined to soliton gases originating from bright solitons on vanishing background and typically restricted to the genus-1 regime. The gap between genus-2 to higher-genus soliton gas (see Conjecture 5.1 in \cite{KDVg2}) is bridged in this current work.
   
   \paragraph{Description of results.}
   In this work, we construct an arbitrary-genus potential of dark soliton gases for the dNLS equation \eqref{dNLS} and examine their asymptotic behaviours. 
   The analytical challenges are twofold: first, the necessity of a nonzero background for dark solitons introduces extra difficulties for the setup of soliton gas Riemann-Hilbert problem (RHP)  and the subsequent deformations; second, the arbitrary-genus nature of the soliton gas requires the evaluation of transcendental functions on hyper-elliptic Riemann surfaces.
   Technically, we use a Joukowski uniformation variable to deal with the nonzero background and define a transformation between the associated RHPs that maps the jumps on the unit circle back to the real axis. For the genus-1 case, we construct for the first time an explicit solution to a novel class of model RHPs with strong symmetries and singularity. This representation is shown to be consistent with our previously established results for $N=1$.
   
   To state the main results of this work, we first introduce some notations.
   Select $2N$ points $\eta_j$ ($j = 1 , \cdots , 2N$) on the upper unit circle and ordered so that $0< \arg \eta_1< \dots < \arg \eta_{2N}< \pi$, and set $\eta_0=1, \eta_{2N+1}=-1$, and denote $\eta_j^{\rm re}:=\re \eta_j$ (note that  $\overline{\eta_j}=\eta_j^{-1}$). Denote the $g$-dimensional Riemann theta function by 
   \begin{equation}
      \Theta^{[g]}\left(\bm{z}; \bm{\tau}^{[g]}\right) = \sum_{\bm{n} \in \mathbb{Z}^g} \exp \Big( \pi i \bm{n}^T \bm{\tau}^{[g]} \bm{n} + 2\pi i \bm{n}^T \bm{z} \Big),
   \end{equation} 
   where $\bm{\tau}^{[g]}$ is the $g \times g$  symmetric matrix with positive definite imaginary part:
   since  we shall use several  Riemann theta functions with different dimensions, we burden the notation by the superscript notation $\Theta^{[g]}$ and $\bm{\tau}^{[g]}$ to distinguish them. 
   \paragraph{Large-space behaviour.}
   The following theorem describes the asymptotic behaviour of a limiting dNLS solution obtained by an accumulation of solitons chosen to fill the arcs of circle between $\eta_{2j-1}$ and $\eta_{2j}$ for $j=1,\cdots, N$.
   
   \begin{thm} \label{Theorem1}
      The initial arbitrary-genus dark soliton gas potential $q(x,0)$ described by reconstruction formula \eqref{rec} and RHP \ref{RHP2} for $t=0$ has the following asymptotic properties:
      \begin{enumerate}[label=(\arabic*), ref=\thethm(\arabic*)]
         \item \label{thm1} As $x\to+\infty$, the potential $q(x,0)$ exponentially approaches $1$:
         \begin{equation}\label{larx}
            q(x,0)=1+\mathcal{O}(e^{-cx}),
         \end{equation}
         where $c\in\mathbb{R}^+$ is a fixed constant. This corresponds to the nonzero background.
         
         \item \label{thm2} As $x\to-\infty$, the potential $q(x,0)$ is characterized by the genus-$N$  finite-gap solution:
         \begin{equation} \label{qx0minus}
            \begin{split}
               q(x,0) &= \frac{\sum_{j=0}^{N} \Big( \eta_{2j}^{\mathrm{re}} - \eta_{2j+1}^{\mathrm{re}} \Big)}{2} \cdot e^{-2i {x}g_{\infty}} \delta_{\infty}^{-2} \\
               &\quad \cdot \frac{\Theta^{[N]}\Big( -\frac{{x}\bm{\Omega+\Delta}}{2\pi} + 2\bm{J}(\infty) ; \bm{\tau}^{[N]} \Big)}{\Theta^{[N]}\Big( 2\bm{J}(\infty) ; \bm{\tau}^{[N]} \Big)} 
               \frac{\Theta^{[N]}\Big( \bm{0} ; \bm{\tau}^{[N]} \Big)}{\Theta^{[N]}\Big( \frac{{x}\bm{\Omega+\Delta}}{2\pi}  ; \bm{\tau}^{[N]} \Big)}  +\mathcal{O}(x^{-1}),
            \end{split}
         \end{equation}        
         where $\Theta^{[N]}(\bm{z};\bm{\tau}^{[N]})$ is the $N$-dimensional Riemann theta function and  $\bm{\tau}^{[N]}$ is given by \eqref{tauell}.  The vectors $\bm{\Omega} , \bm{\Delta}$ and $\bm{J}(\infty)$ are defined in \eqref{ome2} and \eqref{Jell}, respectively, and do not depend on $x$.  The parameters $g_{\infty}$ and $\delta_{\infty}$ are defined by \eqref{tginf} and \eqref{ttdel} with $\ell=N$, respectively.
      \end{enumerate}
   \end{thm}
   \par 
   \begin{figure} [htbp]
      \centering
      \includegraphics[width=13cm]{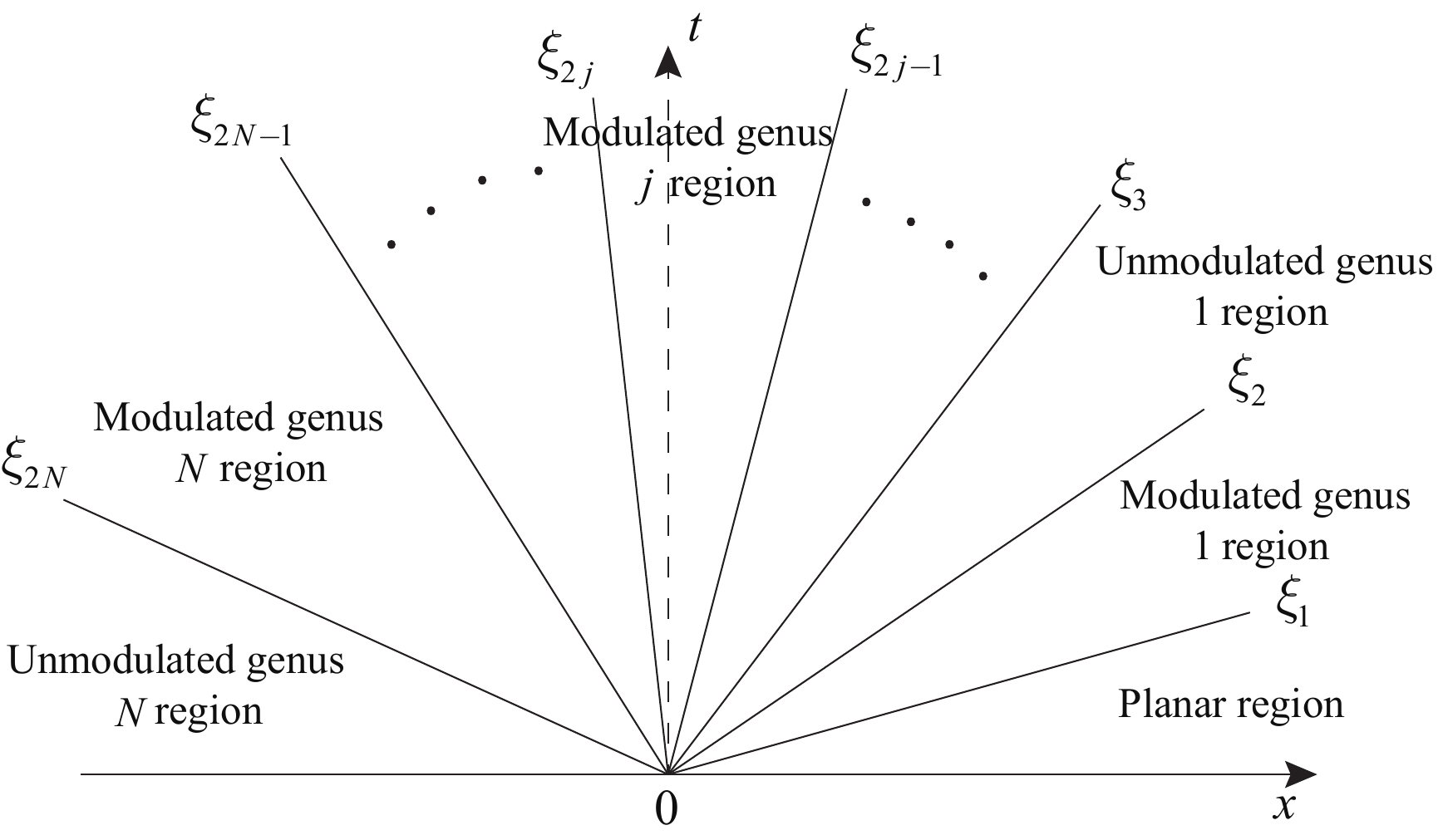}
      \caption{{\protect\small The different time evolution regions on the upper $x$-$t$ plane.}}
      \label{LongTime}
   \end{figure}
   
   \paragraph{Long-time behaviour.} Next,
   consider the long-time evolution of initial arbitrary-genus dark soliton gas potential $q(x,0)$ in Theorem \ref{Theorem1}. Define the self-similar variable $\xi={x}/{t}$. 
   There are $2N+1$ different asymptotic regions, as $t\to \infty$ with $\xi$ fixed,  which are illustrated qualitatively  in Figure \ref{LongTime}: these regions are separated by   $\xi_{2N}<\xi_{2N-1}< \dots< \xi_2< \xi_1 = 2\re(\eta_1)$.   The {\it modulated} regions are for  $\xi_{2j}<\xi< \xi_{2j-1}$ and the  {\it unmodulated} regions are for $\xi<\xi_{2N}$ and $\xi_{2j+1}<\xi<\xi_{2j}$. The boundaries, $\{\xi_j\}_{j=1}^{2N}$, between regions are determined by the degeneration of the Whitham equation, see Definition \ref{evenxis} and Definition \ref{defxiodd}. The behaviour of the soliton gas is described by the next Theorem.

   \begin{thm} \label{Theorem2}
      The long-time asymptotic behaviour of the arbitrary-genus dark soliton gas $q(x,t)$ is characterized by the following three types of regions depending on the self-similar variable $\xi$, i.e., as $t\to \infty$, we have
      
      \begin{enumerate}[label=(\arabic*), ref=\thethm(\arabic*)]
         \item \label{thm3} Planar region: for $\xi > \xi_1 = 2 \eta_1^{\rm re}$, the dark soliton gas behaves
         \begin{equation}\label{Qu1}
            q(x,t)=1+\mathcal{O}(e^{-ct}),
         \end{equation}
         where $c\in\mathbb{R}^+$ is a fixed constant.
         
         \item \label{thm4} Modulated genus-$\ell$ gas regions ($\ell = 1, \cdots, N$): for $\xi_{2 \ell}<\xi<\xi_{2\ell - 1}$, the dark soliton gas is  characterized by the modulated genus-$\ell$ finite-gap solution:
         \begin{equation} \label{so2}
            \begin{split}
               q(x,t) &= \frac{\alpha_{\ell}^{\rm re} + 1 + \sum_{j=0}^{\ell - 1} \left( \eta_{2j}^{\mathrm{re}} - \eta_{2j+1}^{\mathrm{re}} \right)}{2} \cdot e^{ - 2i{t}g_{\ell\infty}} \delta_{\ell\infty}^{-2} \\
               &\quad \cdot \frac{\Theta^{[\ell]}\Big( -\frac{{t}\tilde{\bm{\Omega}}+\tilde{\bm{\Delta}}}{2\pi} + 2\bm{J}_{\ell}(\infty)  ; \bm{\tau}^{[\ell]} \Big)}{\Theta^{[\ell]}\Big(2 \bm{J}_{\ell}(\infty); \bm{\tau}^{[\ell]} \Big)}  \frac{\Theta^{[\ell]}\Big(\bm{0} ; \bm{\tau}^{[\ell]} \Big)}{\Theta^{[\ell]}\Big( \frac{{t}\tilde{\bm{\Omega}}+\tilde{\bm{\Delta}}}{2\pi}  ; \bm{\tau}^{[\ell]} \Big)} + \mathcal{O}(t^{-1}),
            \end{split}
         \end{equation}
         where $\Theta^{[\ell]}(\bm{z}; \bm{\tau}^{[\ell]})$ is the $\ell$-dimensional Riemann theta function. Here $\alpha^{\rm re}_{\ell}:= \re \alpha_{\ell} $, and  $\alpha_{\ell}=\alpha_{\ell}(\xi)$ depends on the modulation parameter $\xi$ according to the Whitham modulation equation \eqref{xic}. The vectors $\tilde{\bm{\Omega}} + \tilde{\bm{\Delta}}$, $\bm\tau$ and $\bm{J}_{\ell}(\infty)$ are defined in \eqref{ome2} and \eqref{Jell}, respectively and also depend on the modulation parameter $\xi$.  The parameters $g_{\ell\infty}$ and $\delta_{\ell\infty}$ are defined by \eqref{tginf} and \eqref{ttdel}, respectively (and depend on $\xi$). 
         
         \item \label{thm5} Unmodulated genus-$\ell$ gas regions ($\ell = 1, \dots, N$): for $\xi_{2 \ell + 1}< \xi < \xi_{2 \ell}$, the dark soliton gas is  characterized by the unmodulated genus-$\ell$ finite-gap solution, which takes the same form as \eqref{so2} but now all the parameters $ \bm \tau, \bm J(\infty), \bm d_\ell$ do not depend on $\xi$ and $g_{\ell\infty}$ and $\bm \Omega$ depend on $\xi$ in an affine way:
         \begin{equation} 
            \begin{split}
               q(x,t) &= \frac{ \sum_{j=0}^{\ell} \left( \eta_{2j}^{\mathrm{re}} - \eta_{2j+1}^{\mathrm{re}} \right)}{2} \cdot e^{- 2i \hat{g}_{ \ell \infty}} \hat{\delta}_{ \ell \infty}^{-2} \\
               &\quad \cdot \frac{\Theta^{[\ell]} \Big( -\frac{\hat{t\bm{\Omega}}+\hat{\bm{\Delta}}}{2\pi} + \hat{2\bm{J}}_{ \ell }(\infty) ; \hat{\bm{\tau}}^{[\ell]} \Big)}{\Theta^{[\ell]} \Big( 2\hat{\bm{J}}_{ \ell }(\infty) ; \hat{\bm{\tau}}^{[\ell]} \Big)}  \frac{\Theta^{[\ell]}\Big( \bm{0} ; \hat{\bm{\tau}}^{[\ell]} \Big)}{\Theta^{[\ell]}\Big( \frac{\hat{t\bm{\Omega}}+\hat{\bm{\Delta}}}{2\pi}  ; \hat{\bm{\tau}}^{[\ell]} \Big)} + \mathcal{O}(t^{-1}),
            \end{split}
         \end{equation}
         where the vectors $\hat{\bm{\Omega}} + \hat{\bm{\Delta}}$, $\hat{\bm{J}}_{\ell}(\infty)$ and $\hat{\bm{\tau}}^{[\ell]}$ are obtained by replacing $\alpha_{\ell}$ with $\eta_{2\ell}$ in Section \ref{mod}.
      \end{enumerate}
   \end{thm}
   \par

    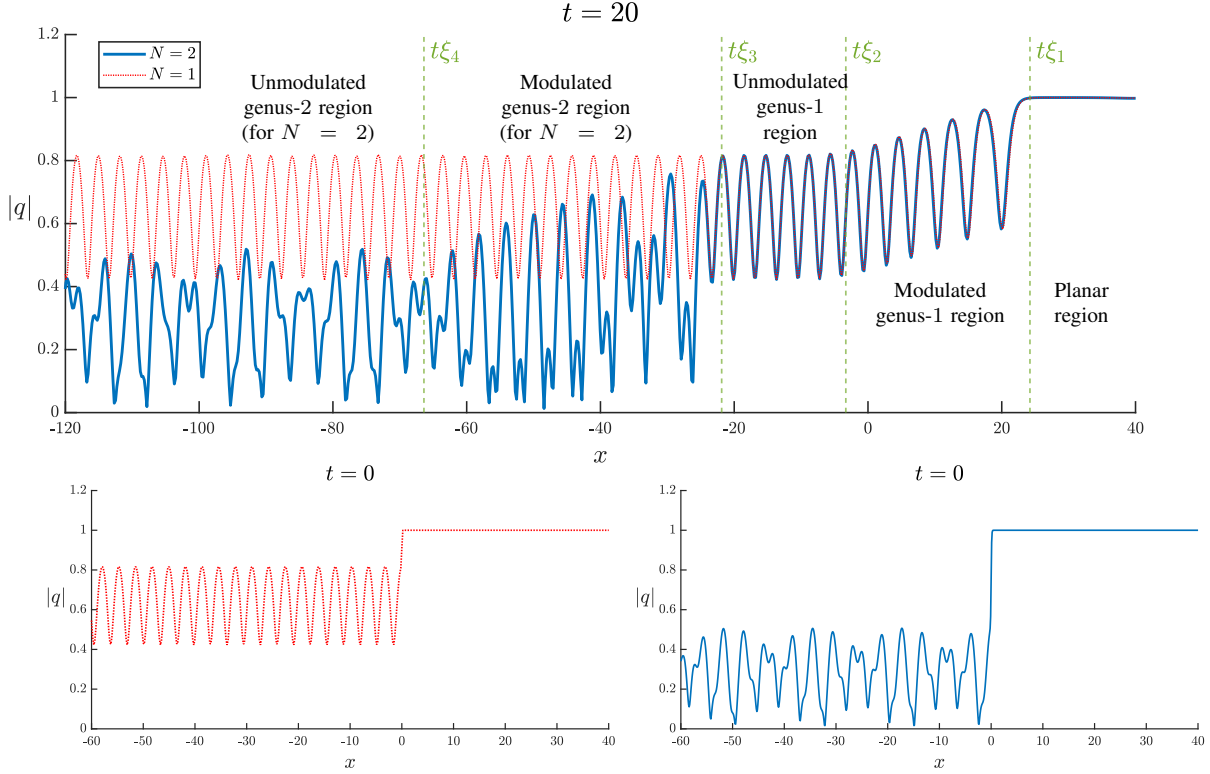
\begin{figure}[t]
   	\centering
   	\begin{tikzpicture}
   		
   		\node[anchor=south west, inner sep=0] (image) at (0,0) {\includegraphics[width=15cm]{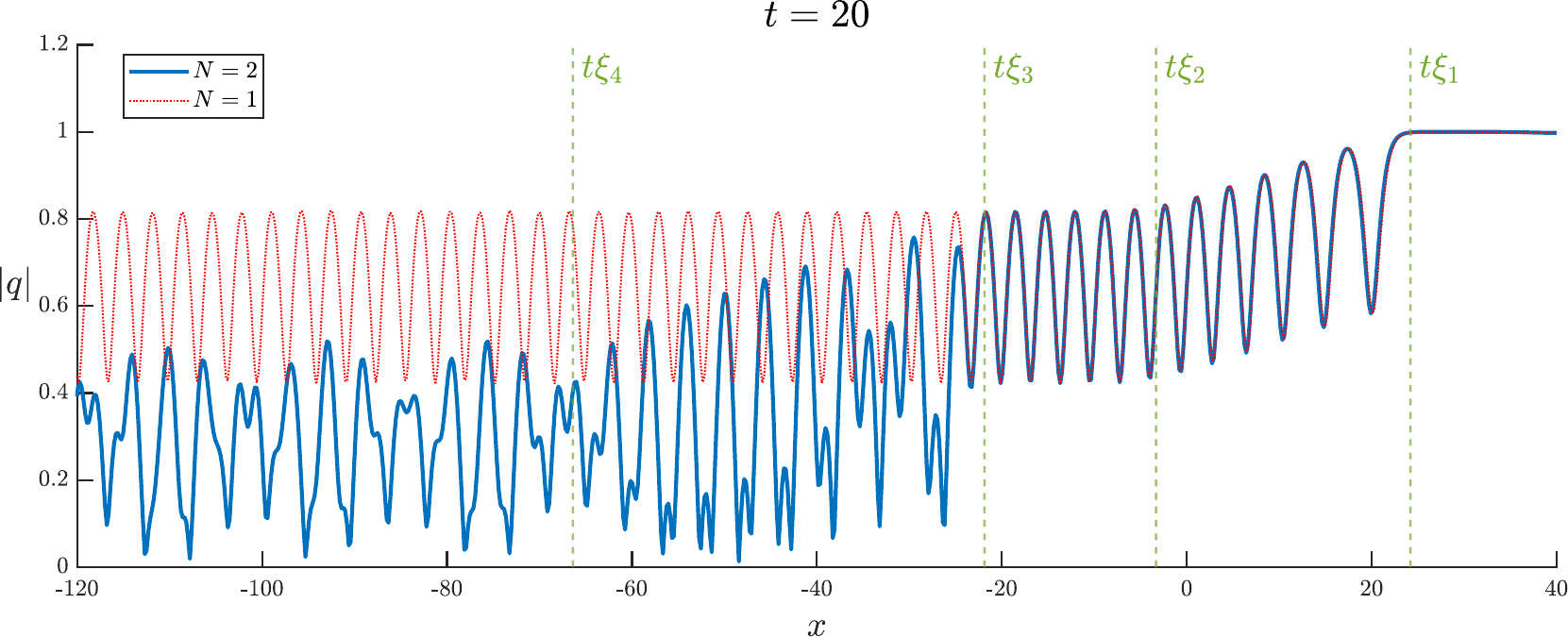}};

   		\newcommand{\imgwidth}{11}
   		\newcommand{\imgheight}{7}

   		\node[black, font=\scriptsize, align=center, text width=4cm] at (0.36*\imgwidth, 0.67*\imgheight) {Unmodulated \\ genus-$2$ region \\
   		(for $N=2$)};
   		\node[black, font=\scriptsize, align=center, text width=4cm] at (0.67*\imgwidth, 0.67*\imgheight) {Modulated \\ genus-$2$ region\\
   			(for $ N = 2 $)};
   		\node[black, font=\scriptsize, align=center, text width=2cm] at (0.94*\imgwidth, 0.67*\imgheight) {Unmodulated \\ genus-$1$ \\ region};
   		\node[black, font=\scriptsize, align=center, text width=4cm] at (1.29*\imgwidth, 0.3*\imgheight) {Planar \\ region};
   		\node[black, font=\scriptsize, align=center, text width=4cm] at (1.12*\imgwidth, 0.3*\imgheight) {Modulated \\ genus-$1$ region};
   	\end{tikzpicture}
	{\includegraphics[width=7.5cm]{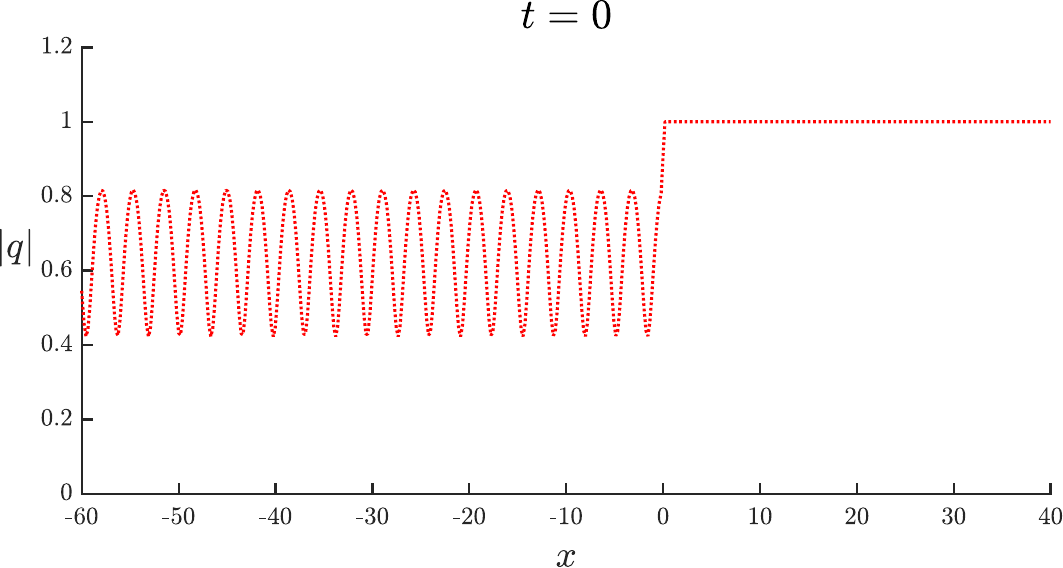}}\hspace{0.2cm}
	{\includegraphics[width=7.5cm]{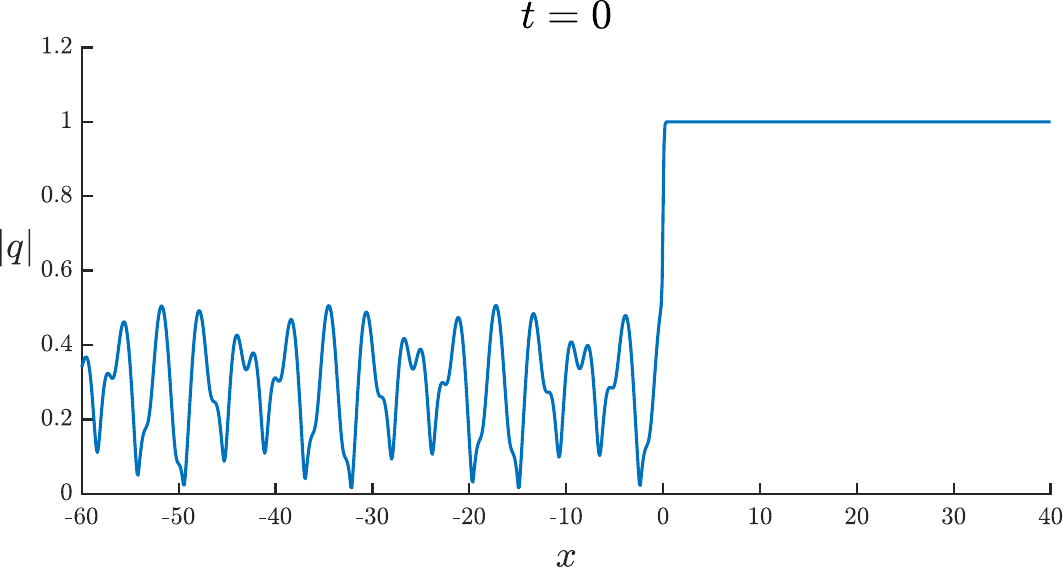}}
	
   	\caption{\small Example with interpolation function $r(z) \equiv 1$, see  \eqref{inpo}. Numerical simulation results of equation \eqref{dNLS} for $N = 1$ and $ N = 2 $, with the initial condition given by \eqref{dirctpl} (interpolated smoothly with a convex combination using $\tanh(x)$).  The initial discrete spectrum restricted by setting $\eta_1 = e^{ 0.293 i \pi } $, $\eta_2 = e^{ 0.423 i \pi } $ (for $N=1$) complemented by   $\eta_3 = e^{ 0.567 i \pi } $, $\eta_4 = e^{ 0.811 i \pi } $ for $N=2$. The breaking line (green) can be calculated by \eqref{xiodd} and \eqref{xieven}: $ \xi_1 \approx 1.211 $, $ \xi_2 \approx -0.168 $, $ \xi_3 \approx -1.094 $, $ \xi_4 \approx -3.320 $. Observe that the $N=1$ and $N=2$ solutions are (exponentially) close to each other in the first modulated/unmodulated region. For the initial conditions (in the second row), we employed the function ${\rm trans}(x):= \frac{1 + \tanh(10 x)}{2} $ and let $q(x,0) := q(x<0,0)(1-{\rm trans}(x)) + q(x>0,0){\rm trans}(x) $ to bridge the discontinuity of the initial value.}
   	\label{RHPsol}
   \end{figure}
   
   We have established asymptotic results for the general genus-$N$ dark soliton gas. Now we focus on the specific case of $N=1$.
   On one hand, the genus-1 model RHP \ref{rh24} in the $z$-plane can be solved explicitly, yielding an alternative expression for the leading-order asymptotics. On the other hand, the squared modulus of the leading term in the genus-1 case can be elegantly represented in terms of Jacobi elliptic functions.
  In Figure \ref{RHPsol}, we present direct numerical simulations for two types of initial data with different genera , namely:
   \begin{equation}\label{dirctpl}
   	q(x,0) = \begin{cases}
   		\text{leading term of \eqref{qx0N1}, i.e., $ N=1 $  }, & x<0, \\
   		1, & x>0,
   	\end{cases}
   	\quad 
   	q(x,0) = \begin{cases}
   		\text{leading term of \eqref{qx0minus} with $N = 2$ }, & x<0, \\
   		1, & x>0.
   	\end{cases}
   \end{equation}
For $N = 1$ case, we set $\eta_1 = e^{0.293i\pi}$ and $ \eta_2 = e^{0.423i\pi}  $. For the case of $ N=2$, we intentionally set the first two parameters to be identical to those in the $N=1$ case, and then specify the remaining two parameters: $\eta_3 = e^{0.567i\pi}$ and $ \eta_4 = e^{0.811i\pi}  $. By selecting appropriate initial values and discrete spectral bands, the asymptotics for large $t$ reveal that the behavior within the genus $1$ region, whether modulated or unmodulated, is entirely insensitive (to leading order) to the dynamics anticipated in the genus $2$ regions. This is perfectly consistent with the asymptotic results of this work. {Moreover, this result is also relevant to the recent research of direct scattering under periodic boundary conditions \cite{Zhu-2026} and \cite{Grava-et-al-2026}.}
   
   \begin{thm}\label{thm:main} 
      For $N = 1$, the large-space and long-time asymptotic behaviours of the genus-1 dark soliton gas $q(x,t)$ are categorized as follows:
      
      I. Large-space behaviour for the initial dark soliton gas potential $q(x,0)$.
      \begin{enumerate}[label=(\arabic*), ref=\thethm(\arabic*)]
         \item \label{thm11} 
         As $x\to+\infty$, the potential $q(x,0)$ exponentially approaches $1$ (the same as \eqref{larx}).
         
         \item \label{thm12} 
         As $x\to-\infty$, the potential $q(x,0)$ can be encoded by the 1-dimensional Riemann theta function $\Theta^{[1]}(v) = \theta_3(\pi v)$:
         \begin{equation} \label{qx0N1}
            q(x,0)=\frac{\pi\theta'_3\left(\pi \frac{1+\tau}{2};\tau\right)\theta_3\left(\frac{\pi\tau}{2},\tau\right)\theta_3\left(4\pi J(\infty)+\frac{x\Omega+\Delta}{2};2\tau\right)}
            {2\omega_1\theta_3\left(\frac{\pi}{2};2\tau\right)\theta_3\left(4\pi J(\infty)-\frac{\pi}{2}-\pi\tau;2\tau\right)\theta_3\left(\frac{x\Omega+\Delta}{2};2\tau\right)}e^{4\pi iJ(\infty)}\delta_\infty^{-2}e^{-2ig_\infty}
            +\mathcal{O}(x^{-1}), 
         \end{equation}
         where $\omega_1, \tau$, $\Omega$, $\Delta$, and $J(z)$ are defined in equations (\ref{mj0}), (\ref{tau}), (\ref{Ome}), \eqref{Ome}, and (\ref{J(z)}), respectively. $\delta_\infty$ and  $g_\infty$ are given by \eqref{tginf} and \eqref{ttdel} with $\ell=1$, respectively.
         
         Furthermore, the square modulus  can be expressed in terms of the Jacobi elliptic function as:
         \begin{equation}\label{ellcn}
            |q(x,0)|^2=\rho_1 - (\rho_1 -\rho_3 ) \dn^2 \Big( \sqrt{\rho_1 -\rho_3} \left(x + x_0\right)-K(m_1); m_1 \Big) + \mathcal O(x^{-1}).
         \end{equation}
         The parameters are
         $\rho_1=(2+\eta_1^{\rm re} - \eta_2^{\rm re})^2/4$, $\rho_3=(\eta_1^{\rm re} + \eta_2^{\rm re})^2/4$, $\rho_1-\rho_3 = (1-\eta_2^{\rm re})(1+\eta_1^{\rm re})>0$,
          and 
         $x_0=\int_{\eta_1}^{\eta_2}\frac{\ln r(z)}{\pi R_+(z)}\d z$, where $r(z)$ is defined by the interpolation \eqref{inpo} and $R(z)$ is given by \eqref{Rzti} for $\ell=N=1$ and $\alpha_N = \eta_{2N}$. Here $K(m_1)$ and $m_1$ are the complete elliptic integral of the first kind and the elliptic modulus of the elliptic Riemann surface of the radical $\mathcal R(k)  = \sqrt{(k^2-1)(k-\eta_1^{\rm re})(k-\eta_2^{\rm re})}$. 
      \end{enumerate}
      
      II. Long-time behaviour of the dark soliton gas $q(x,t)$. As $t \to \infty$, we have
      \begin{enumerate}[label=(\arabic*), ref=\thethm(\arabic*)] 
         \item \label{thm13} 
         Planar region: for $\xi > \xi_1$,  $q(x,t)$ behaves the same as \eqref{Qu1}.
         
         \item \label{thm14} 
         Modulated genus-1 dark soliton gas (dispersive shock wave) region: this holds   for $\xi_2<\xi<\xi_1$, 
                  where $\xi_1 = 2 \re \eta_1$ and $\xi_2$ is given by  
   \begin{gather}
      \xi_2=  ( \eta_1^{\rm re} +  \eta_2^{\rm re}) + \frac{2 ( \eta_1^{\rm re} - \eta_2^{\rm re})(\eta_2^{\rm re} + 1) \, K({m}_1)}{( \eta_1^{\rm re} -  \eta_2^{\rm re})K({m}_1) - ( \eta_1^{\rm re} + 1)E({m}_1)}, \label{xi2}
   \end{gather}
   where $K(m_1)$ and $E(m_1)$ are the complete elliptic integrals of the first and second kind, respectively, 
   of the elliptic curve defined by $\mathcal R(k) = \sqrt{(k^2-1)(k-\eta_2^{\rm re})(k-\eta_1^{\rm re})}$ (see Remark \ref{remell}).
 
    Then the behaviour within this region is 
         \begin{align}
            q(x,t)
            =\frac{\pi\theta_3'\left(\pi\frac{1+{\tau}}{2};{\tau}\right)\theta_3\left(\frac{{\pi\tau}}{2};{\tau}\right)\theta_3\left(4\pi{J}(\infty)
               +\frac{{t\Omega}+{\Delta}}{2};2{\tau}\right)}
            {2\omega_1\theta_3\left(\frac{\pi}{2};2{\tau}\right)\theta_3\left(4\pi {J}(\infty)-{\pi\tau}-\frac{\pi}{2};2{\tau}\right)
               \theta_3\left(\frac{{t\Omega}+{\Delta}}{2};2{\tau}\right)}e^{4\pi i{J}(\infty)}
            \delta_{1\infty}^{-2}e^{-2ig_{1\infty}}+\mathcal{O}(t^{-1}).
         \end{align} 
 Here $J$ is the Abel map (see \eqref{J(z)}),  $\tau\in i\R_+$ is the modulus of the Riemann surface $ R(z)^2 := (z-\eta_1)(z-\eta_1^{-1})(z-\alpha_1)(z-\alpha_1^{-1})$ where  the modulated branch-point  $\alpha_1 = \alpha_1(\xi)\in S^1$ evolves according to \eqref{xic} in Proposition \eqref{alphaWhitham}, which specializes to the following Whitham equation:
   \begin{equation} \label{whithamN1}
      \begin{aligned}
         \xi= 
          ( \eta^{\rm re}_1 +  \alpha^{\rm re}_1) + \frac{2 ( \eta^{\rm re}_1 -  \alpha^{\rm re}_1)(
          \alpha^{\rm re}_1 + 1) \, K({m}_1)}{( \eta^{\rm re}_1 - \alpha^{\rm re}_1)K({m}_1) - ( \eta^{\rm re}_1 + 1)E({m}_1)},
      \end{aligned}   
   \end{equation}
      with $m_1,K({m}_1),E({m}_1)$ as here above  but for the modulated elliptic Riemann surface where $\eta_2\mapsto \alpha_1$. 

          On the other hand, ${\Omega}$, ${\Delta}$, $g_{1\infty}$ and $\delta_{1\infty}$ are given by \eqref{tOmega}  \eqref{tdelta}  \eqref{tginf}, and \eqref{ttdel}, for $\ell = 1$, respectively.
         \par
         Furthermore, similar to the process presented in Appendix \ref{section55}, we have
         \begin{equation} \label{qasalpha}
            |q(x,t)|^2 ={\rho}_1 - ({\rho}_1 -{\rho}_3 ) \dn^2 \Big( \sqrt{{\rho}_1 -{\rho}_3} 
            \left(x-( \eta_1^{\rm re}+\alpha_1^{\rm re})t + {x}_0\right)-K({m}_1); {m}_1 \Big)  + \mathcal{O}(t^{-1}),
         \end{equation}
         where ${m}_1= \sqrt{\frac{(1- \eta_1^{\rm re}) (1 + \alpha_1^{\rm re})}{(1 + \eta_1^{\rm re})(1- \alpha_1^{\rm re})}}$, ${x}_0=\int_{\eta_1}^{\alpha_1}\frac{\ln r(z)}{2\pi R_{1+}(z)}\d z$, ${\rho}_1=(2+  \eta_1^{\rm re} -   \alpha_1^{\rm re})^2/4$ and ${\rho}_3=( \eta_1^{\rm re} +   \alpha_1^{\rm re})^2/4$.

         \item \label{thm15} 
         For $\xi<\xi_2$, the unmodulated genus-1 dark soliton gas behaves
         \begin{equation}
            q(x,t)=\frac{\theta_3'\left(\pi\frac{1+\tau}{2};\tau\right)
            \theta_3\left(\frac{\pi\tau}{2};\tau\right)
            \theta_3\left(4\pi J(\infty)+\frac{\hat{\Omega}+\Delta}{2};2\tau\right)}
            {2\omega_1\theta_3\left(\frac{\pi}{2};2\tau\right)\theta_3\left(4\pi J(\infty)-\pi \tau-\frac{\pi }{2};2\tau\right)
               \theta_3\left(\frac{\hat{\Omega}+\Delta}{2};2\tau\right)}e^{4\pi iJ(\infty)}
            \delta_\infty^{-2}e^{-2i\hat{g}_\infty}+\mathcal{O}(t^{-1}),
         \end{equation} 
         where parameters are the same as \eqref{qx0N1}, and similar expression analogous to \eqref{qasalpha} for the amplitude  holds with $\alpha_1=\eta_2$.
      \end{enumerate}
   \end{thm}

\begin{rmk}
\label{remell}
 The elliptic modulus and elliptic integrals appearing in \eqref{xi2} can be expressed in terms of $\theta$ functions (\cite{WhittakerWatson}, Pag 501 and Ex. 2 Pag 518),  \href{https://dlmf.nist.gov/22.2.E2}{DLMF 22.2.2}, i.e.,
   \begin{align}
   \label{mKE}
   m_1 = \frac {\theta_2^2}{\theta_3^2},  \ \ \ K(m_1) = \frac \pi 2 \theta_3^2, \ \ 
   E(m_1)= \frac \pi 6 \le(2\theta_3^2 -  \frac {\theta_2^4}{\theta_3^2} - \frac {\theta_1'''}{\theta_3^2\theta_1'}\ri) .
   \end{align}
   Here all $\theta$ constants are evaluated at $\tau^{[1]} = \frac{2\omega_3}{\omega_1}$ where the periods    $\omega_1, \omega_3$ are given by \eqref{mj0} but with the replacement $\alpha_1\mapsto \eta_2$.

\end{rmk}

   \begin{rmk}
     {
      The genus-1 results across Theorems \ref{Theorem1}-\ref{thm:main} are in complete agreement, as follows from the uniqueness of the solution of RHP \ref{rh24} or the relation $\tau^{[1]}=2\tau$. 
      In addition, the squared modulus derived from the leading-order term of \eqref{so2} in \cite{Jenkins2015} matches \eqref{qasalpha}, providing further verification of this consistency.
   }
   \end{rmk}

   \begin{figure}
      \centering 
      \subfigure[Initial potential]{\includegraphics[width=6cm]{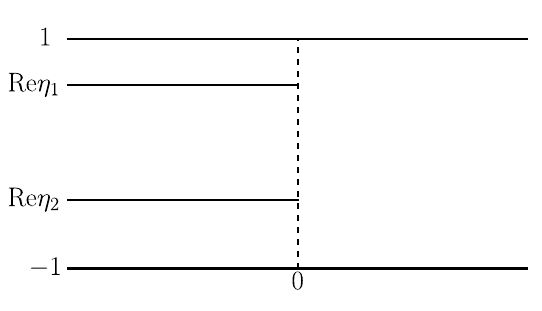}} \qquad
      \subfigure[Time evolution]{\includegraphics[width=6cm]{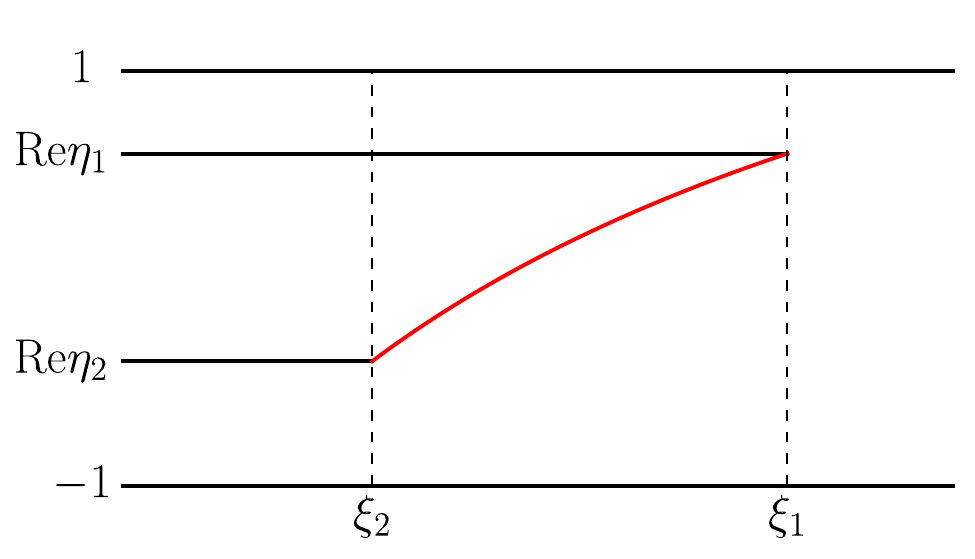}} \\
      \subfigure[Rarefaction: $\eta_1 \to 1$]{\includegraphics[width=6cm]{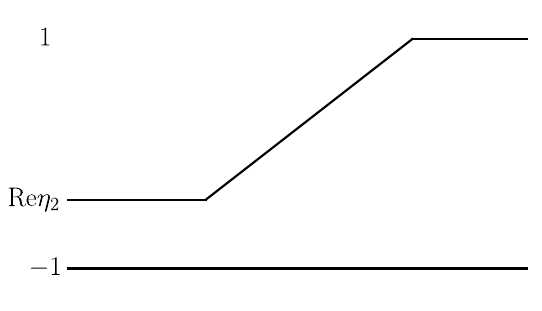}} \qquad
      \subfigure[Shock: $\eta_2 \to -1$]{\includegraphics[width=6cm]{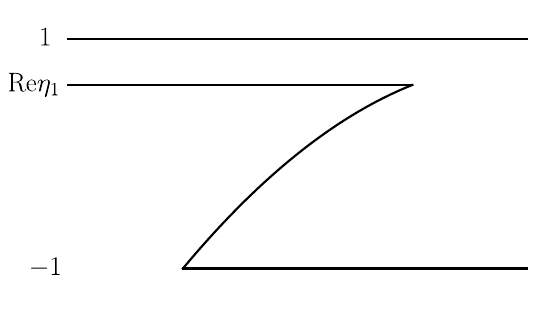}}
      \caption{{\protect\small Illustration of the initial potential and time evolution of the genus-1 dark soliton gas in terms of Riemann invariants.  The red curve shows the evolution of $\re \alpha_1$ with respect to the self-similar variable $\xi$ according to \eqref{whithamN1}.  (c) and (d) are two different types of formal degeneration of the genus-1 dark soliton gas. }}\label{fig:evolution}
      \label{open2}
   \end{figure}
   \begin{rmk}
      As shown in Figure~\ref{fig:evolution}, the initial potential of the dark soliton gas formally corresponds to step-like initial data of the Riemann invariants. The time evolution of these Riemann invariants via Whitham equations is consistent with our rigorous asymptotic analysis. The evolution of the parameter \(\alpha\) is clearly visible in the results.
   \end{rmk}
  
   \begin{rmk}
      Unlike the soliton gases of KdV \cite{CMP2021gas} and mKdV \cite{CPAM2023gas}, the realization of the dark soliton gas of the dNLS equation is more involved.\begin{enumerate}
         \item[-] Only nonzero background  admits dark soliton solutions. From a spectral perspective, the discrete eigenvalues are on the continuous spectrum $[-1,1]$. To set up the RHP, one has to introduce a uniformization variable \cite{Faddev} (a Joukowski change of variable) to ``separate'' the discrete and continuous spectra. 
         \item[-] It is inevitable that this uniformization introduces technical difficulties because the soliton-gas-RHP  must satisfy particular symmetries and additional singularity conditions, see RHP \ref{RHP2}. 
         \item[-] There are formally two distinct types of degeneration for the genus-1 dark soliton gas, each corresponding to a different class of plane wave initial data with step-like discontinuity. The first type $\eta_1 \to 1$ corresponds to rarefaction initial data, namely, data whose evolution leads to a rarefaction wave. The second type  $\eta_2 \to -1$ corresponds to shock initial data, which evolves into a dispersive shock wave. This can be illustrated clearly by the Riemann invariants, see Figure~\ref{open2}.
      \end{enumerate}
   \end{rmk}
   \begin{rmk}
      During this process, we identify an intriguing band-gap conversion mechanism, see Figures \ref{bandgap1} and \ref{N-DNLS6}. The band on the unit circle, formed by the accumulation of discrete spectrum in the uniformization variable, is transformed back into the gap on $[-1,1]$ under the Joukowski transformation, recovering the canonical finite-gap solution. This mechanism is particularly salient in arbitrary-genus dark soliton gases as shown in Figure \ref{N-DNLS6}.  We find that the accumulation of ``embedded'' eigenvalues on the continuous spectrum induces gaps (instead of bands), while the complement being bands. This distinction is important for determining the genus of the dark soliton gas, which is consistent with the genus of the correponding finite-gap solution. This phenomenon should be valid for soliton gases generated from solitons on nonzero background in other integrable systems, such as breather gases of the focusing NLS equation \cite{El-Tovbis-2020, breather}. A detailed investigation of this regime will be the subject of our future work. 
   \end{rmk}
   \begin{figure}
      \centering
      \includegraphics[width=14cm]{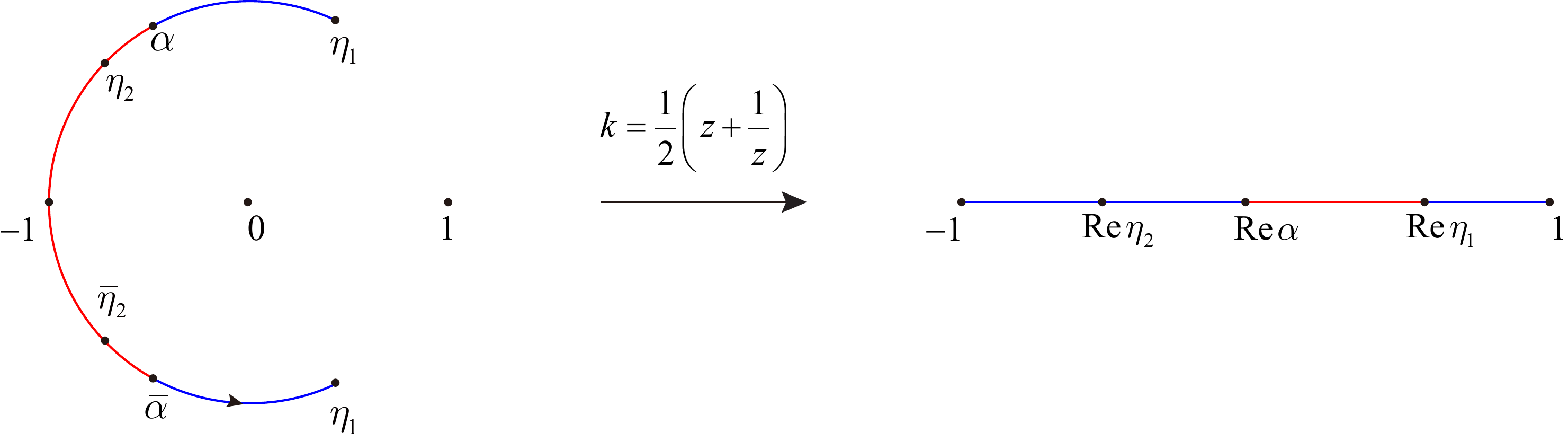}
      \caption{{\protect\small Band-gap conversion. The red curves are gaps, and the blue curves are bands of the genus-one solution. After the  transformation, the roles of the band and the gap are interchanged.}}
      \label{bandgap1}
   \end{figure}
   \par

The rest parts of this paper are organized as follows. In Section \ref{sect1}, we set up the dark soliton gas RHP by considering the limit of the pure dark $\mathcal{N}$-soliton as $\mathcal{N} \to +\infty$. In Section \ref{section33}, we study the large-space behaviour of the initial potential $q(x,0)$ to prove Theorem \ref{Theorem1}, and in Section \ref{section44},  we investigate its long-time evolution to prove Theorem \ref{Theorem2}.  In Appendix \ref{section55} we discuss in detail the genus-1 dark soliton gas to prove Theorem \ref{thm:main}. 
We prove the solvability of the Whitham equations in Appendix \ref{App2}, and give the error estimate in Appendix \ref{sectionerrors}.
   
We conclude this introduction by observing that we do not investigate the transition regions from a modulated to an unmodulated region: an easy speculation suggests that the relevant analysis would require using a special local parametrix that should parallel the parametrix used in \cite{soft_hard}, which is related to the second Painlev\'e\ transcendent.

\section{Defocusing NLS soliton gas derived by the limit of $\mathcal{N}$-soliton as $\mathcal{N} \to +\infty$.}\label{sect1}
   The dNLS equation \eqref{dNLS} is obtained from the compatibility condition of the following two linear spectral problems for matrix-valued function $\Psi=\Psi(z; x,t)$ of the form
   \begin{equation*}
      \Psi_x= \mathcal{L}\Psi, \quad
      i\Psi_t= \mathcal{B}\Psi,
   \end{equation*}
   where the $2\times2$ matrices $\mathcal{L}=\mathcal{L}(z; x,t)$ and $\mathcal{B}=\mathcal{B}(z; x,t)$ are given by
   \begin{equation*}
      \mathcal{L}= i\sigma_3\left(Q-\frac{z+z^{-1}}{2}\right), \quad
      \mathcal{B}= -i(z+z^{-1})\mathcal{L}-(Q^2-\mathbb{I})\sigma_3+iQ_x,
   \end{equation*}
   where
   \begin{equation*}
      Q=Q(x,t)=\begin{pmatrix}
         0 & \overline{q(x,t)} \\
         q(x,t) & 0
      \end{pmatrix}, \quad 
      \mathbb{I}=\begin{pmatrix}
         1 & 0\\
         0 & 1
      \end{pmatrix},
   \end{equation*} 
   and the three Pauli matrices are given by
   \begin{equation*}
      \sigma_1=\begin{pmatrix}
         0 & 1\\
         1 & 0
      \end{pmatrix},
      \quad
      \sigma_2=\begin{pmatrix}
         0 & -i\\
         i & 0
      \end{pmatrix},
      \quad 
      \sigma_3=\begin{pmatrix}
         1 & 0\\
         0 & -1
      \end{pmatrix}.
   \end{equation*}
   Note that the spectral parameter here is chosen to be the Joukowski map 
   \begin{equation} 
      \label{Youkowski}
      k(z) =  \frac {z+z^{-1}}2,
   \end{equation}
   which is the parametrization we foreshadowed in the introduction.
   
   Recall there are $4N$ points $\{\eta_{2j-1}^{\pm1}, \eta_{2j}^{\pm1} \}_{j=1}^{N}$ on the unit circle. Define the disjoint intervals on the unit circle as shown in Figure \ref{N-DNLS3}:
      \begin{equation}\label{gammacontor}
         \begin{aligned}
            &\gamma_j := (\eta_{2j-1}, \eta_{2j})=\{{\rm e}^{i\phi}:\arg\eta_{2j-1}<\phi<\arg\eta_{2j}\}, \quad j=1, \cdots, N,\\
            &\bar{\gamma}_j :=(\eta_{2j}^{-1}, \eta_{2j-1}^{-1})=\{{\rm e}^{i\phi}:\arg\eta_{2j}^{-1}<\phi<\arg\eta_{2j-1}^{-1}\}, \quad j=1, \cdots, N.
         \end{aligned}
      \end{equation}
      There will be $n_j$ simple eigenvalues and their conjugates on each $\gamma_j$ and $\bar{\gamma}_j$.
      The pure $\mathcal{N} =\sum_{j=1}^{N} n_j$ soliton solution of the dNLS equation can be obtained by solving the following RHP, see  \cite{Jenkins2016}.
   \begin{RHP}\label{RH1} Consider a $2\times 2$ matrix-valued function $M^{\rm sol}(z;x,t)$ satisfying
      \begin{enumerate}
         \item Analyticity: $ M^{\rm sol}(z)$ is meromorphic for $z\in\mathbb{C}$ with simple poles $\{z_{l,j}\}_{l=1}^{n_j}\cup\{z_{l,j}^{-1}\}_{l=1}^{n_j}$, where $\{z_{l,j}\}_{l=1}^{n_j}\subset \gamma_j$, $\{z_{l, j}^{-1}\}_{l=1}^{n_j}\subset \bar{\gamma}_j $, $j=1, \cdots, N$. Here $\{z_{l,j}\}_{l=1}^{n_j}$ ($\{z_{l,j}^{-1}\}_{l=1}^{n_j}$) are fixed counterclockwise (clockwise).
         \item Asymptotic behaviours: $ M^{\rm sol}(z)=\mathbb{I}+\mathcal{O}(z^{-1})$ as $z\to\infty$ and $ M^{\rm sol}(z)=\frac{\sigma_1}{z}+\mathcal{O}(1)$ as $z\to0$.
         \item Symmetry conditions: $ M^{\rm sol}(z)=\sigma_1\overline{M^{\rm sol}(\bar{z})}\sigma_1=z^{-1}M^{\rm sol}(z^{-1})\sigma_1$.  
         \item Residue conditions: 
         \begin{align}
            \Res_{z=z_{l,j}} M^{\rm sol}(z)&=\lim_{z\to z_{l,j}} M^{\rm sol}(z)\begin{pmatrix}
               0 & 0 \\
               c_{l,j}{\rm e}^{i\Phi(z_{l,j};x,t)} & 0
            \end{pmatrix}, \quad j=1, \cdots, N,  \\
            \Res_{z=z_{l,j}^{-1}} M^{\rm sol}(z)&=\lim_{z\to z_{l,j}^{-1}} M^{\rm sol}(z)\begin{pmatrix}
               0 & \overline{c_{l,j}}{\rm e}^{-i\Phi(z_{l,j}^{-1};x,t)} \\
               0 & 0
            \end{pmatrix}, \quad j=1, \cdots, N,
         \end{align}
         where $\Phi(z;x,t)=x(z-z^{-1})-t(z^2-z^{-2})$. \hfill $\triangle$
      \end{enumerate}
   \end{RHP}
    The discrete eigenvalues $\{z_{l,j}\}_{l=1}^{n_j}$ ($\{z_{l,j}^{-1}\}_{l=1}^{n_j}$) are chosen constrained within the arc-shaped interval (see Figure \ref{N-DNLS3}) $\gamma_j$ ($\bar{\gamma}_j$), ensuring they are uniformly distributed, and we choose norming constants $c_{l,j}$ by fixing a {\it locally analytic} function $r:\bigcup_{j=1}^N(\gamma_j\cup\bar{\gamma}_j)\to \mathbb{C}$ (satisfying $\overline{r(\bar{z}^{-1})}=r(z)$) in terms of the following interpolation-like formula:
   \begin{equation}\label{inpo}
      c_{l,j}=\frac{r(z_{l,j})(z_{l,j}-z_{l,j-1})}{2\pi}, \quad j=1,\cdots,N, \quad l=1, \cdots, n_j.
   \end{equation}
   The inverse scattering theory \cite{Faddev}  shows that 
   \begin{equation}
      M^{\rm sol}(z;x,t)=\mathbb{I}+\frac{1}{z} \begin{pmatrix} - i \int_{x}^{\infty} |q(y,t)|^2 - 1~ dy & \overline{q(x,t)}  \\  q(x,t) &  i \int_{x}^{\infty} |q(y,t)|^2 - 1~ dy \end{pmatrix} + \mathcal{O}({z}^{-2}),~~ \mathrm{as} ~~ z \to \infty,
   \end{equation}
   \begin{equation}
      M^{\rm sol}(z;x,t)=\frac{\sigma_1}{z} +\begin{pmatrix} \overline{q(x,t)} & - i \int_{x}^{\infty} |q(y,t)|^2 - 1~ dy  \\  i \int_{x}^{\infty} |q(y,t)|^2 - 1~ dy & q(x,t)  \end{pmatrix} + \mathcal{O}(z), \quad \mathrm{as}  ~~~~   z \to 0,
   \end{equation}
   whereby we arrive at the reconstruction formulae,
   \begin{equation}\label{rec1}
      \begin{aligned}
         &q^{\rm sol}(x,t)=\lim_{z\to\infty}zM^{\rm sol}_{21}(z;x,t), \\
         &|q^{\rm sol}(x,t)|^2=1-i \partial_x \left(\lim_{z\to\infty} z  M^{\rm sol}_{11}(z;x,t)\right).
      \end{aligned}
   \end{equation}

   \par 
   Let $\Gamma_j$ ($\bar{\Gamma}_j$) be a counterclockwise (clockwise) closed smooth contours encompassing $\{z_{l,j}\}_{l=1}^{n_j}$ ($\{z_{l,j}^{-1}\}_{l=1}^{n_j}$), $j=1,\cdots, N$, see Figure \ref{N-DNLS3}. 
   To eliminate the residue conditions, introduce the transformation:
   \begin{equation}
      D(z)=M^{\rm sol}(z)\begin{cases}
         \begin{pmatrix}
            1 & 0 \\
            -\sum_{l=1}^{n_j}\frac{c_{l,j}{\rm e}^{i\Phi(z;x,t)}}{z-z_{l,j}}  &  1
         \end{pmatrix}, & z {\rm~within~} \Gamma_j, \quad j=1, \cdots, N,
         \\
         \begin{pmatrix}
            1 & -\sum_{l=1}^{n_j}\frac{\overline{c_{l,j}}{\rm e}^{-i\Phi(z;x,t)}}{z-z_{l,j}^{-1}} \\
            0  &  1
         \end{pmatrix}, & z {\rm~within~} \bar{\Gamma}_j, \quad j=1, \cdots, N,
         \\
         \mathbb{I}, & z\in {\rm elsewhere. }
      \end{cases}
   \end{equation}
   It follows that the residue conditions on $\{z_{l,j}\}_{l=1}^{n_j}\cup\{z_{l,j}^{-1}\}_{l=1}^{n_j}$ are reformulated as the following jump conditions:
   \begin{equation}\label{2.05}
      D_+(z)=D_-(z)\begin{cases}
         \begin{pmatrix}
            1 & 0 \\
            -\sum_{l=1}^{n_j}\frac{c_{l,j}{\rm e}^{i\Phi(z;x,t)}}{z-z_{l,j}} & 1
         \end{pmatrix}, & z\in \Gamma_j, \quad j=1, \cdots, N,
         \\
         \begin{pmatrix}
            1 & \sum_{l=1}^{n_j}\frac{\overline{c_{l,j}}{\rm e}^{-i\Phi(z;x,t)}}{z-z_{l,j}^{-1}} \\
            0  &  1
         \end{pmatrix}, & z\in \bar{\Gamma}_j, \quad j=1, \cdots, N.
      \end{cases}
   \end{equation}
   
   \begin{figure}
      \centering
      \begin{tikzpicture}
      
         \node[anchor=south west, inner sep=0] (image) at (0,0) {\includegraphics[width=6cm]{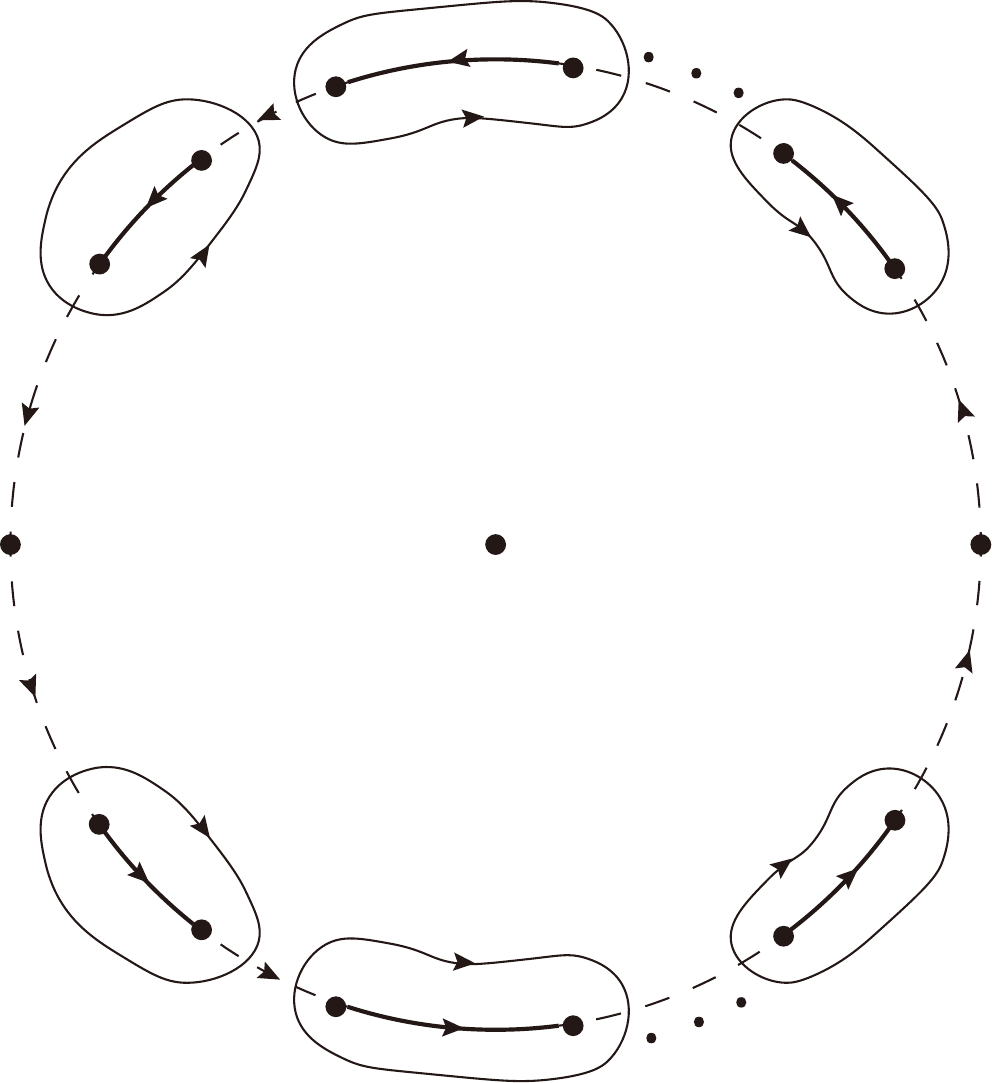}};

         \newcommand{\imgwidth}{11}
         \newcommand{\imgheight}{7}

         \draw[->, >=stealth, black] (0.52*\imgwidth, 0.81*\imgheight) -- (0.48*\imgwidth, 0.77*\imgheight);
         \draw[->, >=stealth, black] (0.26*\imgwidth, 0.96*\imgheight) -- (0.26*\imgwidth, 0.9*\imgheight);
         \draw[->, >=stealth, black] (0.01*\imgwidth, 0.82*\imgheight) -- (0.07*\imgwidth, 0.76*\imgheight);
         \draw[->, >=stealth, black] (0.5*\imgwidth, 0.1*\imgheight) -- (0.47*\imgwidth, 0.16*\imgheight);
         \draw[->, >=stealth, black] (0.27*\imgwidth, -0.02*\imgheight) -- (0.25*\imgwidth, 0.03*\imgheight);
         \draw[->, >=stealth, black] (0.01*\imgwidth, 0.12*\imgheight) -- (0.07*\imgwidth, 0.18*\imgheight);
         \draw[->, >=stealth, black] (0.1*\imgwidth, 0.94*\imgheight) -- (0.145*\imgwidth, 0.84*\imgheight);
         \draw[->, >=stealth, black] (0.1*\imgwidth, -0.01*\imgheight) -- (0.15*\imgwidth, 0.08*\imgheight);
         \node[black, font=\small, align=center, text width=4cm] at (0.43*\imgwidth, 0.7*\imgheight) {$\Gamma_1$};
         \node[black, font=\small, align=center, text width=4cm] at (0.27*\imgwidth, 0.79*\imgheight) {$\Gamma_{N-1}$};
         \node[black, font=\small, align=center, text width=4cm] at (0.15*\imgwidth, 0.7*\imgheight) {$\Gamma_{N}$};
         \node[black, font=\small, align=center, text width=4cm] at (0.14*\imgwidth, 0.26*\imgheight) {$\bar{\Gamma}_{N}$};
         \node[black, font=\small, align=center, text width=4cm] at (0.27*\imgwidth, 0.16*\imgheight) {$\bar{\Gamma}_{N-1}$};
         \node[black, font=\small, align=center, text width=4cm] at (0.43*\imgwidth, 0.24*\imgheight) {$\bar{\Gamma}_1$};
         \node[black, font=\small, align=center, text width=4cm] at (-0.03*\imgwidth, 0.42*\imgheight) {$-1$};
         \node[black, font=\small, align=center, text width=4cm] at (0.275*\imgwidth, 0.42*\imgheight) {$0$};
         \node[black, font=\small, align=center, text width=4cm] at (0.557*\imgwidth, 0.42*\imgheight) {$1$};
         \node[black, font=\small, align=center, text width=4cm] at (0.51*\imgwidth, 0.73*\imgheight) {$\eta_{1}$};
         \node[black, font=\small, align=center, text width=4cm] at (0.44*\imgwidth, 0.83*\imgheight) {$\eta_{2}$};
         \node[black, font=\small, align=center, text width=4cm] at (0.34*\imgwidth, 0.91*\imgheight) {$\eta_{2N-3}$};
         \node[black, font=\small, align=center, text width=4cm] at (0.18*\imgwidth, 0.9*\imgheight) {$\eta_{2N-2}$};
         \node[black, font=\small, align=center, text width=4cm] at (0.09*\imgwidth, 0.84*\imgheight) {$\eta_{2N-1}$};
         \node[black, font=\small, align=center, text width=4cm] at (0.03*\imgwidth, 0.74*\imgheight) {$\eta_{2N}$};
         \node[black, font=\small, align=center, text width=4cm] at (0.52*\imgwidth, 0.2*\imgheight) {$\eta_{1}^{-1}$};
         \node[black, font=\small, align=center, text width=4cm] at (0.44*\imgwidth, 0.09*\imgheight) {$\eta_{2}^{-1}$};
         \node[black, font=\small, align=center, text width=4cm] at (0.34*\imgwidth, 0.01*\imgheight) {$\eta_{2N-3}^{-1}$};
         \node[black, font=\small, align=center, text width=4cm] at (0.18*\imgwidth, 0.01*\imgheight) {$\eta_{2N-2}^{-1}$};
         \node[black, font=\small, align=center, text width=4cm] at (0.1*\imgwidth, 0.09*\imgheight) {$\eta_{2N-1}^{-1}$};
         \node[black, font=\small, align=center, text width=4cm] at (0.025*\imgwidth, 0.2*\imgheight) {$\eta_{2N}^{-1}$};
         \node[black, font=\small, align=center, text width=4cm] at (0.53*\imgwidth, 0.82*\imgheight) {$\gamma_1$};
         \node[black, font=\small, align=center, text width=4cm] at (0.26*\imgwidth, 0.98*\imgheight) {$\gamma_{N-1}$};
         \node[black, font=\small, align=center, text width=4cm] at (0.01*\imgwidth, 0.84*\imgheight) {$\gamma_{N}$};
         \node[black, font=\small, align=center, text width=4cm] at (0.52*\imgwidth, 0.09*\imgheight) {$\bar{\gamma}_1$};
         \node[black, font=\small, align=center, text width=4cm] at (0.27*\imgwidth, -0.04*\imgheight) {$\bar{\gamma}_{N-1}$};
         \node[black, font=\small, align=center, text width=4cm] at (0.01*\imgwidth, 0.1*\imgheight) {$\bar{\gamma}_{N}$};
         \node[black, font=\small, align=center, text width=4cm] at (0.1*\imgwidth, 0.97*\imgheight) {$\gamma^{\rm gap}_{N-1}$};
         \node[black, font=\small, align=center, text width=4cm] at (0.1*\imgwidth, -0.04*\imgheight) {$\bar{\gamma}^{\rm gap}_{N-1}$};
         \node[black, font=\small, align=center, text width=4cm] at (-0.02*\imgwidth, 0.57*\imgheight) {$\gamma^{\rm gap}_N$};
         \node[black, font=\small, align=center, text width=4cm] at (0.57*\imgwidth, 0.57*\imgheight) {$\gamma^{\rm gap}_0$};
      \end{tikzpicture}
      \caption{\small Schematic diagram of solitons concentrated within $2N$ curve intervals.}
      \label{N-DNLS3}
   \end{figure}
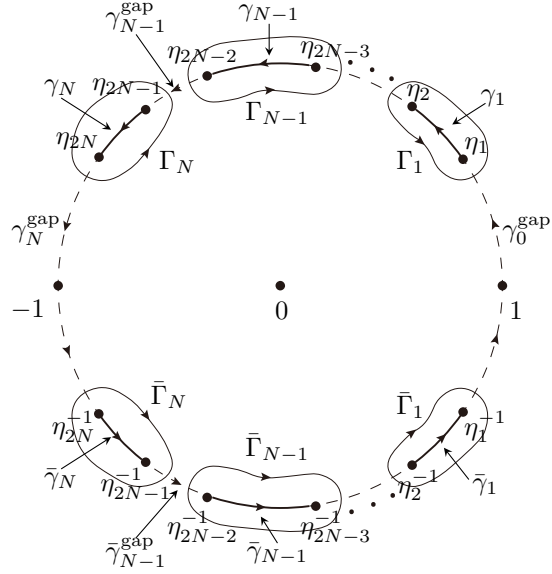
   \par 
   Letting each $n_j\to+\infty$, we have
   \begin{equation}
      \begin{aligned}
         \sum_{l=1}^{n_j}\frac{c_{l,j}}{z-z_{l,j}}&= \sum_{l=1}^{n_j}\frac{ir(z_{l,j})}{z-z_{l,j}}\frac{z_{l,j}-z_{l,j-1}}{2\pi i}
         \to\int_{\gamma_j}\frac{ir(s)}{z-s}\frac{ds}{2\pi i},
         \\
         \sum_{l=1}^{n_j}\frac{\overline{c_{l,j}}}{z-z_{l,j}^{-1}}&= \sum_{l=1}^{n_j}\frac{i\overline{r(z_{l,j})}}{z-z_{l,j}^{-1}}\frac{z_{l,j}^{-1}-z_{l,j-1}^{-1}}{2\pi i}
         \to-\int_{\bar{\gamma}_j}\frac{i\overline{r(\bar{s})}}{z-s}\frac{ds}{2\pi i}.
      \end{aligned}
   \end{equation}
   Therefore, as $n_j\to+\infty$, from (\ref{2.05}), we obtain the following jump discontinuities for $D^{(\infty)}$:
   \begin{equation}
      D^{(\infty)}_+(z)=D^{(\infty)}_-(z)\begin{cases}
         \begin{pmatrix}
            1 & 0 \\
            {\rm e}^{i\Phi(z;x,t)}\int_{\gamma_j}\frac{ir(s)}{z-s}\frac{ds}{2\pi i} & 1
         \end{pmatrix}, & z\in\Gamma_j, \quad j=1, \cdots, N,
         \\
         \begin{pmatrix}
            1 & {\rm e}^{-i\Phi(z;x,t)}\int_{\bar{\gamma}_j}\frac{i\overline{r(\bar{s})}}{z-s}\frac{ds}{2\pi i} \\
            0  &  1
         \end{pmatrix}, & z\in\bar{\Gamma}_j, \quad j=1, \cdots, N.
      \end{cases}
   \end{equation}
   Then we define
   \begin{equation}
      M(z)=D^{(\infty)}(z)\begin{cases}
         \begin{pmatrix}
            1 & 0 \\
            -{\rm e}^{i\Phi(z;x,t)}\int_{\gamma_j}\frac{ir(s)}{z-s}\frac{ds}{2\pi i} & 1
         \end{pmatrix}, & z {\rm~within~} \Gamma_j,
         \\
         \begin{pmatrix}
            1 & {\rm e}^{-i\Phi(z;x,t)}\int_{\bar{\gamma}_j}\frac{i\overline{r(\bar{s})}}{z-s}\frac{ds}{2\pi i} \\
            0  &  1
         \end{pmatrix}, & z {\rm~within~} \bar{\Gamma}_j,
         \\
         \mathbb{I}, & z\in {\rm elsewhere},
      \end{cases}
   \end{equation}
   and utilize the Sokhotski-Plemelj formula to obtain the following RHP:
   \begin{RHP} (Soliton gas RHP)\label{RHP2}\ Consider a $2\times 2$ matrix-valued function $M(z;x,t)$ satisfying
      \begin{enumerate}
         \item Analyticity: $ M(z)$ is analytic for $z\in\mathbb{C}\setminus(\{0\} \cup \bigcup_{j=1}^{N}\gamma_j\cup\bar{\gamma}_j)$.
         \item Asymptotic behaviours: $ M(z)=\mathbb{I}+\mathcal{O}(z^{-1})$ as $z\to\infty$ and $ M(z)=\frac{\sigma_1}{z}+\mathcal{O}(1)$ as $z\to0$.
         \item Symmetry conditions: $ M(z)=\sigma_1\overline{M(\bar{z})}\sigma_1=z^{-1}M(z^{-1})\sigma_1$.
         \item Jump conditions: $M_+(z;x,t)=M_-(z;x,t) J_M(z;x,t)$, where $J_M(z;x,t)$ is given by
         \begin{equation}\label{jump2}
            J_M(z;x,t)=\begin{cases}
               \begin{pmatrix}
                  1 & 0 \\
                  -ir(z){\rm e}^{i\Phi(z;x,t)} & 1
               \end{pmatrix}, & z\in \cup_{j=1}^N \gamma_j,
               \\
               \begin{pmatrix}
                  1 & i\overline{r(\bar{z})}{\rm e}^{-i\Phi(z;x,t)} \\
                  0 & 1
               \end{pmatrix}, & z\in \cup_{j=1}^N \bar{\gamma}_j,
            \end{cases}
         \end{equation}
         where \begin{equation}
            \label{defPhi}
            \Phi(z;x,t)=x(z-z^{-1})-t(z^2-z^{-2}). 
         \end{equation}\hfill$\triangle$
      \end{enumerate}
   \end{RHP}
   Due to the symmetry conditions, one can prove the existence of the solution $M(z;x,t)$ by Zhou's vanishing lemma \cite{Zhou1989, Zhou2002}.  We refer the reader to \cite{breather} for an analogous proof.
\begin{pro} 
\label{propq0} 
There exists a unique solution $M(z;x,t)$ to RHP \ref{RHP2}.
The dark soliton gas potential $q(x,t)$ for the dNLS equation can reconstructed by the formula
\begin{equation} \label{rec}
   q(x,t)=\lim_{z\to\infty}zM_{21}(z;x,t),
\end{equation}
where $M_{21}(z;x,t)$ is the (2, 1)-entry of the matrix $M(z;x,t)$.
   \end{pro}
 \section{Long-time asymptotic behaviour} \label{section44}
 It is more effective to analyze first the long-time behaviour because several of the constituent components of the construction appear with minor modifications in the study of large negative $x$ and thus we can optimize the presentation.  
  The analysis of the time evolution comes from the investigation of the RHP \ref{RHP2} with $t\neq 0$.
 To begin with, we define the self-similar variable $\xi=\frac{x}{t}$, along with critical parameters $\xi_1 = 2  \eta_1^{\rm re} > \xi_2 > \cdots > \xi_{2N}$. These will be  defined recursively by the combination of Def. \ref{evenxis} and Def. \ref{defxiodd} below.
 The time evolution of the dark soliton gas can be partitioned into three distinct regions in the upper $x$-$t$ plane:

   \begin{enumerate}
      \item Planar region: For $\xi>\xi_1$, the potential converges to the background $1$, exponentially fast as $t\to+\infty$.
      \item Modulated genus-$\ell$ gas region: For $\xi_{2 \ell}< \xi < \xi_{2 \ell - 1}$, $\ell = 1, \cdots, N$, the dark soliton gas is  characterized by the modulated genus-$\ell$ finite-gap solution with modulating parameter $\alpha_{\ell}(\xi)$, which satisfies the Whitham equation \eqref{xic}. See Section \ref{mod}.
      \item Unmodulated genus-$\ell$ gas region: For $\xi_{2 \ell + 1}< \xi < \xi_{2 \ell}$, $\ell = 1, \cdots, N$, where $\xi_{2N+1} := -\infty$, the dark soliton gas is  characterized by the unmodulated genus-$\ell$ finite-gap solution of the dNLS equation. See Section \ref{sunun}.
   \end{enumerate}
  
 Intuitively, for the simplest case $N = 1$, this corresponds to the different choices of $g$-function and the subsequent sign structure of the phase functions  (see Figure \ref{signc}). 
   \begin{figure}
      \centering  
      \subfigbottomskip=10pt 
      \subfigcapskip=-5pt 
      \subfigure[$\xi=-3.8<\xi_2$ for $\re (ti\hat{\varphi}(z)) $]{
         \includegraphics[width=5.2cm]{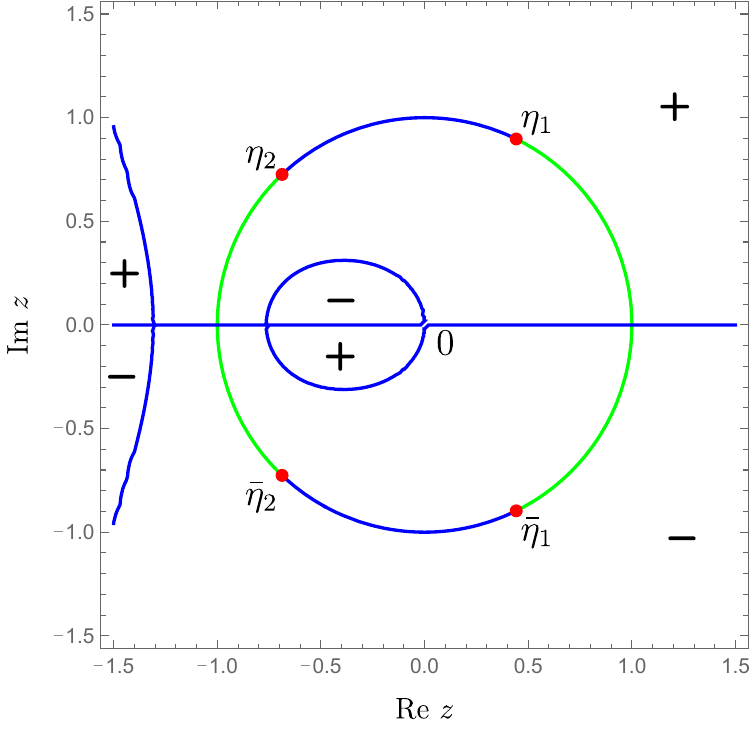}\label{eep1}}
      \subfigure[$\xi=-2.959\approx\xi_2$ for $\re (ti\hat{\varphi}(z))$]{
         \includegraphics[width=5.2cm]{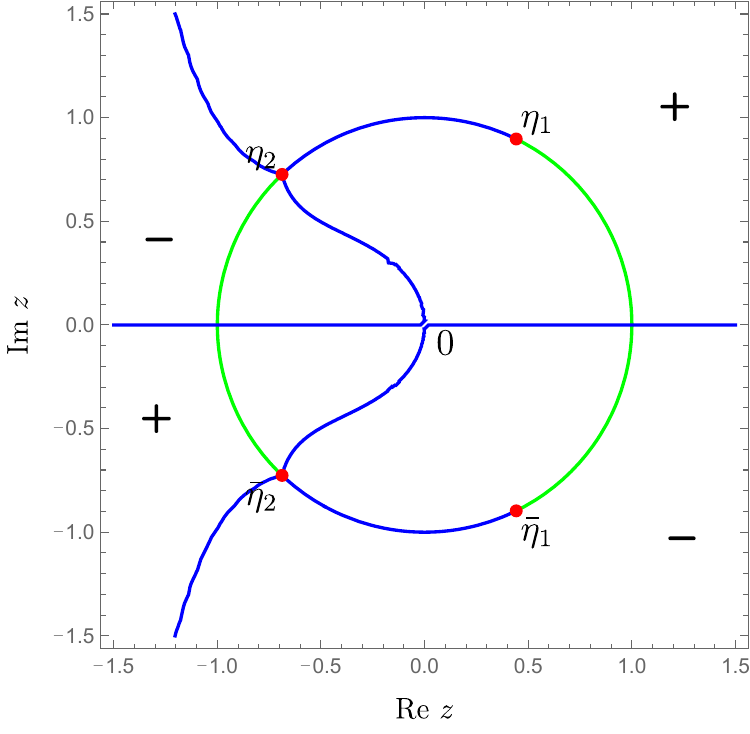}\label{eep2}}
      \subfigure[$\xi_2<\xi=-1<\xi_1$ for $\re (ti\varphi_{\ell}(z))$]{
         \includegraphics[width=5.2cm]{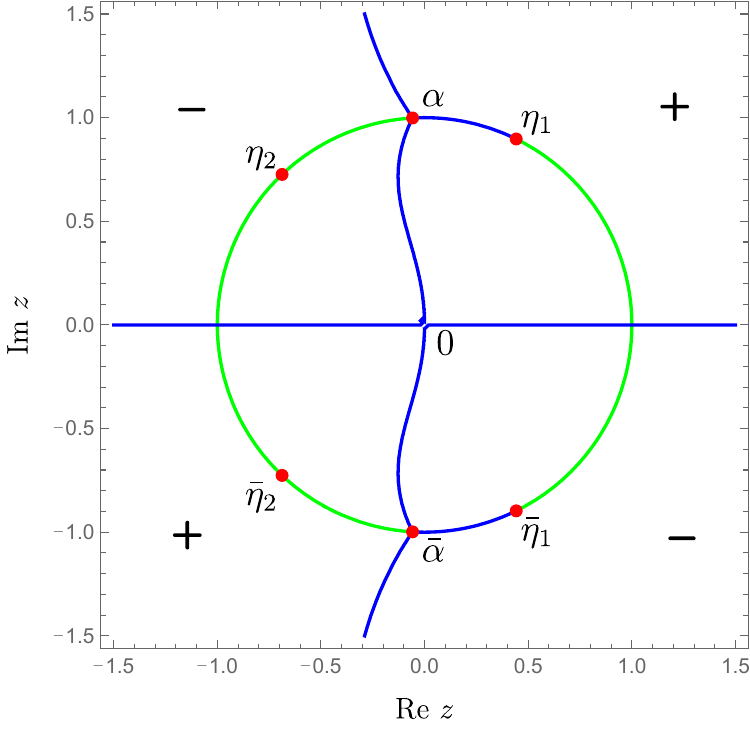}\label{eep3}}
      \\
      \subfigure[$\xi=0.8839\approx\xi_1$ for $\re (ti\varphi_{\ell}(z))$]{
         \includegraphics[width=5.2cm]{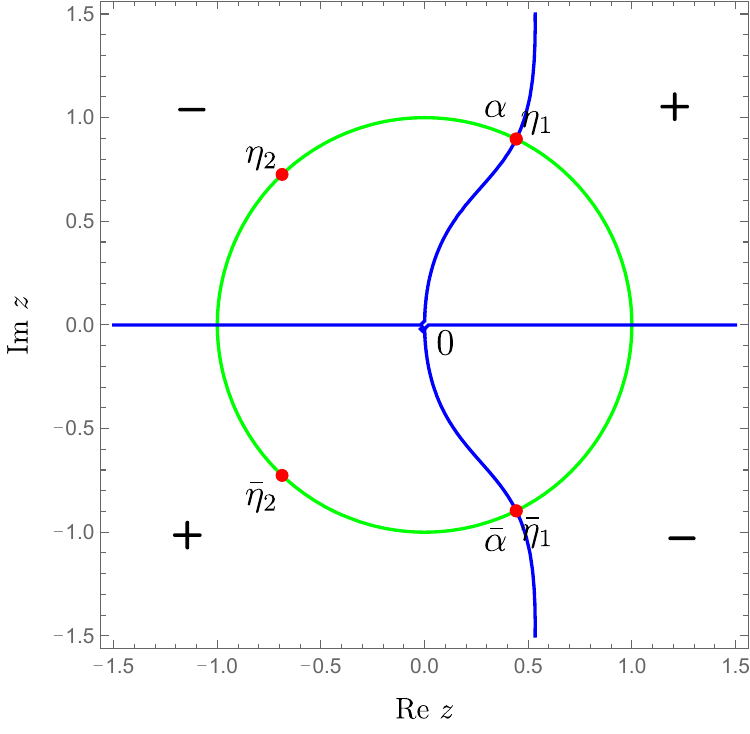}\label{eep4}}
      \subfigure[$\xi=2.22>\xi_1$ for $\re (ti\Phi(z))$]{
         \includegraphics[width=5.2cm]{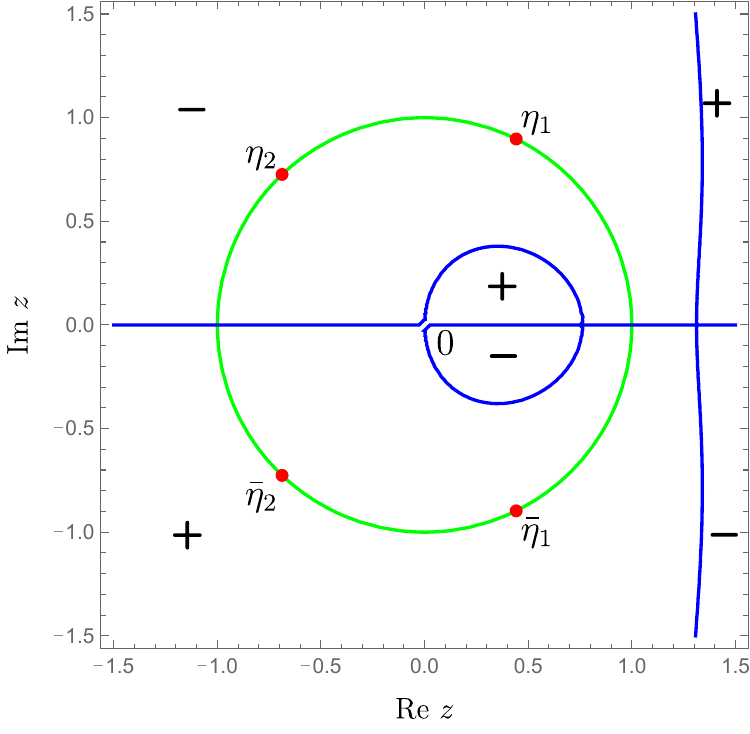}\label{eep5}}
      \caption{Example with genus $1$: shown are snapshots of the zero level sets (blue) for $\re[ti\varphi_{\ell}(z;\xi,1)]$ (\ref{tildg}), $\re[ti\hat{\varphi}(z;\xi,1)]$ (\ref{hatg}), and $\re[ti\Phi(z;\xi,1)]$ across different values of $\xi$. The blue curves partition the complex plane into regions, with the labels indicating the sign of the corresponding function in each domain. Green curve: the unit circle. Plot data: $\eta_1={\rm e}^{0.3543 i \pi}$, $\eta_2= {\rm e}^{0.7414 i \pi}$, $t=10$, $\xi_1\approx 0.8839$ and $\xi_2\approx -2.959$. }
      \label{signc}
   \end{figure}
   
   \subsection{Planar region}
   For $\xi  > \xi_1=2\re{\eta_1}$, we have the following estimate for the jump matrices $J_M(z;x,t)$ in (\ref{jump2})
   \begin{equation} \label{estpl}
      \norm{J_M(z;x,t)-\mathbb{I}}= \mathcal{O}\left({\rm e}^{-2 t  \min \left\{\im{\eta_1}, \im{\eta_{2N}} \right\}  (\xi-2\re{\eta_1})}\right).
   \end{equation}
   By a standard small-norm argument, we conclude that 
   \begin{equation}
      M(z;x,t)= \mathbb{I} + \frac{\sigma_1}{z} + \mathcal{O}\left({\rm e}^{-2 t  \min \left\{\im{\eta_1}, \im{\eta_{2N}} \right\}  (\xi-2\re{\eta_1})}\right), \quad \mathrm{as} ~~ t \to +\infty~~\mathrm{within} ~~ \xi  > \xi_1.
   \end{equation}
   Using \eqref{rec}, we obtain that the potential $q(x,t)$ exhibits the following long-time asymptotic behaviour:
   \begin{equation}
      q(x,t)=1+\mathcal{O}({\rm e}^{-ct}),\quad t\to+\infty,
   \end{equation}
   where $c\in\mathbb{R}^+$ is a fixed constant. Namely, the  nonzero plane wave background remains unperturbed. This completes the proof of Theorem \ref{thm3}.
   
\subsection{Modulated genus-$\ell$ gas region, $\ell =1 ,\dots, N$}\label{mod}
%

When $\xi<\xi_1 = 2 \eta_1^{\rm re}$, the estimate \eqref{estpl} does not provide an effective bound on the  jump matrices and we need to implement a  $g$-function approach. 

\subsubsection{Formulation of the $g$--function} 
The $g$--function will be constructed explicitly below, but here we want to formulate it in a spirit which is close to \cite{BertoTovbisCMP2026}. 
Let us denote by $\mathcal K \subset \mathbb H $  the union of the bands of accumulation of the solitons in the upper half plane, namely, the union of arcs of unit circle $[\eta_{2j-1}, \eta_{2j}]$, $j=1,\dots N$. 
Denote by $G(z,w)$ the Green's function of the upper half plane, $G(z,w) = \ln \left|\frac {z-\overline w}{z-w}\right|$. Let the {\it external field} be defined as 
\be
\label{defV}
Q(z;\xi):=-Q_2(z)  +  \xi Q_1(z), \ \ \ Q_j(z):= \im \le(z^j - \frac 1 {z^j}\ri),\  j=1,2.
\ee
Alternatively we can observe that $Q  =  \im \Phi(z;\xi,1)$ with $\Phi$ as in \eqref{defPhi}.
Consider then the set of all non-negative measures of arbitrary total mass and  with support contained in $\mathcal K$. On this set we define  the energy-functional 
\be
\label{defcalE}
\mathcal E[\d\mu] :=  \int_{\mathcal K\times \mathcal K} G(z,w) \d\mu(z)\d\mu(w) + \int_{\mathcal K} Q(z;\xi)\d\mu(z).
\ee
Let $\d\mu_\xi$ be the measure minimizing $\mathcal E$: the variational equations then imply that 
\be
\le\{
\begin{array}{cc}
\displaystyle 
\int_{\mathcal K} \ln \le|\frac {z-\overline w}{z-w} \ri| \d\mu_\xi(w) +\frac 1 2 Q(z;\xi) \equiv 0, & z\in {\rm supp}(\d\mu_\xi)\subset \mathcal K,\\[20pt]
\displaystyle 
\int_{\mathcal K} \ln \le|\frac {z-\overline w}{z-w} \ri| \d\mu_\xi(w) +\frac 1 2  Q(z;\xi) \geq  0, &  z\in  \mathcal K\setminus {\rm supp}(\d\mu_\xi).
\end{array}
\ri. 
\ee
Alternatively and equivalently we can formulate the above as a sort of Dirichlet problem of Signorini type  \cite{Stampacchia} by observing that the above equations for the Green's potential 
\be
H(z):= \int_{\mathcal K} \ln \le|\frac {z-\overline w}{z-w} \ri| \d\mu_\xi(w) 
\ee
can be expressed by the conditions:
\begin{enumerate}
\item $H(z)$ is continuous and bounded on the upper half plane, vanishing at $\infty$ and zero on $\R$;
\item $H(z)$ is harmonic away from ${\rm supp}(\d\mu_\xi)\subset \K$; 
\item it satisfies $H\geq \max \{0, -Q(z;\xi)\}$ on $\mathcal K$, and $H\geq 0$ everywhere. 
\end{enumerate}
The function $H(z)$ is also clearly extensible to the lower half plane and satisfies $H(\overline z) = -H(z)$; let us denote by the same symbol this extension and let $J(z)$ be a (multi-valued) harmonic conjugate on the (universal cover of the) domain of harmonicity of $H$. Then the function  (recall the definition of $\Phi$ \eqref{defPhi})
\be
\varphi(z;\xi):=\frac 1 {2} \Phi(z;\xi, 1) +  \le(J(z;\xi) + i H(z;\xi)\ri)
\ee
has consequently the properties: 
\begin{enumerate}
\item $\im (\varphi(z;\xi)) \geq 0$ on $\mathcal K$ and equals to zero on the support ${\rm supp}(\d\mu_\xi)\subset \mathcal K$;
\item $\partial_{n_\pm}\im \varphi<0$ at all interior points of the support, $\partial_{n_\pm}$ denotes the normal derivative  on the two sides of the circle. This is a consequence of the positivity of the measure $\d\mu_\xi$.
\item As a consequence of the previous point, $\im \varphi$ is negative on either sides of the support of $\d\mu_\xi$, within a small neighbourhood.

\end{enumerate}

This is an {\it obstacle problem} \cite{Stampacchia} where the support of the measure is a ``free boundary'' (within $\mathcal K$). See Figure \ref{obstaclefig}.
Here the ``free boundary'' is the fact that the support of the measure $\d\mu_\xi$ in general is only a subset of $\mathcal K$ and coincides with the set where the solution $H$ of the obstacle problem saturates the inequality and is equal to the obstacle function $Q$. It should be clear that the boundary of the support, within $\mathcal K$ can only occur if $Q$ is non-positive in a neighbourhood thereof; in particular if $Q$ is positive on the whole $\mathcal K$ then the minimization occurs clearly for the zero measure and $H \equiv 0$. 
In particular observe that $Q(z;\xi) =  \im \Phi(z;\xi,1)$ is positive on $\mathcal K$ as long as $\xi>\xi = 2\re(\eta_1)$. As $\xi$ decreases below $\xi$ there  is a free-boundary at a point $\alpha$ that moves leftwards along the first band until it reaches $\eta_2$ for a value $\xi_2<\xi_1$. The evolution of $\alpha$ and the transitions are determined explicitly in terms of Whitham equations in the next section.

The general description of the successions of modulated and unmodulated regions is as follows.
As $\xi$ decreases further into the first unmodulated region, the support of $\d\mu$ remains constant and equals to the first band. As $\xi<\xi_2$ the function $\im \varphi$ remains positive on the subsequent bands but has a negative  minimum along the  gap at a point, call it $\beta$,  that moves clockwise, and hence the sign of $\im \varphi$ is positive in the gap near the first band, and negative elsewhere on the gap up to $-1$. Denote the point of sign change by $\rho(\xi)$; this, too, moves clockwise until it will reach $\eta_3$, the start of the next band. At this point the support of $\d\mu$ will start ``populating'' the next band. The process repeats $\ell$ more times.

The computational approach to the above pictorial description is what is contained in the next sections.
A numerical illustration of the process is depicted in Figure \ref{obstaclefig}.
\begin{figure}[t]
\includegraphics[width=0.24\textwidth]{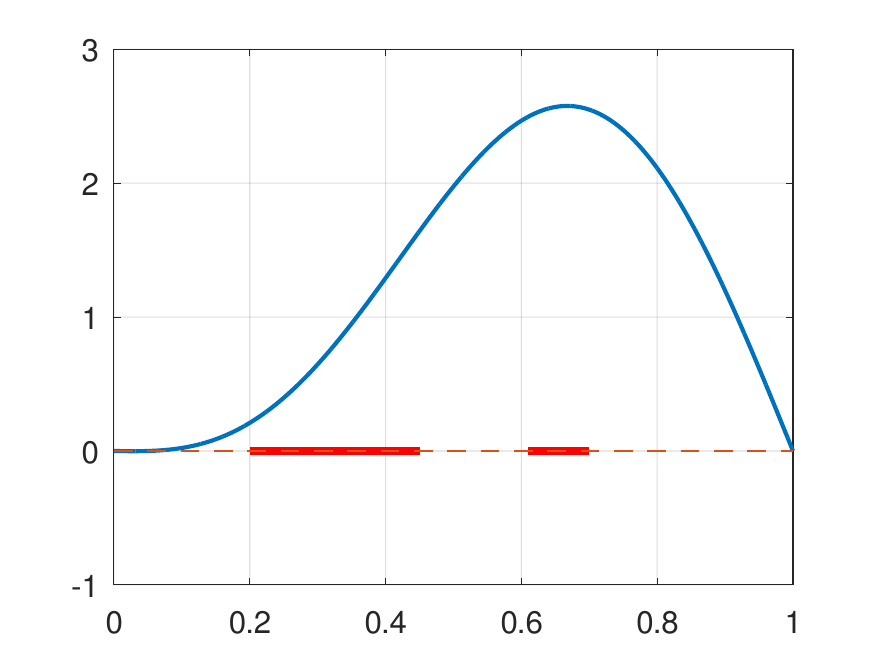}
\includegraphics[width=0.24\textwidth]{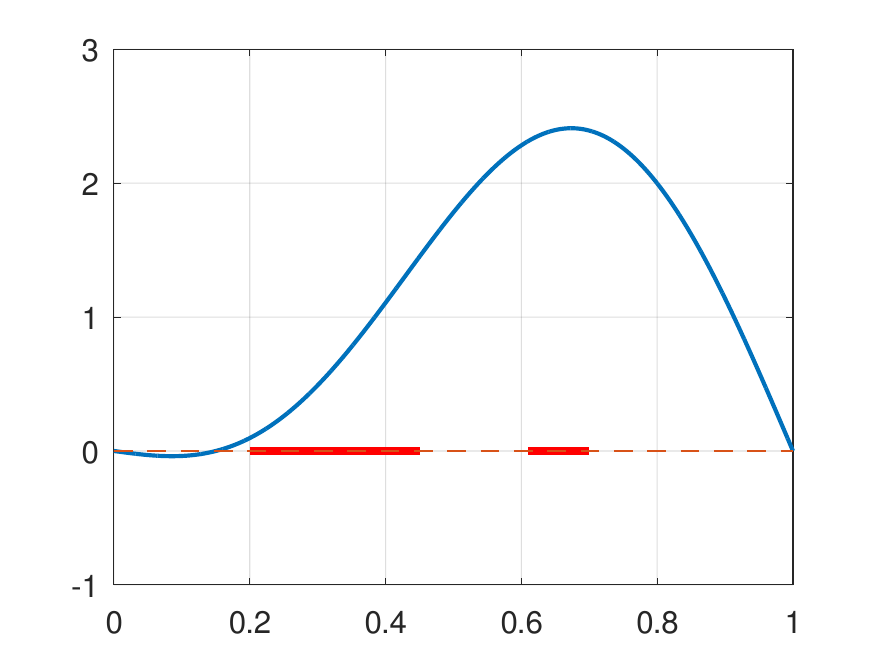}
\includegraphics[width=0.24\textwidth]{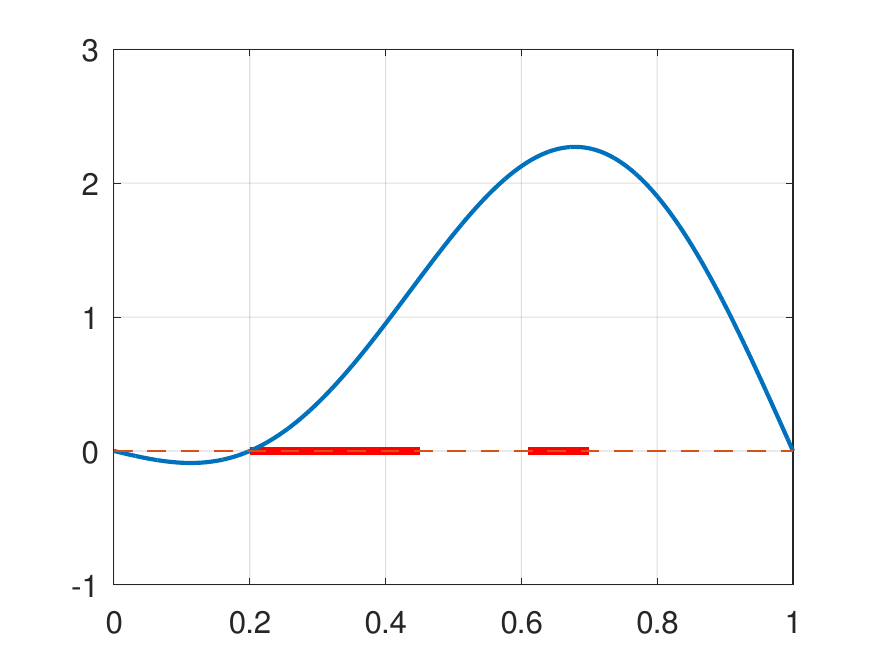}
\includegraphics[width=0.24\textwidth]{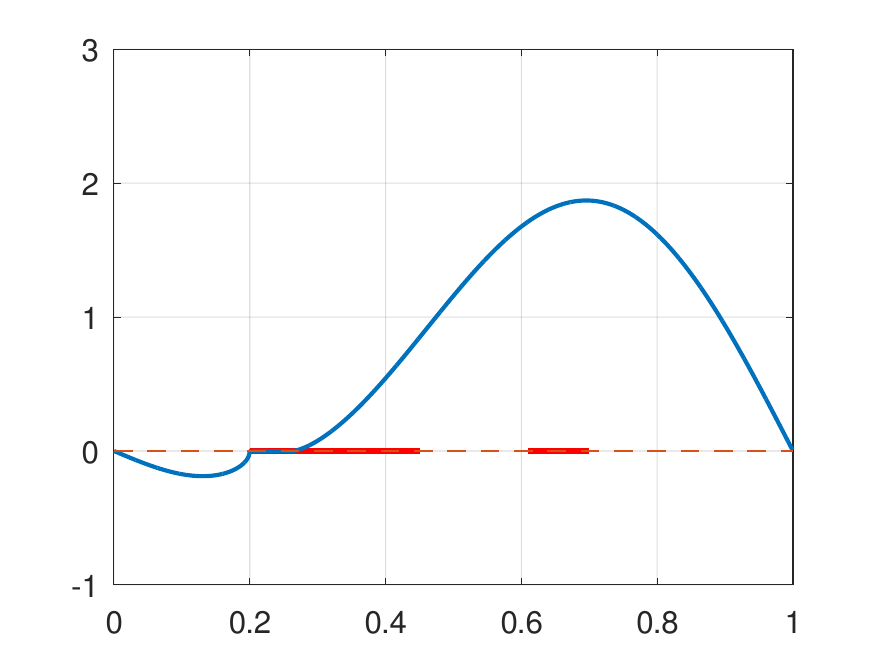}
\includegraphics[width=0.24\textwidth]{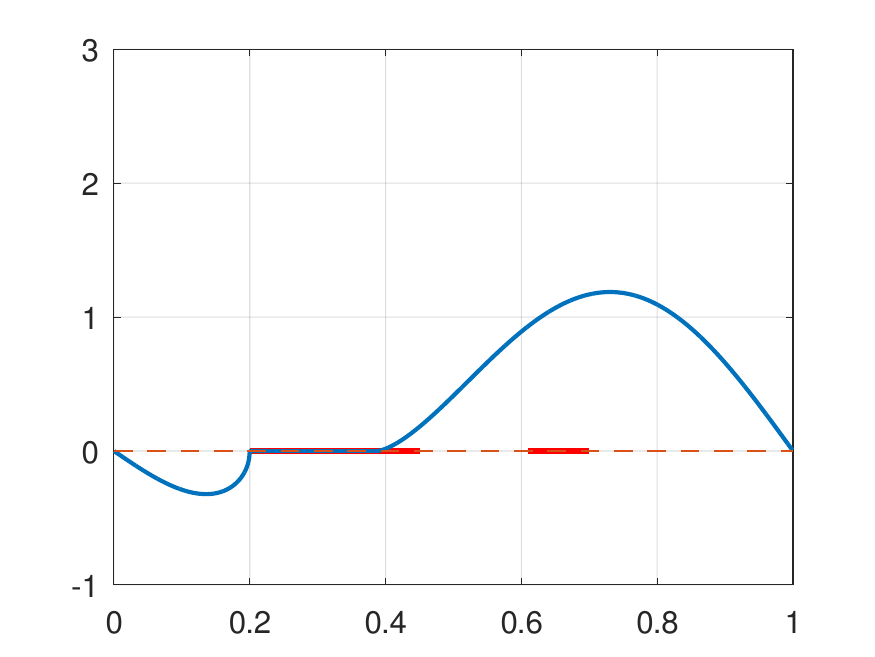}
\includegraphics[width=0.24\textwidth]{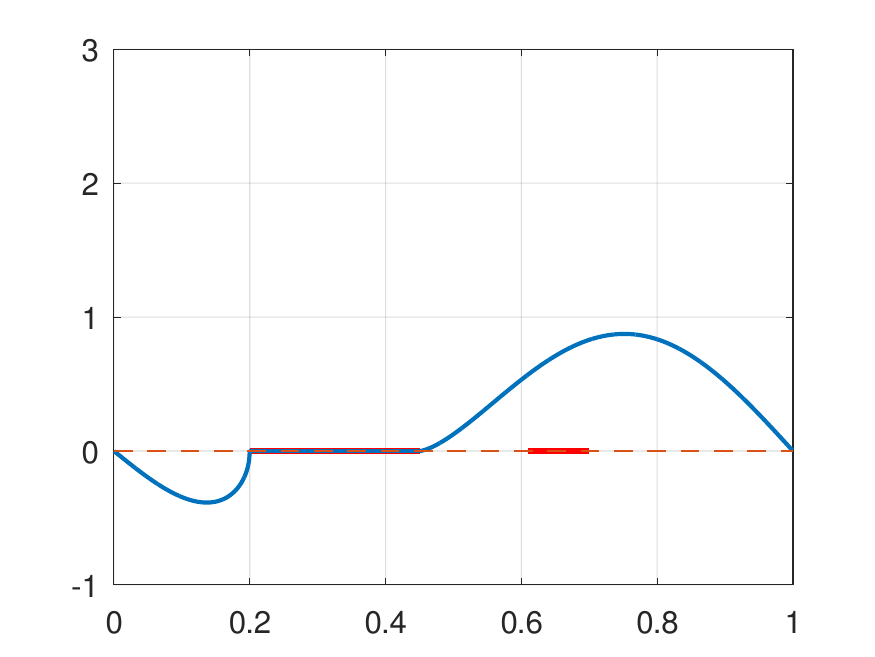}
\includegraphics[width=0.24\textwidth]{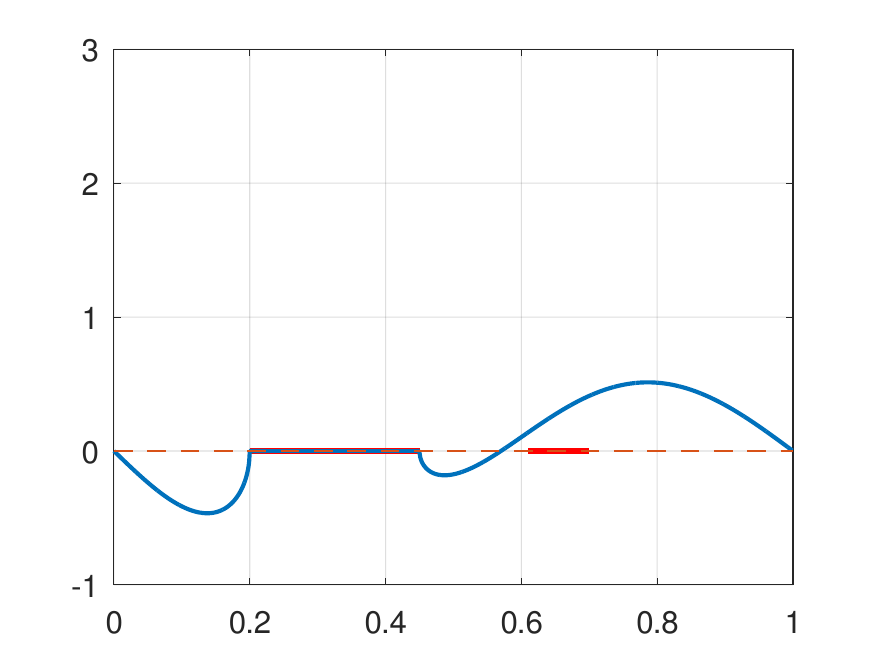}
\includegraphics[width=0.24\textwidth]{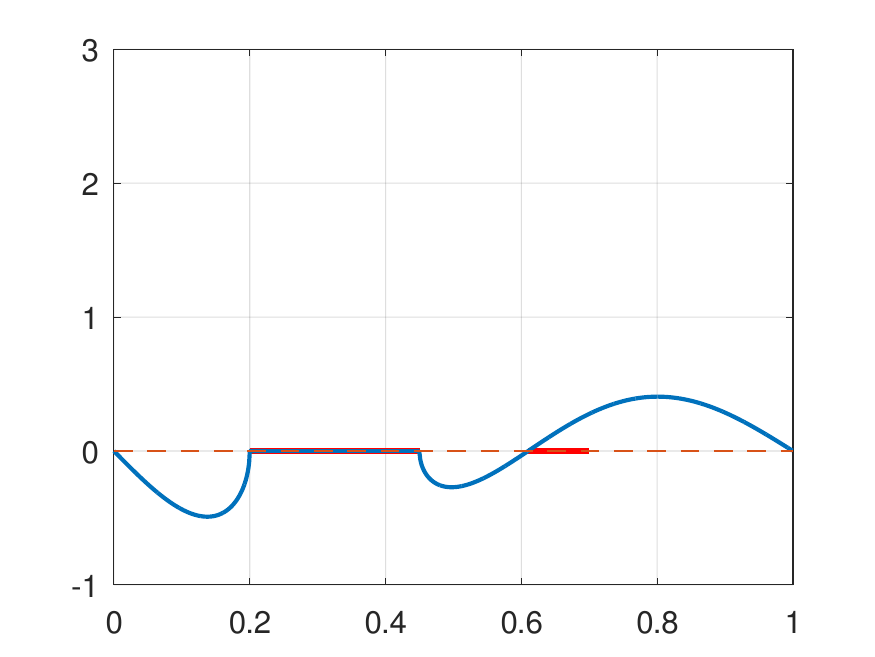}
\includegraphics[width=0.24\textwidth]{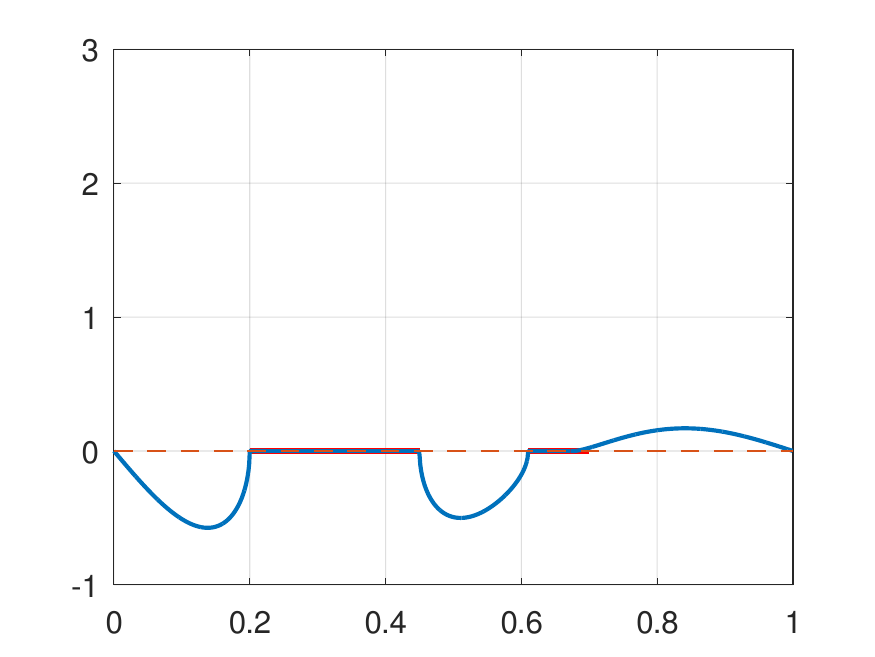}\,\,
\includegraphics[width=0.24\textwidth]{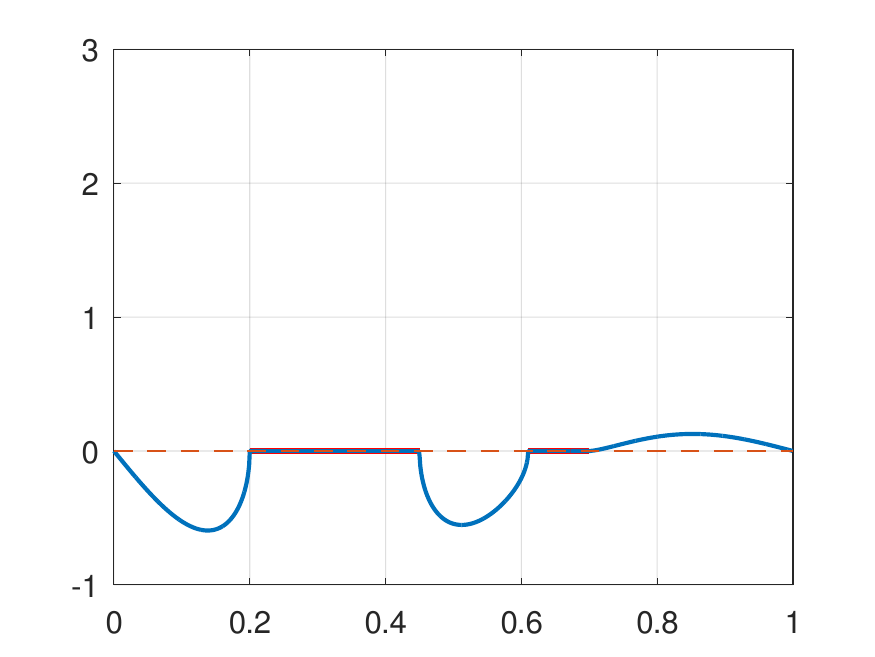}\,\,
\includegraphics[width=0.24\textwidth]{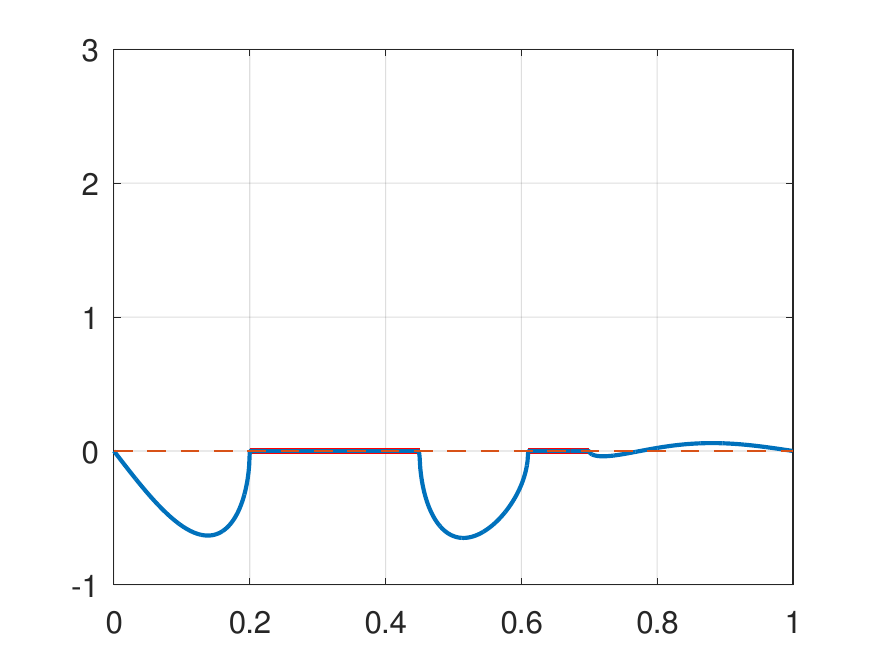}\,\,
\includegraphics[width=0.24\textwidth]{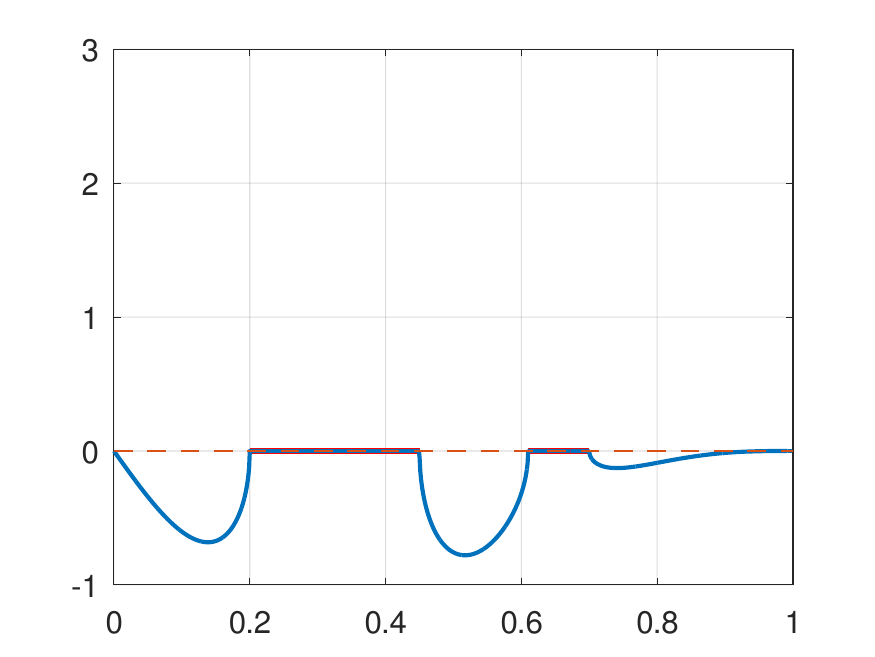}

\caption{Plots of $\im(\varphi)$ for different values of $\xi$, illustrating the obstacle problem. Here $N=2$; the horizontal axis is a parametrization of the upper semicircle ${\rm e}^{i\pi s}$, $s\in [0,1]$; indicated are the two bands with $\{\eta_1,\dots, \eta_4\} = \{ {\rm e}^{i\pi 0.2}, {\rm e}^{i\pi 0.45}, {\rm e}^{i\pi 0.61}, {\rm e}^{i\pi 0.7}\}$. The corresponding points of transition between modulated and unmodulated regions are 
$\{\xi_1,\dots, \xi_4\} \simeq \{1.618, -0.51, -1.5052 - 2.5292\}$, computed numerically using Def. \ref{evenxis} and Def. \ref{defxiodd}. In the above sequence the values of $\xi$ are, from left to right and top to bottom: 
$\xi \simeq 1.975,\ 1.782, 1.618, 1.101, 0.065, -0.501, -1.263, -1.505, -2.326, -2.529, -2.908,-3.410$}
\label{obstaclefig} 
\end{figure}

\subsubsection{Construction of the $g$-function in general: modulated regions.}
The general idea is that, as $\xi = \frac x t$ decreases, we need to define the $g$--function piecemeal on intervals of $\xi\in \R \setminus \{\xi_{2N}< \cdots <\xi_1 = 2 \re(\eta_1)\}$. Except for $\xi_1$, the other points have to be determined by a solution of a problem later on. We define the {\it modulated regions} to be the intervals $(\xi_{2\ell}, \xi_{2\ell-1})$, $ \ell =1,\dots, N$,  while the complementary intervals are the {\it unmodulated regions}.

\paragraph{Modulated region.} Fix $\ell \in \{ 1, \cdots, N \}$ and for $\xi\in (\xi_{2\ell}, \xi_{2\ell-1})$ we introduce a parameter $\alpha_{\ell}$ on the unit circle  with $\re(\alpha_\ell)\leq \re(\eta_{2\ell-1})$: this is the ``free boundary'' of the support of $\d\mu_\xi$ within the $\ell$-th band $[\eta_{2\ell-1}, \eta_{2\ell}]\subset S^1$ and the support consists of the arcs  $\gamma_j = [\eta_{2j-1} , \eta_{2j}] \ j\leq \ell-1$ together with the arc
\begin{align}
 \label{gammaells}
 \gamma_{{\ell}} &:= (\eta_{2\ell - 1}, \alpha_{\ell}) = \{{\rm e}^{i\phi} : \arg \eta_{2\ell - 1} < \phi < \arg \alpha_{\ell} \}.
\end{align}
  The parameter $\alpha$ is going to depend on $\xi$ in a step-like fashion: as $\xi$ decreases in a modulated region, $(\xi_{2\ell}, \xi_{2\ell-1})$, we will see that $\re(\alpha_\ell)$ decreases until $\xi$ reaches $\xi_{2\ell-1}$ as dictated by \eqref{xic} and Proposition \ref{alphaWhitham}; in the adjacent unmodulated region $\alpha_\ell$ is constant and equal to $\eta_{2\ell}$, only to jump to $\eta_{2\ell+1}$ as $\xi$ decreases further into the next modulated region.
We will denote, throughout this section, the gaps as the following arcs of unit circle\footnote{The notation is potentially confusing because the $\ell$--gap when discussing the $\ell$--th modulated region is ``dynamical'', while the others are ``frozen''. Since the only dynamical band or gap is the last one, we  find the issue not too critical to warrant overloading the notation.}
 \begin{align}
 \label{gammagap}
 \gamma_{j}^{\rm gap} = [\eta_{2j}, \eta_{2j+1}],\ \  j = 1,\dots, \ell-1, \qquad  \gamma^{\rm gap}_{{\ell}} &:= [ \alpha_{\ell}, -1] .
 \end{align}
 Consider the function
 \begin{equation}
 \label{Rzti}
 R_{\ell}(z)= \sqrt{(z-\alpha_{\ell}))(z-\alpha_{\ell}^{-1})} \prod_{j=1}^{2 \ell - 1} \sqrt{(z-\eta_j)(z-\eta_j^{-1})},
 \end{equation}
with  branch cuts along $ \bigcup_{j=1}^{ \ell  }(\gamma_j\cup\bar{\gamma}_j) $ and using the principal branch, defined by the asymptotic condition $R_{\ell}(z)\sim  z^{2 \ell}$ as $z\to\infty$ and $R_{\ell}(0)=1$.
 Recalling the definition of $\Phi$ in \eqref{defPhi}, the function $g_{\ell}(z)=g_{\ell}(z;\xi)$ can be constructed as
 \begin{equation}\label{tildg}
 g_{\ell}(z)=  -\frac{\Phi(z;\xi,1)}{2}  + \varphi_{\ell}(z;\xi),
 \end{equation}
 where
 \begin{gather}
 \label{defphiell}
 \varphi_{\ell}(z) = \frac{\xi}{2} \varphi_{\ell}^{(1)}(z) -  \varphi_{\ell}^{(2)}(z).
\end{gather}
Here the two differentials $\d\varphi^{(j)}_\ell$ on the Riemann surface of $R_\ell$ in \eqref{Rzti} are  the (unique) second-kind (i.e. residueless) differentials normalized to have vanishing $a$--periods (where the $a$--cycles are counterclockwise loops around the gaps) and with singular behaviour near $z=0, z=\infty$ 
\be
\label{singpart}
\d\varphi_{\ell}^{(j)} \sim  \le(z^{j-1} + \frac 1 {z^{j+1}}\ri) \d z  \ ,\ \  j=1,2, \ 
\ee 
as $z\to 0, \infty$ on the main sheet, where $R_\ell$ tends to $1$ near $0$ and  behaves as $z^{2\ell}$ near $\infty$.
They can be written explicitly as follows:
\begin{gather}
\varphi_\ell^{(1)} := 
\int_{\eta_1}^{z}\frac{P_{\ell}(s)}{R_{\ell}(s)}ds , \ \ \ 
\varphi_\ell^{(2)} := \int_{\eta_1}^{z}\frac{Q_{\ell}(s)}{R_{\ell}(s)}ds, \\
 P_{\ell}(z)=  \frac{\sum_{j=0}^{\ell} B_j z^{2\ell+2-j} + \sum_{j=0}^{\ell} B_j z^j + B_{\ell+1} z^{\ell+1}}{z^2}, \label{Pz}\\
 Q_{\ell}(z)= \frac{\sum_{j=0}^{\ell+1} C_j z^{2\ell+4-j} + \sum_{j=0}^{\ell+1} C_j z^j + C_{\ell+2} z^{\ell+2}}{z^3}, \label{Qz}
 \end{gather}
where the integration is performed on the domain 
\be
\label{defdomain}
\C^\times \setminus \big\{|z|=1, \ \arg(z) \in [\arg(\eta_1),\pi]\cup (-\pi, -\arg(\eta_1)]\big \}.
\ee
Note that, although the domain is not simply connected, the integral is well defined in the same domain because by construction the differentials are residueless.
The real constants $\{B_j\}_{j=0}^{\ell + 1}$ are determined by the singular part near $z=0,\infty$ \eqref{singpart}, the vanishing of the residues at the poles and the  vanishing of the $a$--periods (gap-integrals) along the gaps
\begin{equation}\label{Blet}
 \begin{gathered} 
 B_0=1, \quad  B_1=-\sum_{j=1}^{2 \ell - 1} \eta_{j}^{\rm re} - \alpha_{\ell}^{\rm re},  \\
 \int_{\gamma^{\rm gap}_j} \frac{P_{\ell}(s)}{R_{\ell}(s)} dz = \int_{\bar{\gamma}^{\rm gap}_j} \frac{P_{\ell}(s)}{R_{\ell}(s)} dz = 0, \quad j = 1, 2, \cdots, \ell,
 \end{gathered}
\end{equation}
with the running understanding that the last band and gap depend on $\alpha_\ell$, see (\ref{gammaells}, \ref{gammagap}). 
 Similar constraints for the constants $\{C_j\in \R\}_{j=0}^{\ell + 2}$ yield the system:
\begin{equation}\label{Clet}
 \begin{gathered} 
 C_0 = 1, \quad C_1 = B_1, \\ 
 C_2 = \frac{1}{8} \left[\left(\sum\limits_{j=1}^{2 \ell -1}  (  \eta_j + \eta_j^{-1} )\right)^2 -2 \sum\limits_{j=1}^{2\ell - 1} (\eta_j^2 + \eta_j^{-2} ) + 2 + 2(\alpha_{\ell} + \alpha_{\ell}^{-1} ) \sum\limits_{j = 1}^{2 \ell -1} ( \eta_j + \eta_j^{-1} ) -\alpha_{\ell}^2 - \alpha_{\ell}^{-2} \right],  \\
 \int_{\gamma^{\rm gap}_{\alpha_{\ell}}} \frac{Q_{\ell}(s)}{R_{\ell}(s)}  dz = \int_{\bar{\gamma}^{\rm gap}_{\alpha_{\ell}}} \frac{Q_{\ell}(s)}{R_{\ell}(s)} dz = 0, \quad
 \int_{\gamma^{\rm gap}_j} \frac{Q_{\ell}(s)}{R_{\ell}(s)} dz = \int_{\bar{\gamma}^{\rm gap}_j} \frac{Q_{\ell}(s)}{R_{\ell}(s)} dz = 0, \quad j = 1, 2, \cdots, \ell - 1,
 \end{gathered}
 \end{equation}
where $\alpha_{\ell}^{\rm re} := \re \alpha_{\ell}$.
We observe that due to the symmetries, the zeros of $P_\ell, Q_\ell$ come in pairs exchanged by $z\to z^{-1}$ and $z\to \overline z$: in particular (see below) all zeros of $P_\ell$ must be on the unit circle, while  all but at most two of those of $Q_\ell$ are on the unit circle, with the last pair either on the unit circle or on the real axis.

Note that the above gap-contours lift to a basis in homology for the Riemann surface of $R_\ell$ and since the differentials $\d\varphi_\ell^{(j)}$ have zero residues (guaranteed by the definitions of the constants $B_1$ and $C_2$ in (\ref{Blet}, \ref{Clet})), we also have the vanishing of the integral along the ``0'' gap 
\be
\gamma^{\rm gap}_{0} = \Big\{{\rm e}^{i \phi}, \ \ 0< \phi< \arg(\eta_1)\Big\}.\label{gamma0}
\ee

\par 
 Note that here we are somewhat abusing notation for brevity: indeed,  for different $\ell$ (e.g. $ s_1 \neq s_2$ ), the coefficients $\{B_j\}_{j = 0}^{s_1 + 1 }$ and $\{B_j\}_{j = 0}^{s_2 + 1 }$ are not necessarily the same, even if their subscripts appear identical for the first $\min\{s_1 ,s_2\} + 1$ terms.
 
Because of the Schwarz symmetry of the branchcuts  and of the function $\Phi$, one can verify the two symmetries 
\be
\label{Shwartz}
\varphi_\ell^{(j)} (z) =  -\varphi_\ell^{(j)} (z^{-1}) =\overline{ \varphi_\ell^{(j)} (\overline z)}. 
\ee
This implies that the imaginary part of either functions is continuous in the whole $\C\setminus \{0\}$, even across the dissection of the domain \eqref{defdomain}. 
 But more is true: since the coefficients $B_j, C_j$ are real, it also follows that all the integrals along the bands (the $b$--periods) are real, and hence both $\varphi_\ell^{(j)}$ are ``real normalized'', meaning, all periods are purely real. This implies that their imaginary parts satisfy
 \begin{enumerate}
 \item $\im\le(\varphi_\ell^{(j)}\ri)$ is continuous in $\C \setminus \{0\}$ and harmonic away from the bands and is identically zero on 
 $
\bigcup_{j=1}^{\ell} \gamma_j ,
 $
see  \eqref{gammaells}, \eqref{gammacontor}. 
 \item Both functions $\im\le(\varphi_\ell^{(j)}\ri)$  vanish  at $z = \pm 1 $\footnote{This latter is on account of the symmetry \eqref{Shwartz}.}. 
 \item The functions 
 \be
 \label{harmphij}\im\le(\varphi_\ell^{(j)}\ri) - \frac  1j\im \le(z^j-\frac 1 {z^j}\ri), \ \ j=1,2
 \ee
  extend to harmonic functions near $z=0, \infty$. 
 \item Since $\im(z-\frac 1 z)>0$ in the upper half plane, as a consequence of the minimum principle 
 \be
 \text{$\im \varphi_{\ell}^{(1)}$ is strictly positive}
 \label{phi1+}
 \ee
  in its domain of harmonicity in the upper half plane (negative in the lower), and zero on its boundary, namely, the bands $\gamma_j$'s and $\R$. 
 \item The differential $\d \varphi^{(1)}_\ell$ (namely $P_\ell$) has $2\ell+2$ zeros, one in each of the $\ell+1$ gaps $\gamma^{\rm gap}_{0},\gamma^{\rm gap}_{1}, \cdots,  \gamma^{\rm gap}_{\ell},$ together with one in each of the corresponding  conjugate gaps.
 \item The differential $\d \varphi^{(2)}_\ell$ (namely $Q_\ell$) has $2\ell+4$ zeros  and at least one in each of the $2\ell+2$ gaps above. 
 \item The functions $\varphi_\ell^{(j)}$ (and consequently $\varphi_\ell$) can be defined and analytic in $\C^\times \setminus \frak B$ with  $\frak B := \{{\rm e}^{i s}, s \in [\arg(\eta_1),\pi] \cup (-\pi, -\arg(\eta_1)]\}$ (the union of bands and gaps) as the integrals 
 \be
 \varphi_\ell^{(j)}(z) := \int_{\eta_1}^z \d \varphi_{\ell}^{(j)}.
 \ee
 In the gaps we have 
 \be
 \varphi_{\ell}^{(j)} (z_+) -  \varphi_{\ell}^{(j)} (z_-)  = \Omega_{ k}^{(j)}, \ \ \ \Omega_{k}^{(j)}  = 2 \int_{\eta_1}^{\eta_{2k}} \d \varphi_{\ell}^{(j)}, \ \ z\in \gamma_{k}^{\rm gap},  \ \ k =1,\dots, \ell,
 \label{defOmega_j}
 \ee
 where the integral is understood with a path that is outside of the unit circle with those endpoints. Observe that we have omitted the indication of $\ell$ from the above periods to avoid excessive subscripting: it should be clear that they depend on the context, namely, which of the modulated region we are inspecting. To phrase it differently, these $\Omega$'s are the $b$--periods of the normalized second kind differentials $\d\varphi_{\ell}^{(j)}$, and they are all real numbers by the Schwarz symmetries.
 In particular we have  
  \be
 \Omega_{k}(\xi) = \xi \Omega_{k}^{(1)} - \Omega_{ k}^{(2)}, \ \ \ k =1,\dots, \ell.
 \label{defOmega}
 \ee
 We point out that
 \begin{enumerate}
 \item  in the $\ell$-th modulated region the  parameters $\Omega_{k}^{(j)}$ {\it depend on }$\xi$ because so does the last endpoint $\alpha_\ell = \alpha_\ell(\xi)$;
 \item in the unmodulated regions, they are $\xi$--independent so that \eqref{defOmega} defines $\Omega_{k}$ as a genuinely affine function of $\xi$.  
\end{enumerate}
\end{enumerate}

 \paragraph{Modulated regions and definition of $\boldsymbol {\xi_{2j}},\ {j=1,\dots, N}$.} 
 In the modulated regions $\alpha_\ell\in S^1$ is going to be presently defined in such a way that $\re(\alpha_\ell)$ is an increasing function of $\xi$. 
 
 \begin{pro}
\label{alphaWhitham}
Let 
\be
\label{xic}
\xi = \frac {2Q_{\ell}(\alpha_\ell)} {P_{\ell}(\alpha_{\ell})} =  \frac {2Q_{\ell}(\alpha_\ell^{-1})}{P_{\ell}(\alpha_{\ell}^{-1})},
\ee
with $P_\ell, Q_\ell$ defined in \eqref{Pz}, \eqref{Qz}, see also  \eqref{Blet}, \eqref{Clet}. 
Then 
\begin{equation}\label{paxp}
\partial_{\phi_{\ell}}\xi=\partial_{\phi_{\ell}}\left(\frac{2Q_{\ell}({\rm e}^{i\phi_{\ell}})}{P_{\ell}({\rm e}^{i\phi_{{\ell}}})}\right)<0,
\end{equation}
with $\phi_{\ell} = \arg(\alpha_\ell)$,  and thus $\re (\alpha_\ell)$ can be defined as an increasing function of $\xi$.
The function $g_{\ell}(z)$ defined by $g_\ell = \varphi_{\ell} - \frac \Phi 2$ (with $\varphi_\ell$ in \eqref{defphiell}) satisfies the following properties:
 \begin{enumerate}[label=(\arabic*)]
 \item Analyticity: $g_{\ell}(z)$ is analytic for $z\in\mathbb{C}\setminus(\bigcup_{j=1}^{ \ell  }(\gamma_j\cup\bar{\gamma}_j\cup\gamma^{\rm gap}_j\cup\bar{\gamma}^{\rm gap}_j)
 $.
 \item Jump conditions: 
  \begin{align}\label{tigj}
 g_{\ell+}(z)+g_{\ell-}(z)&= - \Phi(z;\xi,1), && z\in \bigcup_{j=1}^{ \ell }(\gamma_j\cup\bar{\gamma}_j), \\
 g_{\ell+}(z)-g_{\ell-}(z)&= \Omega_j=\Omega_j(\xi), && z\in \gamma^{\rm gap}_j\cup\bar{\gamma}^{\rm gap}_j, \quad j =1 ,\cdots , \ell ,
 \end{align}
 where for brevity  $\gamma_\ell^{\rm gap}= \gamma_{\alpha_\ell}^{\rm gap}$ and $\Omega_j$'s are defined by \eqref{defOmega} (we omit the reference to $\ell$ for brevity). 

 \item Asymptotic behaviours: 
 \begin{equation}\label{tginf}
 \begin{aligned}
 g_{\ell}(z;\xi) &= g_{\ell\infty}(\xi) + \frac{g_{\ell}^{(1)}(\xi)}{z} + \mathcal{O}(z^{-2}),&& z\to\infty,  \\
 g_{\ell}(z;\xi)&= -g_{\ell\infty}(\xi)+\mathcal{O}(z),   && z\to 0,   \\
 g_{\ell}(z;\xi)&= -\frac{\Phi(z;\xi,1)}{2} + \mathcal{O}\left((z-\eta_j^{\pm 1})^{\frac{1}{2}}\right), && z\to\eta_j^{\pm 1}, \quad j=1, \cdots, 2\ell -1, \\
 g_{\ell}(z;\xi)&= -\frac{\Phi(z;\xi,1)}{2} + \mathcal{O}\left((z-\alpha_{\ell}^{\pm 1})^{\frac{3}{2}}\right), && z\to\alpha_{\ell}^{\pm 1}.
 \end{aligned}
 \end{equation}
 \item Symmetry: $g_{\ell}(z)+g_{\ell}(z^{-1})=0$.
 \end{enumerate}
 \end{pro}
 \begin{proof}
 Recall the expression of $g_{\ell} (z)$ in \eqref{tildg}. The choices of parameters in \eqref{Blet} and \eqref{Clet} are such that $g_{\ell}(z)$ is analytic at $z=0$ and $z=\infty$  because of \eqref{harmphij}. The remaining properties follow from similar properties for $\varphi_\ell$ and are an exercise.
 \par 
Looking at \eqref{defphiell} we observe that the $\frac{3}{2}$-vanishing behaviour of $\varphi_{\ell}(z)$ near $\alpha_{\ell}^{\pm1}$ (hence vanishing of $\varphi'_{\ell}$ as square-root) is a consequence of the definition \eqref{xic}, which is the requirement of the vanishing of the numerator of $\varphi'$, namely $\frac x2 P_\ell - t Q_\ell$.
 We can see the variation plot of $\alpha_{\ell}$ with respect to $\xi$ for the case $N=1$ in Figure \ref{signc}). The rest of the proof requires to establish \eqref{paxp}. This  is presented in Appendix \ref{App2}.
\end{proof}

\begin{defi}
\label{evenxis}
The parameters $\xi_{2\ell}, \  \ell=1,\cdots, N,$ delimiting the transition from a modulated to unmodulated region  (as $\xi$ decreases)  are defined by the condition 
\be
\xi_{2\ell} := \xi(\phi_\ell)\bigg|_{\phi_\ell = \arg (\eta_{2\ell})} = \frac {2Q_\ell(\eta_{2\ell})}{P_\ell(\eta_{2\ell})} \label{xieven}.
\ee
Here $P_\ell, Q_\ell$ are the numerators of the real-normalized second-kind differentials $\d\varphi_{\ell}^{(j)}$ defined in (\ref{Pz}, \ref{Qz}) and with the conditions (\ref{Blet}, \ref{Clet}) in place, and with $R_\ell$ as in \eqref{Rzti} with $\alpha_\ell = \eta_{2\ell}$. 
\hfill $\triangle$ \end{defi}
To elucidate the definition \ref{evenxis} we observe that Prop. \ref{alphaWhitham} says that $\xi(\phi_\ell)$ is a decreasing  function of $\phi_\ell$ and hence we define $\xi_{2\ell}<\xi_{2\ell-1}$. Now, the odd transition points $\xi_{2\ell-1}$ have not been defined yet, except for $\xi_1 = 2 \re(\eta_1)$, which serves as base-case for an inductive construction. The other odd transition points are defined in the Section \ref{sunun}, Def. \ref{defxiodd}.

\subsection{Unmodulated genus-$\ell$ gas region and definition of $\xi_{2\ell+1}$.}
 \label{sunun}

Recall that $\xi_1 = 2 \re(\eta_1)$; the rest of the transition points $\xi_\ell$ are defined recursively as follows. As $\xi<\xi_1$ enters the first modulation equation, we define $P_1,Q_1$ as described  in the previous section. Then $\alpha_1(\xi)$ moves leftwards along the unit circle by virtue of Prop. \ref{alphaWhitham}.We define $\xi_2$ when $\alpha_1(\xi) = \eta_2$. As $\xi$ keeps decreasing into the first un-modulated region we keep $\alpha_1=\eta_2$ constant and hence $P_1(z), Q_1(z)$ are {\it independent of $\xi$} within the unmodulated region. 

We proceed inductively; when $\xi= \xi_{2\ell}$ (as in Def. \ref{evenxis}) we have $\alpha_\ell = \eta_{2\ell}$ and so define 
\be
\hat R_\ell(z) = \prod_{j=1}^{2 \ell } \sqrt{(z-\eta_j)(z-\eta_j^{-1})},
\label{Rhat}
\ee
which will remain constant with respect to $\xi<\xi_{2\ell}$. Then define the function 
\be
\label{hatphi}
V(z;\xi):= \frac 1 t{\hat  \varphi}_{\ell}(z) = \frac {\xi}2 {\hat \varphi}^{(1)}_{\ell}(z)  -  {\hat \varphi}^{(2)}_{\ell}(z) .
\ee 
Here $\hat{\varphi}^{(j)}_{\ell}(z)$ are the same differentials defined earlier in (\ref{Pz}, \ref{Qz}) but on the {\it $\xi$--independent} Riemann surface of $\hat R_\ell$ \eqref{Rhat} and in particular are also $\xi$-independent. We denote $\hat P_\ell, \hat Q_\ell$ the corresponding numerators of their differentials.

\begin{pro}
\label{unmodulatedphi} Consider the function $V(z;\xi)$ in \eqref{hatphi}.
\begin{enumerate}
\item for $\xi<\xi_{2\ell}$ the imaginary part changes sign exactly once in said gap at a point,  $\rho(\xi)$, changing sign from negative to positive as $z$ traverses the gap counterclockwise; it has  a nondegenerate positive maximum and negative minimum at two points $ s(\xi), \beta(\xi)$, respectively. 
\item $\arg(\beta(\xi), \arg(s(\xi)), \arg(\rho(\xi))$ are all decreasing functions of $\xi$ and there exists $\xi_0<\xi_{2\ell}$ such that 
$$
\lim_{\xi\to \xi_0^+} s(\xi) =\lim_{\xi\to \xi_0^+} \rho(\xi) =-1.
$$  
\end{enumerate}
\end{pro}
\begin{rmk}
Similar statements can be made for the conjugate gap and points, given the Schwarz symmetry ${\hat \varphi}_\ell(z) = - \overline{{\hat \varphi}_\ell(\overline{z})}.$
\end{rmk}
\begin{proof}
During the proof it is helpful to refer to Figure \ref{obstaclefig}.\\
{\bf 1.} 
The function 
\be
\label{integralV}
V(z;\xi)  = \int_{\eta_1}^z V'(z;\xi)\d z = \int_{\eta_1}^z \frac{\xi \hat P_{\ell}(z) - \hat Q_{\ell}(z)}{\hat R_{\ell}(z)} \d z, 
\ee
with integration on the simply connected domain \eqref{defdomain},  
 is, in principle, analytic only in the same domain. 
%
%
The imaginary part $\im V$ is continuous and harmonic across all the gaps. 
Observe that  $\im V(-1;\xi)=0 = \im V(\eta_{2\ell};\xi)$ (which follows from the vanishing of the gap integrals). Moreover, for $\xi = \xi_{2\ell}$ the imaginary part is positive throughout $\gamma_{\ell}^{\rm gap}$. The dependence of $V(z;\xi)$ on $\xi$ is linear, \eqref{integralV}, and the derivative with respect to $\xi$ is
\be
\frac {\d}{\d \xi} V(z;\xi) =  \int_{\eta_1}^z \frac{\hat P_{\ell}(z)}{\hat R_{\ell}(z)} \d z=: \dot V(z) = \hat \varphi^{(1)}_\ell.
\ee
The imaginary part of the above expression is a strictly positive function   in the upper half plane by the same reason as explained around \eqref{phi1+}, vanishing on the bands and behaving like $\im(z)$ at $z=\infty$ and $\im (-\frac 1 z)$ near $z=0$. Thus the linear pencil  of functions $\im V(z;\xi)$ on $\gamma_{\ell}^{\rm gap}$ is a strictly increasing family (pointwise), vanishing at the two endpoints $\eta_{2\ell}, -1$ identically in $\xi$. Thus, as $\xi$ {\it decreases}, the function $V(z;\xi)$ necessarily will become negative somewhere along the last gap and become entirely negative for some finite  $\xi$ sufficiently smaller than $\xi_{2\ell}$, after which it remains negative. 

We need to show that for $\xi<\xi_{2\ell}$  there is a unique sign change at a point $\rho(\xi)$ that moves towards $-1$  (as $\xi$ decreases) and reaches it in finite $\xi$-time. 

Indeed:
\begin{enumerate}
\item [-] for any (small) $\epsilon>0$ the function $\im V$ becomes negative for $\xi \in (\xi_{2\ell}-\epsilon, \xi_{2\ell})$  in a neighbourhood of $\eta_{2\ell}$ along $\gamma_{\ell}^{\rm gap}$: this is so because $V(z;\xi_{2\ell})$ vanishes like $|z-\eta_{2\ell}|^\frac 32$ while its $\xi$-derivative $\dot V(z)$ vanishes like $|z-\eta_{2\ell}|^\frac 1 2$. Thus there is at least one sign change of $\im V$ along $\gamma_{\ell}^{\rm gap}$ and at least a positive maximum and negative minimum. 
\item[-] However the function $\im V$ can have at most two critical points in $\gamma_{\ell}^{\rm gap}$. Indeed the critical points of $\im V$ coincide with the zeros of $\d V$ and the numerator of $\d V$ \eqref{integralV} has $\ell+2$  conjugate pairs of zeros on the unit circle.  The vanishing of the gap integrals then forces the remaining zeros to be one in each gap. 
Since there must be at least two critical points of $V$ along  $\gamma_\ell^{\rm gap}$ and one in each of the remaining $\ell$ gaps in the upper half plane, by the pigeon-hole principle there are exactly two critical points in the last gap (and exactly one in each of the other gaps).
 Since now there are two critical points on $\gamma_{\ell}^{\rm gap}$, there is exactly  one sign change at a point that we denote by $\rho(\xi)$. 
\item[-] 
Let $\beta(\xi), s(\xi) $ denote the points of minimum and maximum, respectively, and $\rho(\xi)$ denote the zero  of $\im V$ in $\gamma_\ell^{\rm gap}$: quite clearly we must have, for all $\xi<\xi_{2\ell}$
\be
\pi< \arg s(\xi)<\arg\rho(\xi)< \arg\beta(\xi)<  \arg(\eta_{2\ell}).
\ee
\item[-] Finally, as $\xi$ decreases so must the function $\im V(z;\xi)$ pointwise along $\gamma_{\ell}^{\rm gap}$; thus the point $\rho(\xi)$ must move towards $-1$ and reach it in finite $\xi$-time since the function $V(x;\xi)$ is completely negative for $\xi$ small enough (and never identically zero since ${\hat \varphi}_{\ell}^{(j)},\ j=1,2$ are linearly independent). 
\end{enumerate}
{\bf 3.}  Since $\im V(z;\xi)$ depends linearly on $\xi$ and $\im \dot V$ is a positive function in the last gap, for a finite value of $\xi$ smaller than $\xi_{2\ell}$  the  function $\im V$ becomes completely negative on $\gamma_\ell^{\rm gap}$. Thus $\rho(\xi)$, by a simple continuity argument, must reach $-1$ at some finite value of $\xi$. Since $s(\xi)$ (the minimum of $\im V)$ is squeezed between the two nodes, also $s(\xi)$ tends to $-1$ at the same value of $\xi$. Incidentally, asymptotically for very large and negative $\xi$ the local minimum $\beta(\xi)$  will reach the local minimum of $\dot V$, namely the zero of $\wh P_\ell$  in the last gap.
\end{proof}

Having established Prop. \ref{unmodulatedphi} we can thus proceed with the definition of the transition points $\xi_{2\ell+1}$, thus completing the induction procedure to define all the modulated/unmodulated regions. Colloquially, this happens when $\im V(z;\xi)$ starts becoming negative on the support of the next band $[\eta_{2\ell+2}, \eta_{2\ell+1}] \subset S^1$. 
\begin{defi}
\label{defxiodd}
The parameters $\xi_{2\ell+1}, \ \ell=1,\cdots, N,$ are defined implicitly by the condition 
\be
\rho(\xi_{2\ell+1}) = \eta_{2\ell+1}, 
\ee
with $\rho(\xi)$ as in Proposition \ref{unmodulatedphi}, 
or explicitly by the equation  
\be
\xi_{2\ell+1} = 2\int_{\eta_{2\ell}}^{\eta_{2\ell+1}} \frac {{\hat Q_{\ell} (z)
\d z}}{\hat R_{\ell} (z)} 
\bigg/\int_{\eta_{2\ell}}^{\eta_{2\ell+1}} \frac {{\hat P_{\ell} (z)
\d z}}{\hat R_{\ell} (z)}, \label{xiodd}
\ee
where $\hat P_\ell, \hat Q_\ell$ are defined  below \eqref{hatphi}. 
\hfill $\triangle$ \end{defi}

In a parallel fashion as what was done in the modulated case, we define the function $\hat{g}_{\ell}(z)=\hat{g}_{\ell}(z;x,t)$:
 \begin{equation}\label{hatg}
 \hat{g}_{\ell}(z;\xi )=-\frac{\Phi(z;\xi,1)}{2} + \frac{\xi}{2}\int_{\eta_1}^{z}\frac{\hat{P}_{\ell}(s)}{\hat{R}_{\ell}(s)}ds
 -\int_{\eta_1}^{z}\frac{\hat{Q}_{\ell}(s)}{\hat{R}_{\ell}(s)}ds,
 \end{equation}
 where recall that $\hat{R}_{\ell}(z)$, $\hat{P}_{\ell}(z)$ and $\hat{Q}_{\ell}(z)$ are defined by \eqref{Rzti}, \eqref{Pz} and \eqref{Qz} with $\alpha_{\ell} = \eta_{2 \ell}$, respectively (and $\Phi$ is in \eqref{defPhi}). 
 The following lemma guarantees the appropriate inequalities that will be used in the unfolding of the steepest descent analysis:
    \begin{lem}\label{lem2}
      For in the modulated region $\xi_{2 \ell}< \xi < \xi_{2 \ell - 1}$, the function $\varphi_{\ell}(z)$ satisfies the following properties:
\begin{align}
         &\im(\varphi_{\ell}(z))<0,\quad   z\in\Gamma_{\alpha_{\ell}+}\cup\Gamma_{\alpha_{\ell}-}\setminus \left(\bigcup_{j=1}^{2\ell - 1} \{\eta_j\} \cup \{\alpha_{\ell}\}\right), \label{em21} 
         \\
         &\im[(\varphi_{\ell_+}(z) + \varphi_{\ell_-}(z))] >0, \quad   z\in [\alpha_\ell,-1]\subset S^1, \label{em23}
\end{align}
with, correspondingly, reversed inequalities in the lower half plane on account of the Schwarz symmetry \eqref{Shwartz}.
{ In the $\ell$-th unmodulated region (with $\ell = N$ being the last unbounded one $(-\infty, \xi_{2\ell})$) the formulation is the  same, with $\alpha_\ell = \eta_\ell$ and $\gamma_{\alpha_\ell} = \gamma_\ell$.}
   \end{lem}
   
   \begin{proof}
      We start from the observation that
      \begin{equation}\label{abx}
        \d  \varphi_{\ell}(z)=-\frac{ (z - \alpha_{\ell}) (z - \alpha_{\ell}^{-1}) \prod_{j = 0}^{\ell } (z - s_j) (z - s_j^{-1})  }{z^3R_{\ell}(z)}.
      \end{equation}
 Beside $\alpha_\ell$ the other  zeros  in the upper half plane must fall within the gaps, $ s_j \in \gamma^{\rm gap}_j \cap \mathbb{C}^+ $, $j = 0, \cdots, \ell$: this is so because  of the vanishing of the periods $\int_{\gamma_j^{\rm gap}} \d \varphi_\ell  =0$ and the Schwarz symmetry \eqref{Shwartz} so that the integrand is purely imaginary along the gaps and this its imaginary part  must change sign within the gaps at least once, the rest following by the pigeonhole principle.

For the inequality \eqref{em21} we reason as follows: the boundary values of the  integrand $\d\varphi_\ell(z_\pm)$  are  real along the bands, do not change sign and $-\d\varphi_\ell(z_-)= \d \varphi_\ell(z_+)<0$ within  the $j$-th band $\gamma_j$. Consider the $+$ boundary value:  real part of the integrand being negative, the real part $\varphi_\ell(z_+)$ is increasing, and by the Cauchy-Riemann equations the imaginary part is decreasing in the normal direction, guaranteeing the inequality in  a left neighbourhood of the bands $\gamma_j$,  $j=1,\dots, \ell$. On the other side, the argument leads to the same conclusion because  real part is (instead) increasing so that the imaginary part is also {\it decreasing} (because the normal is the negatively oriented).  

For the inequality \eqref{em23} along the last gap $\gamma_\ell^{\rm gap}=[\alpha_\ell, -1]$ we reason as follows: since $ {\Omega}_j\in\mathbb{R}$, $j = 1, \cdots, \ell$,  we have
      \begin{equation}
         \im (\varphi_{\ell}(z_+))=\im( \varphi_{\ell}(z_-))
         =\im\left[\int_{\alpha_{\ell}}^{z}\frac{\frac{\xi}{2}P_{\ell}(s)-Q_{\ell}(s)}{R_{\ell}(s)}ds\right].
      \end{equation}
      It is straightforward to check that the integrand is purely imaginary along the circle and the imaginary part changes sign only once, at $z= s_\ell\in [\alpha_\ell,-1]$, passing from positive to negative. Namely $\im \int_{\alpha_{\ell}}^{z}\frac{\frac{x}{2}P_{\ell}(s)-tQ_{\ell}(s)}{R_{\ell}(s)}ds$ increases as $z$ moves from $\alpha_\ell$ to $s_\ell$ and then decreases. Since also the total integral of the last gap is zero, $\im\varphi_\ell(-1)=0$ and thus $\im\varphi_\ell(z)$ never changes sign within $[\alpha_\ell,-1]$.
We invite the reader to look at Figure \ref{signc}.
%
%
 \end{proof}

%
%

%
\paragraph{Szeg\"o function $\delta$.}
In the steepest descent analysis we will need the so-called Szeg\"o function. This is an Abelian integral defined on the same hyper-elliptic Riemann surface $R_\ell$ or $\wh R_\ell$ in the modulated and unmodulated case, respectively. Thus we can define it in one stroke for either cases by looking only at the case of the modulated regions, with the understanding that in the unmodulated regions the same formulas below will apply but setting everywhere $\alpha_\ell =\eta_{2\ell}$. 
 The function $\delta_{\ell}(z)=\delta_{\ell}(z;x,t)$ is given by
\begin{equation}\label{deltas}
\begin{aligned}
\delta_{\ell}(z) = \exp \Biggl[ & \frac{R_{\ell}(z)}{2\pi i} \sum_{j=1}^{\ell -1} \bigg( \int_{\gamma_j}\frac{-\ln r(s)}{R_{\ell}(s)(s-z)}ds
+\int_{\bar{\gamma}_j}\frac{\ln \overline{r(\bar{s})}}{R_{\ell}(s)(s-z)}ds
+\int_{\gamma^{\rm gap}_j \cup \bar{\gamma}^{\rm gap}_j}\frac{i {\Delta}_j}{R_{\ell}(s)(s-z)}ds \\
& +\int_{\gamma_{\alpha_{\ell}}}\frac{-\ln r(s)}{R_{\ell}(s)(s-z)}ds
+\int_{\bar{\gamma}_{\alpha_{\ell}}}\frac{\ln \overline{r(\bar{s})}}{R_{\ell}(s)(s-z)}ds
+\int_{\gamma^{\rm gap}_{\alpha_{\ell}} \cup \bar{\gamma}^{\rm gap}_{\alpha_{\ell}}}\frac{i {\Delta}_{\ell}}{R_{\ell}(s)(s-z)}ds \bigg) \Biggr],
\end{aligned}
\end{equation}
where the parameter $ {\Delta}_j$, $j=1, \cdots, \ell,$ are determined by linear equations
\begin{equation}\label{tdelta}
\begin{aligned}
\sum_{j=1}^{\ell - 1} \Biggl( 
& \int_{\gamma_j}\frac{-\ln r(s) s^m}{R_{\ell}(s)}ds
+\int_{\bar{\gamma}_j}\frac{\ln \overline{r(\bar{s})} s^m}{R_{\ell}(s)}ds
+\int_{\gamma^{\rm gap}_j \cup \bar{\gamma}^{\rm gap}_j}\frac{i s^m  {\Delta}_j}{R_{\ell}(s)}ds \\
& +\int_{\gamma_{\alpha_{\ell}}}\frac{-\ln r(s) s^m}{R_{\ell}(s)}ds
+\int_{\bar{\gamma}_{\alpha_{\ell}}}\frac{\ln \overline{r(\bar{s})} s^m}{R_{\ell}(s)}ds
+\int_{\gamma^{\rm gap}_{\alpha_{\ell}} \cup \bar{\gamma}^{\rm gap}_{\alpha_{\ell}}}\frac{i s^m  {\Delta}_{\ell}}{R_{\ell}(s)}ds 
\Biggr) = 0, 
\\
 &m = 0, \cdots, \ell -1.
\end{aligned}
\end{equation}
 Here we have also used the property $\overline{r(\bar{z}^{-1})}=r(z)$.
 \begin{pro}
Let $\xi$ be in the $\ell$-th modulated region $\xi_{2 \ell}< \xi < \xi_{2 \ell - 1}$, and denote for brevity   $\gamma_\ell = \gamma_{\alpha_\ell}$ and $\gamma_\ell^{\rm gap} = \gamma_{\alpha_\ell}^{\rm gap}$. The function $\delta_{\ell}(z)=\delta_{\ell}(z;x,t)$ satisfies the following properties:
 \begin{enumerate}[label=(\arabic*)]
 \item Analyticity: $\delta_{\ell}(z)$ is analytic for $z\in\mathbb{C}\setminus(\bigcup_{j=1}^{ \ell  }(\gamma_j\cup\bar{\gamma}_j\cup\gamma^{\rm gap}_j\cup\bar{\gamma}^{\rm gap}_j)$. 
 \item Jump conditions: \begin{equation}
 \begin{aligned}
 \delta_{\ell+}(z)\delta_{\ell-}(z)&= r(z)^{-1}, && z\in  \bigcup_{j=1}^{\ell } \gamma_j,  \\
 \delta_{\ell+}(z)\delta_{\ell-}(z)&= \overline{r(\bar{z})},  && z\in  \bigcup_{j=1}^{\ell } \bar{\gamma}_j , \\
 \delta_{\ell+}(z)/\delta_{\ell-}(z)&= {\rm e}^{i {\Delta}_j}, && z\in \gamma^{\rm gap}_j \cup \bar{\gamma}^{\rm gap}_j, \quad j=1, \cdots, \ell .
 \end{aligned}
 \end{equation}
 \item Asymptotic behaviours: 
 \begin{equation}
 \begin{aligned}
 \delta_{\ell}(z)&= \delta_{\ell \infty}(x,t)+\mathcal{O}(z^{-1}),&& z\to\infty,  \\
 \delta_{\ell}(z)&= \delta_{\ell \infty}(x,t)^{-1}+\mathcal{O}(z),   && z\to 0.  
 \end{aligned}
 \end{equation}
 \item Symmetry: $\delta_{\ell}(z)\delta_{\ell}(z^{-1})=1$.
 \item  in an unmodulated region $\xi_{2\ell+1} < \xi< \xi_{2\ell}$ the formulas are the same but with $\alpha_\ell = \eta_{2\ell}$ (independent of $\xi$). In the last unbounded unmodulated region $-\infty<\xi<\xi_{2N}$ we have $\ell = N$ and $\alpha_N = \eta_{2N}$. In the unmodulated regions $\delta_{\ell}(z)$ is also independent of $\xi$.  
 \end{enumerate}
 \end{pro}
 \begin{proof}
 The proof proceeds by inspection of \eqref{deltas} and use of Sokhotski-Plemelj formula. Note that 
 \begin{equation}\label{ttdel}
 \begin{aligned}
 \delta_{\ell\infty} := \delta_{\ell\infty}(\xi) = \exp \Biggl[& -\frac{1}{2\pi i} \sum_{j=1}^{\ell -1} \bigg( \int_{\gamma_j}\frac{-\ln r(s)s^{2\ell-1}}{R_{\ell}(s)}ds + \int_{\bar{\gamma}_j}\frac{\ln \overline{r(\bar{s})}s^{2\ell-1}}{R_{\ell}(s)}ds + \int_{\gamma^{\rm gap}_j \cup \bar{\gamma}^{\rm gap}_j}\frac{i s^{2\ell-1}  {\Delta}_j}{R_{\ell}(s)}ds \\
 &\quad  + \int_{\gamma_{\alpha_{\ell}}}\frac{-\ln r(s)s^{2\ell-1}}{R_{\ell}(s)}ds + \int_{\bar{\gamma}_{\alpha_{\ell}}}\frac{\ln \overline{r(\bar{s})}s^{2\ell-1}}{R_{\ell}(s)}ds + \int_{\gamma^{\rm gap}_{\alpha_{\ell}} \cup \bar{\gamma}^{\rm gap}_{\alpha_{\ell}}}\frac{i s^{2\ell-1}  {\Delta}_{\ell}}{R_{\ell}(s)}ds \bigg) \Biggr],
 \end{aligned}
 \end{equation}
 where the dependence on $x$, $t$ arises only from $\alpha_{\ell}=\alpha_{\ell}(\xi)$. By contrast, in the $\ell$-th unmoodulated region where $\alpha_\ell \equiv \eta_{2\ell}$, therefore, $\delta$ is $\xi$-independent.
 \end{proof}
 \par 
 Now, we can introduce the following transformation of the Riemann--Hilbert problem:
  \begin{equation}\label{tsf}
  {S}(z)=\delta_{\ell \infty}^{-\sigma_3}{\rm e}^{-i {t}g_{\ell\infty}\sigma_3}M(z){\rm e}^{i{t}g_{\ell}(z)\sigma_3}\delta_{\ell}(z)^{\sigma_3}.
 \end{equation}
 This  yields the following RHP, in which $ {S}(z) =  {S}(z, \ell; x,t)$ actually depends on $\ell$ but we omit the explicit reference  for brevity.
 \begin{RHP}
 With the understanding that here $\gamma_\ell = \gamma_{\alpha_\ell}$ and $\gamma_{\ell}^{\rm gap} =\gamma_{\alpha_\ell}^{\rm gap}$, we have the following properties:
 \begin{enumerate}
 \item Analyticity: $  {S}(z)$ is analytic for $z\in\mathbb{C}\setminus(\bigcup_{j=1}^{ \ell  }(\gamma_j\cup\bar{\gamma}_j\cup\gamma^{\rm gap}_j\cup\bar{\gamma}^{\rm gap}_j) \cup \{0\} )$.
 \item Asymptotic behaviours: $  {S}(z)=\mathbb{I}+\mathcal{O}(z^{-1})$ as $z\to\infty$ and $  {S}(z)=\frac{\sigma_1}{z}+\mathcal{O}(1)$ as $z\to0$.
 \item Symmetry: $  {S}(z)=\sigma_1\overline{ {S}(\bar{z})}\sigma_1=z^{-1} {S}(z^{-1})\sigma_1$.
 \item Jump conditions: for $j=1, \cdots, \ell$, it follows that
 \begin{equation}
  {S}_+(z)= {S}_-(z)\begin{cases}
 \begin{pmatrix}
 {\rm e}^{-2i{t}\varphi_{\ell-}(z)}r^{-1}(z)\delta_{\ell-}^{-2}(z) & 0 \\
 -i & {\rm e}^{-i{t}\varphi_{\ell+}(z)}r^{-1}(z)\delta_{\ell+}^{-2}(z)
 \end{pmatrix}, & z\in \bigcup_{j=1}^{\ell } \gamma_j ,
 \\
 \begin{pmatrix}
 {\rm e}^{2i{t}\varphi_{\ell+}(z)}\overline{r(\bar{z})}^{-1}\delta_{\ell+}^2(z) & i  \\
 0 & {\rm e}^{2i{t}\varphi_{\ell-}(z)}\overline{r(\bar{z})}^{-1}\delta_{\ell-}^2(z)
 \end{pmatrix}, & z\in \bigcup_{j=1}^{\ell } \bar{\gamma}_j ,
 \\
 {\rm e}^{i({t} {\Omega}_j+ {\Delta}_j)\sigma_3}, & z\in \gamma^{\rm gap}_j \cup \bar{\gamma}^{\rm gap}_j,
\\
\begin{pmatrix}
 {\rm e}^{i({t} {\Omega}_{\ell}+ {\Delta}_{\ell})} & 0 \\
 -ir(z)\delta_{\ell+}(z)\delta_{\ell-}(z){\rm e}^{i{t}(\varphi_{\ell+}(z)+\varphi_{\ell-}(z))} & {\rm e}^{-i({t}  {\Omega}_{\ell}+ {\Delta}_{\ell})}
 \end{pmatrix}, 
& z\in\gamma_{\alpha_{\ell}}^{\rm gap} \setminus \gamma^{\rm gap}_{\ell},
 \\
 \begin{pmatrix}
 {\rm e}^{i({t} {\Omega}_{\ell}+ {\Delta}_{\ell})} & i\overline{r(\bar{z})}\delta_{\ell+}^{-1}(z)\delta_{\ell-}^{-1}(z){\rm e}^{-i{t}(\varphi_{\ell+}(z)+\varphi_{\ell-}(z))} \\
 0 & {\rm e}^{-i({t} {\Omega}_{\ell}+ {\Delta}_{\ell})}
 \end{pmatrix}, &
z\in\bar{\gamma}_{\alpha_{\ell}}^{\rm gap} \setminus \bar{\gamma}^{\rm gap}_{\ell}.
 \end{cases}
 \end{equation}
 \item In the unmodulated regions the RHP is the same as above but setting $\alpha_\ell = \eta_{2\ell}$, independent of $\xi$. In these cases $\delta_\ell$ is independent of $\xi$  and the parameters $\Omega_j$ are affine functions of $\xi$. \hfill $\triangle$
 \end{enumerate}
 \end{RHP}
  In the following of this section we will write $\varphi$ instead of $\varphi_\ell$ for the modulated, or $\hat\varphi_\ell$ for the unmodulated case. With similar understanding we use $\delta(z)$ for the Szeg\"o function $\delta_{\ell}(z)$. 
    We recall that we assumed that $r(z)$ is  analytic in a neighbourhood of  $\bigcup_{j=1}^N(\gamma_j\cup\bar{\gamma}_j)$ and with the Schwarz symmetry $\overline{r(\bar{z}^{-1})}=r(z)$.
   Note the following matrix factorizations (the subscript $_\pm$ denotes boundary values at a point $z\in S^1$)
   \begin{equation}\label{juzhf}
      \begin{aligned}
         &\begin{pmatrix}
            {\rm e}^{-2i{t}\varphi_-(z)}r^{-1}\delta_-^{-2} & 0 \\
            -i & {\rm e}^{-2i{t}\varphi_+}r^{-1}\delta_+^{-2}
         \end{pmatrix}
         \\
         &= \begin{pmatrix}
            1 & i{\rm e}^{-2i{t}\varphi_-}r^{-1}\delta_-^{-2}  \\
            0 & 1
         \end{pmatrix}
         \begin{pmatrix}
            0 & -i \\
            -i & 0
         \end{pmatrix}
         \begin{pmatrix}
            1 & i{\rm e}^{-2i{t}\varphi_+}r^{-1}\delta_+^{-2}  \\
            0 & 1
         \end{pmatrix},
      \end{aligned}
   \end{equation}
   and similar factorization for $z$ in the lower half plane (with $r(z)\mapsto \overline r(\overline z)$).
 By employing the matrix factorizations (\ref{juzhf}), we will implement the procedure of ``opening  lenses'' which involves introducing the following transformation:
 \begin{equation}
  {T}(z)= {S}(z)\begin{cases}
 \begin{pmatrix}
 1 & \pm i{\rm e}^{-2i{t}\varphi (z)}r^{-1}(z)\delta ^{-2}(z)  \\
 0 & 1
 \end{pmatrix}, & z {\rm~within~} \Gamma_{\alpha_{\ell}\pm}\cup \bigcup_{j=1}^{\ell -1} \gamma_j  \cup \gamma_{\alpha_{\ell}}, \\
 \begin{pmatrix}
 1 & 0 \\
 \pm i{\rm e}^{2i{t}\varphi (z)}\overline{r(\bar{z})}^{-1}\delta ^2(z) & 1
 \end{pmatrix}, & z {\rm~within~} \bar{\Gamma}_{\alpha_{\ell}\pm}\cup \bigcup_{j=1}^{\ell -1} \bar{\gamma}_j \cup \bar{\gamma}_{\alpha_{\ell}}, \\
 \mathbb{I}, & {\rm elsewhere},
 \end{cases}
 \end{equation}
 \begin{figure}
 \centering
 \includegraphics[width=6cm]{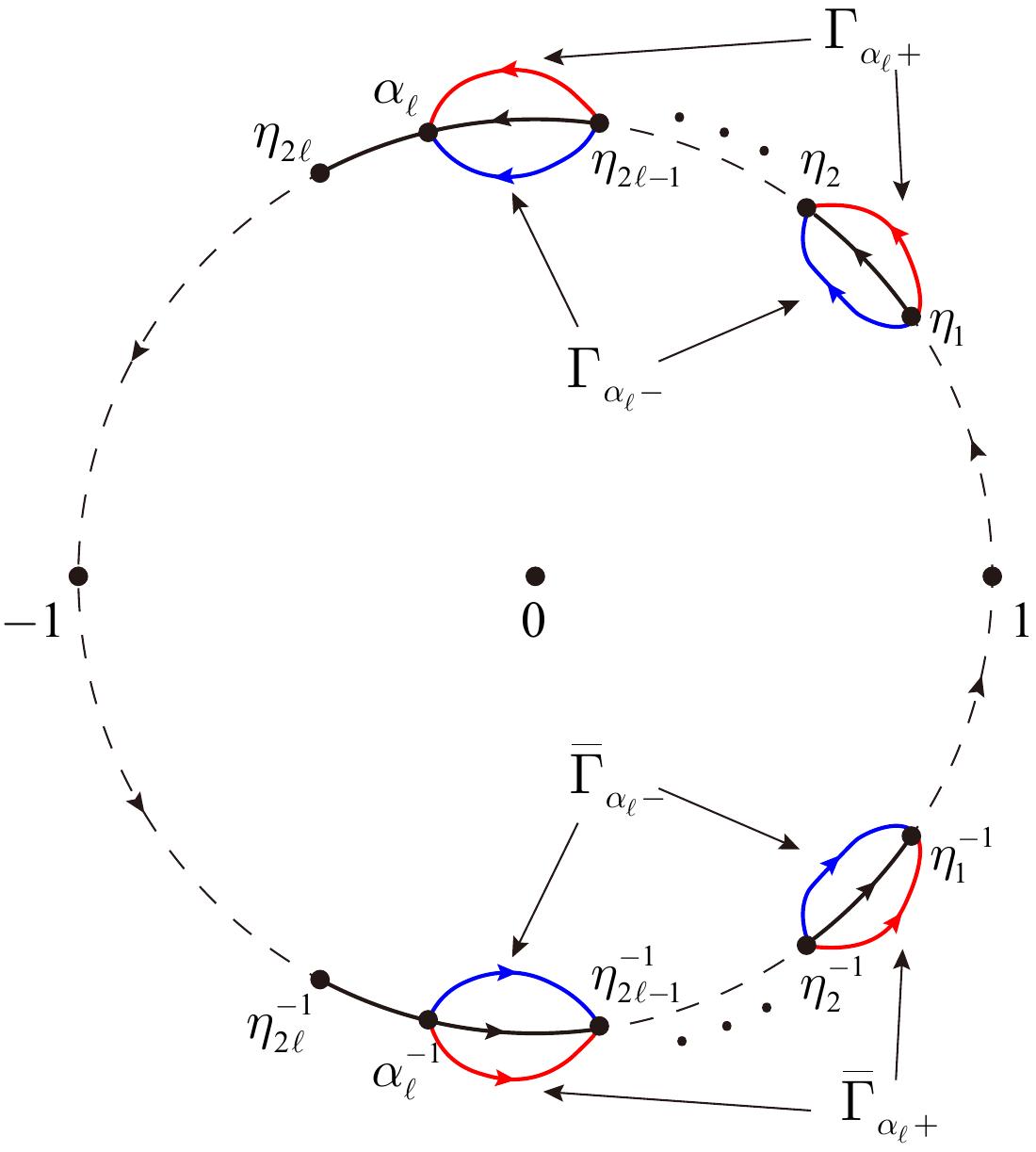}
 \caption{{\protect\small In the upper half-plane, the red (blue) curves from above (below) $\gamma_j$, $j=1, \cdots, \ell$, are denoted by $\Gamma_{\alpha_{\ell}+}$ ($\Gamma_{\alpha_{\ell}-}$). Symmetrically, in the lower half-plane, the red (blue) curves from below (above) $\bar{\gamma}_j$, $j=1, \cdots, \ell$, are denoted by $\bar{\Gamma}_{\alpha_{\ell}+}$ ($\bar{\Gamma}_{\alpha_{\ell}-}$).}}
 \label{open1}
 \end{figure}
 where the contours $\Gamma_{\alpha_{\ell}\pm}$ and $\bar{\Gamma}_{\alpha_{\ell}\mp}$ are illustrated in Figure \ref{open1}. Then for $j=1, \cdots, \ell,$ the jump matrices of $ {T}(z)$ are given by
 \begin{equation}
  {T}_+(z)= {T}_-(z)\begin{cases}
 \begin{pmatrix}
 1 &  i{\rm e}^{-2i{t}\varphi (z)}r^{-1}(z)\delta ^{-2}(z)  \\
 0 & 1
 \end{pmatrix}, & z \in \Gamma_{\alpha_{\ell}\pm},  \\
 \begin{pmatrix}
 1 & 0 \\
 -i{\rm e}^{2i{t}\varphi (z)}\overline{r(\bar{z})}^{-1}\delta ^2(z) & 1
 \end{pmatrix}, & z \in \bar{\Gamma}_{\alpha_{\ell}\pm},  \\
 -i\sigma_1, & z\in \bigcup_{j=1}^{\ell -1} \gamma_j  \cup \gamma_{\alpha_{\ell}}, 
 \\
 i\sigma_1, & z\in \bigcup_{j=1}^{\ell -1} \bar{\gamma}_j \cup \bar{\gamma}_{\alpha_{\ell}}, 
 \\
 {\rm e}^{i({t} {\Omega}_j+ {\Delta}_j)\sigma_3}, & z\in \gamma^{\rm gap}_j \cup \bar{\gamma}^{\rm gap}_j,
 \\
 \begin{pmatrix}
 {\rm e}^{i({t} {\Omega}_{\ell}+ {\Delta}_{\ell})} & 0 \\
 -ir(z)\delta_{+}(z)\delta_{-}(z){\rm e}^{i{t}(\varphi_{+}(z)+\varphi_{-}(z))} & {\rm e}^{-i({t} {\Omega}_{\ell}+ {\Delta}_{\ell})}
 \end{pmatrix}, & z\in \gamma_{\ell} \setminus \gamma_{\alpha_{\ell}},
 \\
 \begin{pmatrix}
 {\rm e}^{i({t} {\Omega}_{\ell}+ {\Delta}_{\ell})} & i\overline{r(\bar{z})}\delta_{+}^{-1}(z)\delta_{-}^{-1}(z){\rm e}^{-i{t}(\varphi_{+}(z)+\varphi_{-}(z))} \\
 0 & {\rm e}^{-i({t} {\Omega}_{\ell}+ {\Delta}_{\ell})}
 \end{pmatrix}, & z\in \bar{\gamma}_{\ell}\setminus\bar{\gamma}_{\alpha_{\ell}}.
 \end{cases}
 \end{equation}
   Lemma \ref{lem2} establishes that the jump matrices of $ {T}(z)$ exhibit exponential convergence to a simplified form, corresponding precisely to the jump matrices in the following modulated genus-$\ell$ model RHP.
   \begin{RHP}\label{RH5}\  Consider a $2\times 2$ matrix-valued function $Y(z;x,t)$ satisfying
      \begin{enumerate}
         \item Analyticity: $  {Y}(z)$ is analytic for
         $z\in\mathbb{C}\setminus(\bigcup_{j=1}^{ \ell -1 }(\gamma_j\cup\bar{\gamma}_j\cup\gamma^{\rm gap}_j\cup\bar{\gamma}^{\rm gap}_j)
         \bigcup ( \gamma_{\alpha_{\ell}}\cup\bar{\gamma}_{\alpha_{\ell}}\cup\gamma^{\rm gap}_{\alpha_{\ell}}\cup\bar{\gamma}^{\rm gap}_{\alpha_{\ell}} ) \cup \{0\} )$.
         \item Asymptotic behaviours: $  {Y}(z)=\mathbb{I}+\mathcal{O}(z^{-1})$ as $z\to\infty$ and $  {Y}(z)=\frac{\sigma_1}{z}+\mathcal{O}(1)$ as $z\to0$.
         \item Symmetry: $  {Y}(z)=\sigma_1\overline{ {Y}(\bar{z})}\sigma_1=z^{-1} {Y}(z^{-1})\sigma_1$.
         \item Jump matrices: for $j=1, \cdots, \ell$, it follows that
         \begin{equation}
             {Y}_+(z;x,t)= {Y}_-(z;x,t)\begin{cases}
               \begin{pmatrix}
                  0 & -i \\
                  -i & 0
               \end{pmatrix}, & z\in \bigcup_{j=1}^{\ell -1} \gamma_j  \cup \gamma_{\alpha_{\ell}},
               \\
               \begin{pmatrix}
                  0 & i  \\
                  i & 0
               \end{pmatrix}, & z\in \bigcup_{j=1}^{\ell -1} \bar{\gamma}_j \cup \bar{\gamma}_{\alpha_{\ell}},
               \\
               {\rm e}^{i({t} {\Omega}_j+ {\Delta}_j)\sigma_3}, & z\in \gamma^{\rm gap}_j \cup \bar{\gamma}^{\rm gap}_j,
               \\
               {\rm e}^{i({t} {\Omega}_{\ell}+ {\Delta}_{\ell})\sigma_3}, & z\in \gamma^{\rm gap}_{\alpha_{\ell}} \cup \bar{\gamma}^{\rm gap}_{\alpha_{\ell}}.
            \end{cases}
         \end{equation}
\item    In the $\ell$-th unmodulated region the formulation is the same but with $\alpha_\ell = \eta_{2\ell}$, $\gamma_{\alpha_\ell} = \gamma_\ell$ and $\gamma_{\alpha_\ell}^{\rm gap} = \gamma_\ell^{\rm gap}$.
      \end{enumerate}   \hfill $\triangle$
   \end{RHP}
   The resolution of RHP \ref{RH5} similarly necessitates addressing the particular symmetry of the RHP and the singularity at zero. This can be bypassed by transforming the $z$-plane back into $k$-plane (Figure \ref{N-DNLS6}) by using $ {F}(k)=\frac{1}{\sqrt{ 1-z(k)^{-2} }} {Y}(z(k))$. This gives a new RHP on the $k$-plane:
   \begin{RHP}\label{RH46}\ Consider a $2\times 2$ matrix-valued function $F(k;x,t)$ satisfying
      \begin{enumerate}
         \item Analyticity: $  {F}(k)$ is analytic for $k\in\mathbb{C}\setminus \left( \bigcup_{j=0}^{\ell-1} [\eta_{2j+1}^{\rm re},\eta_{2j}^{\rm re}] \cup [-1 , \alpha_{\ell}] \right)$ with jump conditions:
         \begin{equation}
             {F}_+(z)= {F}_-(z)\begin{cases}
               \begin{pmatrix}
                  0 & -i{\rm e}^{i({t} {\Omega}_j+ {\Delta}_j)} \\
                  -i{\rm e}^{-i({t} {\Omega}_j+ {\Delta}_j)} & 0
               \end{pmatrix}, & k\in (\eta_{2j+1}^{\rm re},\eta_{2j}^{\rm re}), \quad j=1,2,\cdots, \ell - 1,
               \\
               \begin{pmatrix}
                  0 & -i{\rm e}^{i({t} {\Omega}_{\ell}+ {\Delta}_{\ell})} \\
                  -i{\rm e}^{-i({t} {\Omega}_{\ell}+ {\Delta}_{\ell})} & 0
               \end{pmatrix}, & k\in (-1 , \alpha_{\ell}),
               \\
               \begin{pmatrix}
                  0 & -i \\
                  -i & 0
               \end{pmatrix}, & k\in (\eta_{1}^{\rm re},1).
            \end{cases}
         \end{equation}
         \item Asymptotic behaviours: $  {F}(k)=\mathbb{I}+\mathcal{O}(k^{-1})$ as $k \to\infty$.
         \item Symmetry conditions: $  {F}(k)=\sigma_1\overline{ {F}(\bar{k})}\sigma_1$.\hfill $\triangle$
      \end{enumerate}
   \end{RHP}
   \par 
   To express the result, we shall therefore establish the requisite notation and state the results directly, deferring detailed proof techniques to prior analogous arguments. 
   \par 
   We introduce the vectors
   \begin{equation}\label{ome2}
       {\bm{\Omega}} := \left(  {\Omega}_1, \cdots,  {\Omega}_{\ell} \right), \quad 
       {\bm{\Delta}} := \left(  {\Delta}_1, \cdots,  {\Delta}_{\ell} \right),
   \end{equation}
   with the $\Omega_{j}$'s defined as the $b$-periods \eqref{defOmega}. 
   To solve the RHP \ref{RH46}, we also define the suitable canonical homology basis cycles in Figure $\ref{jump3}$ and introduce the function (recall $\eta_0 = 1$)
   \begin{equation}
      \label{defcurlR}
      \mathcal{R}_{\ell}(k) := \sqrt{(k+1)(k-\alpha_{\ell}^{\rm re})\prod_{j=0}^{\ell - 1} (k - \eta_{2j}^{\rm re})(k - \eta_{2j+1}^{\rm re})},
   \end{equation} 
   with $\mathcal{R}_{\ell}(k) \sim k^{\ell + 1}$, as $k \to \infty$.  
   We define the Abel differentials of the first kinds as $ {\omega}^\ell_j := \frac{k^{j-1}}{\mathcal{R}_\ell(k)} dk$, $j = 1, 2, \cdots, \ell$ and the corresponding normalized Abelian differentials $\{\omega^\ell_j\}_{j=1}^{\ell}$, such that $\oint_{a_l} \omega^\ell_j = \delta_{lj}$. The $b$--periods of these normalized differentials  define the {\it normalized $b$-period matrix}  $\oint_{b_l} \omega^\ell_j = \tau^{[\ell]}_{lj}$ and the imaginary part of the matrix
   \begin{equation}\label{tauell}
      \bm{\tau}^{[\ell]} := \{\tau^{[\ell]}_{l j}\}_{\ell \times \ell}
   \end{equation}
   is positive definite.   
   We define the Abel-Jacobi map 
    \begin{equation}\label{Jell}
      \bm{J}_{\ell}(k) := \int_1^k (\omega_1^{\ell}, \omega_2^{\ell}, \cdots, \omega_{\ell}^{\ell})^T
   \end{equation}
   and the constant vector
   \begin{equation}\label{dell}
      \bm{d}_{\ell} =  \sum_{j=1}^{\ell} \bm{J}_{\ell}\left(p_{j}^{\ell}\right) - \frac{\ell}{2} \sum_{j=1}^{\ell} \bm{e}_{j}^{\ell} + \frac{1}{2} \sum_{j=1}^{N} \bm{\tau}_{j}^{[\ell]}
      = -\bm{J}_{\ell}(\infty_+) \quad \mod \{\bm{e}^{\ell}_j, \bm{\tau}^{[\ell]}_j\}_{j = 1,\cdots,\ell},
   \end{equation}
   where $\infty_+, p_1^{\ell}, p_2^{\ell}, \cdots , p_{\ell}^{\ell}$ are zeros of the meromorphic function 
             \begin{equation*}
      h_{\ell }(k) := \frac{k - \alpha^{\rm re}_{\ell}}{k + 1} \prod_{j=0}^{\ell - 1} \left( \frac{k-\eta_{2j}^{\rm re}}{k-\eta_{2j+1}^{\rm re}} \right)^{\frac{1}{2}}.
   \end{equation*} 
   Here $\bm{e}_{j}^{\ell} $ and $ \bm{\tau}_{j}^{[\ell]} $ are, respectively, the $j$-th column of the $ \ell \times \ell $ identity matrix and the $j$-th column of the matrix $ \bm{\tau}^{[\ell]} $. We get the explicit solution of RHP \ref{RH46}:
   \begin{equation}
   \label{solF}
      \begin{split}
          {F}(k) = & \begin{pmatrix}
            \frac{\Theta^{[\ell]}\left(\bm{J}_{\ell}(\infty_+) + \bm{d}_{\ell}; \bm{\tau}^{[\ell]}\right)}{\Theta^{[\ell]}\left(-\frac{ {\bm{\Omega}}+ {\bm{\Delta}}}{2\pi} + \bm{J}_{\ell}(\infty_+) + \bm{d}_{\ell}; \bm{\tau}^{[\ell]}\right)} & 0 \\
            0 & \frac{\Theta^{[\ell]}\left(-\bm{J}_{\ell}(\infty_+) - \bm{d}_{\ell}; \bm{\tau}^{[\ell]}\right)}{\Theta^{[\ell]}\left(-\frac{ {\bm{\Omega}}+ {\bm{\Delta}}}{2\pi} - \bm{J}_{\ell}(\infty_+) - \bm{d}_{\ell}; \bm{\tau}^{[\ell]}\right)}
         \end{pmatrix} \\
         & \times \begin{pmatrix}
            \frac{h_{\ell}(k) + h_{\ell}(k)^{-1}}{2} \frac{\Theta^{[\ell]}\left(-\frac{ {\bm{\Omega}}+ {\bm{\Delta}}}{2\pi} + \bm{J}_{\ell}(k) + \bm{d}_{\ell}; \bm{\tau}^{[\ell]}\right)}{\Theta^{[\ell]}\left(\bm{J}_{\ell}(k) + \bm{d}_{\ell}; \bm{\tau}^{[\ell]}\right)} 
            & 
            -\frac{h_{\ell}(k) - h_{\ell}(k)^{-1}}{2} \frac{\Theta^{[\ell]}\left(-\frac{ {\bm{\Omega}}+ {\bm{\Delta}}}{2\pi} - \bm{J}_{\ell}(k) + \bm{d}_{\ell}; \bm{\tau}^{[\ell]}\right)}{\Theta^{[\ell]}\left(-\bm{J}_{\ell}(k) + \bm{d}_{\ell}; \bm{\tau}^{[\ell]}\right)}
            \\[10pt]
            -\frac{h_{\ell}(k) - h_{\ell}(k)^{-1}}{2} \frac{\Theta^{[\ell]}\left(-\frac{ {\bm{\Omega}}+ {\bm{\Delta}}}{2\pi} + \bm{J}_{\ell}(k) - \bm{d}_{\ell}; \bm{\tau}^{[\ell]}\right)}{\Theta^{[\ell]}\left(\bm{J}_{\ell}(k) - \bm{d}_{\ell}; \bm{\tau}^{[\ell]}\right)}
            & 
            \frac{h_{\ell}(k) + h_{\ell}(k)^{-1}}{2} \frac{\Theta^{[\ell]}\left(-\frac{ {\bm{\Omega}}+ {\bm{\Delta}}}{2\pi} - \bm{J}_{\ell}(k) - \bm{d}_{\ell}; \bm{\tau}^{[\ell]}\right)}{\Theta^{[\ell]}\left(-\bm{J}_{\ell}(k) - \bm{d}_{\ell}; \bm{\tau}^{[\ell]}\right)}
         \end{pmatrix},
      \end{split}
   \end{equation}
   where $\Theta^{[\ell]}(\bm{z} ; \bm{\tau}^{[\ell]})  $ denotes the $\ell$-dimensional Riemann theta function corresponding to the hyper-elliptic Riemann surface (of genus $g=\ell$) of the algebraic function $\mathcal R_\ell(k)$ given by \eqref{defcurlR}. 
   \par 
   Now we are ready to prove Theorem \ref{thm4}.
   \begin{figure}
      \centering
      \includegraphics[width=9cm]{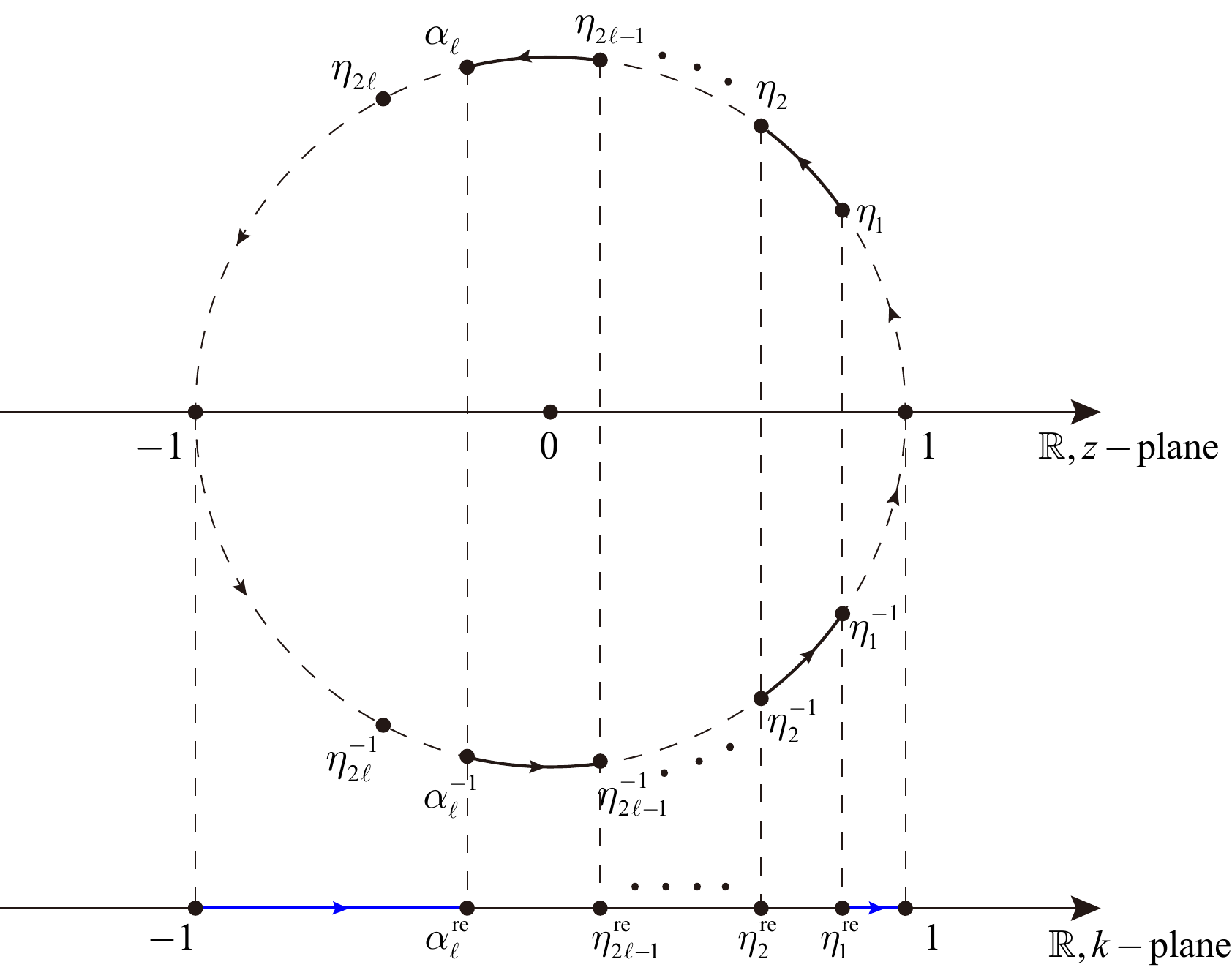}
      \caption{{\protect\small Schematic diagram of the jump line transformation from the $z$-plane to the $k$-plane in modulated genus $\ell$ region. The gap (band) on the $z$-plane becomes the band (gap) on the $k$-plane, with the blue segments representing the bands on the $k$-plane.}}
      \label{N-DNLS6}
   \end{figure}
   \begin{figure}
      \centering
      \includegraphics[width=9cm]{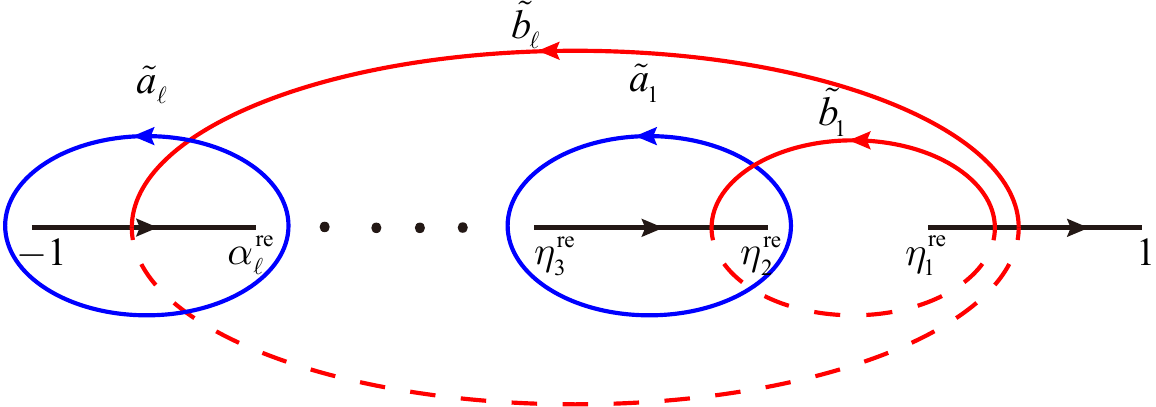}
      \caption{{\protect\small The canonical homology basis $\{\tilde{a}_j, \tilde{b}_j\}_{j=1}^{\ell}$ for a hyper-elliptic Riemann surface of genus $\ell$, defined as a two-sheeted covering of the complex $k$-plane branched over $2\ell+2$ points.   }}
      \label{jump3}
   \end{figure}    
   Recalling the transformation \eqref{tsf}, we have 
   \begin{equation}
      M(z;x,t)=\delta_{\ell \infty}^{\sigma_3}{\rm e}^{ig_{\ell\infty}\sigma_3}{E}(z) {P}(z){\rm e}^{-ig_{\ell}(z)\sigma_3}\delta_{\ell}(z)^{-\sigma_3},
   \end{equation}
   where $ {E}(z)$ and $ {P}(z)$ are defined by \eqref{te} and \eqref{tp}, respectively.
   \par 
   Letting $z\to\infty$, we obtain
   \begin{equation}
      \begin{aligned}
         M(z;x,t)& =  \delta_{\ell \infty}^{\sigma_3} {\rm e}^{ig_{\ell\infty}\sigma_3} \left(\mathbb{I}+ \frac{ {E}^{(1)}}{z}+ \mathcal{O}(z^{-2}) \right)   \left(\mathbb{I}+ \frac{ {Y}^{(1)}}{z}+ \mathcal{O}(z^{-2}) \right) \\
         & \qquad \left(\mathbb{I}- \frac{ig_{\ell}^{(1)}\sigma_3}{z} +\mathcal{O}(z^{-2}) \right) \left(\mathbb{I}- \frac{\delta_{\ell}^{(1)}\sigma_3}{z} +\mathcal{O}(z^{-2}) \right)
         {\rm e}^{-ig_{\ell\infty}\sigma_3}\delta_{\ell \infty}^{-\sigma_3} \\
         & = \mathbb{I} + {z}^{-1}  \delta_{\ell \infty}^{\sigma_3} {\rm e}^{ig_{\ell\infty}\sigma_3}  \left(  {E}^{(1)} +  {Y}^{(1)} - ig_{\ell}^{(1)}\sigma_3 - \delta_{\ell}^{(1)}\sigma_3  \right) {\rm e}^{-ig_{\ell\infty}\sigma_3}\delta_{\ell \infty}^{-\sigma_3} + \mathcal{O}(z^{-2}),
      \end{aligned}
   \end{equation}
   where $ {E}^{(1)}$ is defined by \eqref{tE1}. 
   Using reconstruction formula \eqref{rec} and the transformation $ {F}(k)=\frac{1}{\sqrt{ 1-z(k)^{-2} }} {Y}(z(k))$ and solution of $  {F}(k) $ \eqref{solF}, we obtain the expression \eqref{so2}. Theorem \ref{thm5} can be proved in a completely analogous manner by setting $\alpha_\ell = \eta_{2\ell}$.

\section{Large-space asymptotics of the potential\ $q(x,0)$} \label{section33}
Many  of the {\it accoutrements} needed for the larg-space analysis have been already introduced in the section for the long-time analysis.
   This section considers the large-space behaviour of the initial potential $q(x,0)$. 
   
   \subsection{Behaviour of initial potential $q(x,0)$ as $x\to+\infty$.} \label{qx0+}
   When $x>0$, we have the following estimate for the jump matrices on $\bigcup_{j=1}^N(\gamma_j\cup\bar{\gamma}_j),$ i.e.,
   \begin{equation}
      \norm{J_M(z)-\mathbb{I}}= \mathcal{O}\left({\rm e}^{-2 x  \min \left\{\im{\eta_1}, \im{\eta_{2N}} \right\}}\right).
   \end{equation}
   By a standard small-norm argument, we obtain that 
   \begin{equation}
      M(z)= \mathbb{I} + \frac{\sigma_1}{z} + \mathcal{O}\left({\rm e}^{-2 x  \min \left\{\im{\eta_1}, \im{\eta_{2N}} \right\} }\right), \quad \mathrm{as} ~~~ x \to +\infty.
   \end{equation}
   Using the reconstruction formulae (\ref{rec}),  we conclude that as $x\to+\infty$, it follows that
   \begin{equation}\label{q+i}
      q(x,0)=1+\mathcal{O}({\rm e}^{-cx}),
   \end{equation}
   where $c\in\mathbb{R}^+$ is a fixed constant. This is consistent with the nonzero background and leads to item (1) in Theorem \ref{Theorem1}.

   \subsection{Behaviour of initial potential $q(x,0)$ as $x\to-\infty$.}\label{sec3}
   As $x\to-\infty$, the entries of jump matrices \eqref{jump2} grow exponentially. This implies that the leading term of $q(x,0)$ is not as simple as (\ref{q+i}), and the jump matrices on $\bigcup_{j=1}^N(\gamma_j\cup\bar{\gamma}_j)$ contribute to the leading term. Following the general principle of the Deift-Zhou  method \cite{Deift-Zhou1993}, we start a chain of transformations of the RHP \ref{RHP2}, the last of which is amenable to a perturbative analysis. 
   The first transformation is \eqref{tsf} where now, however, the $g$ function is simply the leading behaviour as $\xi\to-\infty$ of the previous one, namely the whole construction proceeds with the same steps and formulas as if in the last unbounded unmodulated region $(-\infty, \xi_{2N})$ with the only change being that $\varphi_{N}$ and the corresponding $g$ functions are replaced by 
\begin{align}
   \varphi(z) &= \lim_{\xi\to-\infty} \frac 1 {\xi} {\hat \varphi}_{2N}(z;\xi) = {\hat \varphi}_{2N}^{(1)},\\
   g(z) &={ \hat \varphi}_{2N}^{(1)} - \frac 1 2 \le(z - \frac 1 z\ri).
   \label{gdeflargex}
\end{align}
Then all the formulas proceed identically as for the previous cases with the understanding that $\ell = N$, $\alpha_\ell \equiv \eta_{2N}$ and $\bm \Omega  =\bm \Omega^{(1)}$.  

    \section*{Acknowledgments}
    The work of D.S.W was supported by the National Natural Science Foundation of China (Grant Nos. 12371247 and 12431008) and Beijing Natural Science Foundation Grant No. 1262012. The work of M. B. was supported in part by the Natural Sciences and Engineering Research Council of Canada (NSERC) grant RGPIN-2023-04747. Part of the work was completed while in residence as Royal Society Wolfson Visiting Fellow at the School of Mathematics in Bristol University. P. Y. acknowledges the scholarship from the Institut des sciences mathématiques (ISM).

\appendix

\section{Case analysis: the genus 1 dark soliton gas.}
   \label{section55}
   In this section, we will consider the specific case of genus-$1$. In this setting we can express  the parameters in terms of classical functions. In particular, we will propose an alternative instructive approach for solving RHP \ref{RH5} (with $N=1$ and $\alpha_1= \eta_{2}$) of the genus-$1$ model whereby we  construct the solution directly on the  $z$-plane by preserving the symmetries and zero singularities of the RHP.
   
   \subsection{An alternative approach to the genus 1 model RHP \ref{RH5}}\label{subg1}
   In this subsection, we will directly construct the solution by carrying over the zero singularities from RH Problem \ref{RH5}.  To begin with, consider the following genus $1$ RHP by letting $N = 1$ in \ref{RH5}:
   \begin{RHP}\label{rh24}\  Consider a $2\times 2$ matrix-valued function $Y(z;x,t)$ satisfying
      \begin{enumerate}
         \item Analyticity: $ Y(z)$ is analytic for $z\in\mathbb{C}\setminus( \gamma_1\cup\bar{\gamma}_1 \cup \{0\})$ with jump conditions:
         \begin{equation}\label{JUmT}
            Y_+(z)=Y_-(z)\begin{cases}
               \begin{pmatrix}
                  0 & -i \\
                  -i & 0
               \end{pmatrix}, & z\in  \gamma_1,
               \\
               \begin{pmatrix}
                  0 & i \\
                  i & 0
               \end{pmatrix}, & z\in   \bar{\gamma}_1,
               \\
               e^{i(\Omega+\Delta)\sigma_3}, & z\in \gamma^{\rm gap}_1 \cup \bar{\gamma}^{\rm gap}_1.
            \end{cases}
         \end{equation}
         \item Asymptotic behaviours: $ Y(z)=\mathbb{I}+\mathcal{O}(z^{-1})$ as $z\to\infty$ and $ Y(z)=\frac{\sigma_1}{z}+\mathcal{O}(1)$ as $z\to0$.
         \item Symmetry conditions: $ Y(z)=\sigma_1\overline{Y(\bar{z})}\sigma_1=z^{-1}Y(z^{-1})\sigma_1$.\hfill $\triangle$
      \end{enumerate}
   \end{RHP}
  
   For $ N = 1 $ case, we have a genus-1 Riemann surface $\mathcal{S}$ as the two-sheeted covering by 
   \begin{equation}
      \mathcal{S}=\{(z,w)\in\mathbb{C}^2:w^2=R^2(z)=(z-\eta_1)(z-\eta_1^{-1})(z-\alpha_1)(z-\alpha_1^{-1})\}.
   \end{equation}
The first sheet is defined by the asymptotic behaviour $R(z)\sim  z^2$, $z\to\infty$.
   Here $\alpha_1$ is to be determined by the modulation equations in the modulated region and $\alpha_1=\eta_2$ in the unmodulated region.

   A key difficulty in solving this model RHP \ref{rh24} on $z$-plane arises from the second symmetric condition and the singularity at $z=0$. In what follows, we will construct a quasi-periodic solution satisfying these two conditions to address this problem.
   
   Let $a$ and $b$ form the canonical homology basis cycles (Figure \ref{jump}).
   \begin{figure}
      \centering  
      \subfigbottomskip=10pt 
      \subfigcapskip=-5pt 
      \subfigure[The red contour denotes path $a$, and the blue contour denotes path $b$.]{
         \includegraphics[width=6.8cm]{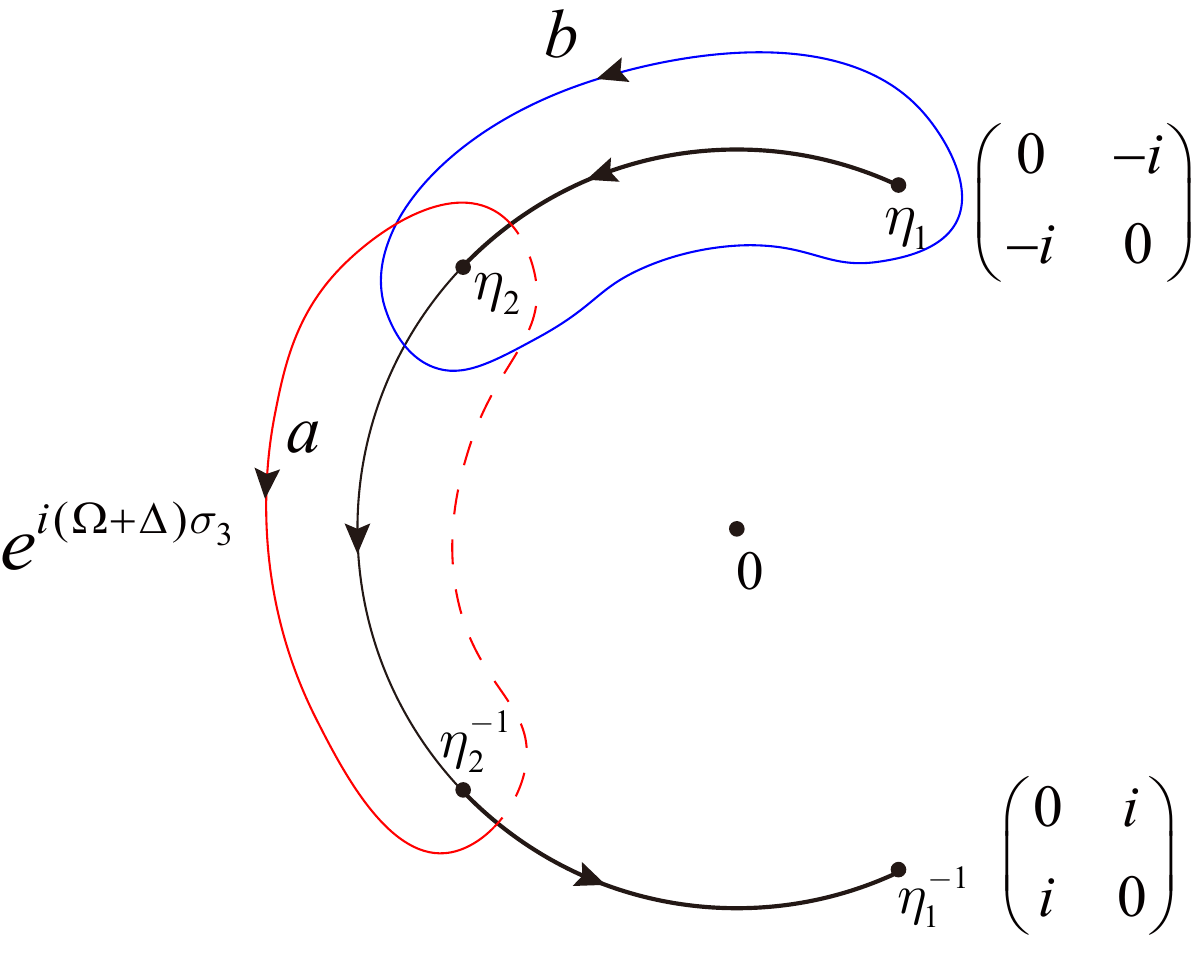}\label{jump}}
      \quad\quad\qquad
      \subfigure[The red contour denotes path $\tilde{a}$, and the blue contour denotes path $\tilde{b}$.]{
         \includegraphics[width=4.6cm]{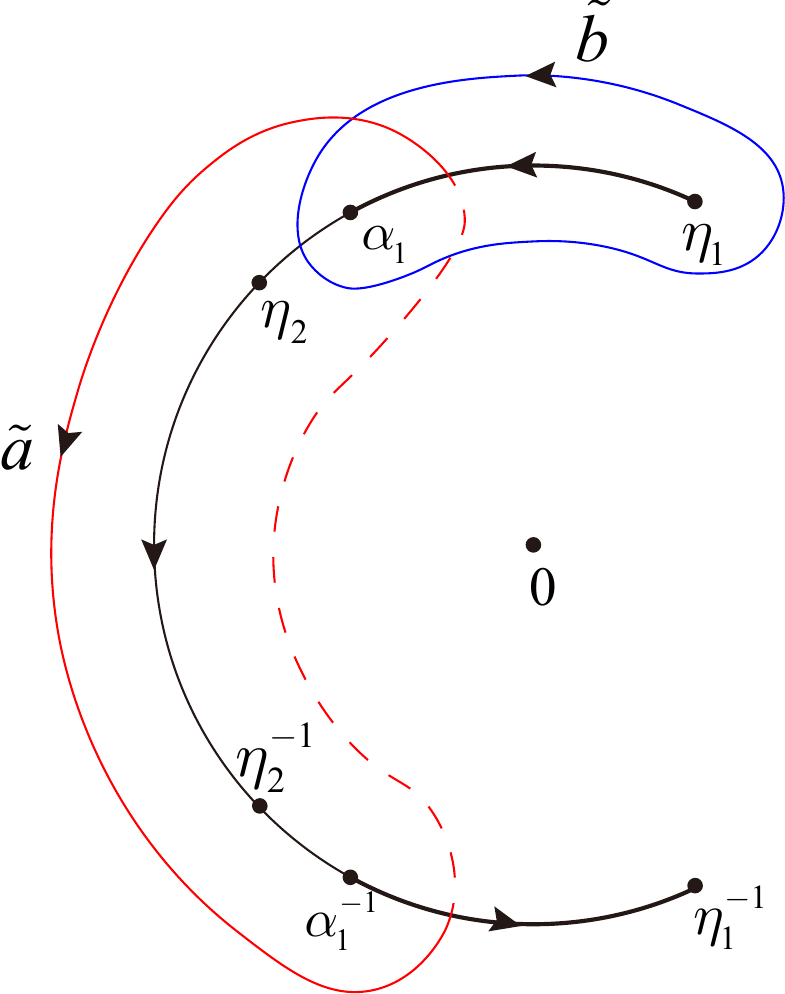}\label{open3}}
      \caption{The canonical homology basis cycles for genus $1$ base on the $z$-plane.}
      \label{signc1}
   \end{figure}

   Define the Abel differential of the first kind
   \begin{equation}\label{mj0}
      \omega=\frac{1}{2\omega_1 R(z)}\d z , \qquad  \omega_1 = \int_{\alpha_1}^{\alpha_{1}^{-1}} \frac {\d z}{R(z)} =  \int_{\eta_1}^{\eta_{1}^{-1}} \frac {\d z}{R(z)},\ \ \ \omega_3:= \int_{\eta_1}^{\alpha_1} \frac {\d z}{R(z_+)}
   \end{equation}
   such that 
   \begin{equation}\label{tau}
      \oint_{a}\omega=1,\quad \tau:= \frac {\omega_3}{\omega_1} = \frac{iK(\sqrt{1-m^2})}{K(m)}  \in i\mathbb{R}^+.
   \end{equation}
  The understanding is that the above formul\ae\  work also in the unmodulated region simply by replacing $\alpha_1$ with $\eta_2$.
   Consider the Abel map restricted to $z\in\mathbb{C}\setminus(\gamma_1\cup\bar{\gamma}_1\cup\gamma^{\rm gap}_1\cup\bar{\gamma}^{\rm gap}_1)$
   \begin{equation}
   \label{J(z)}
      J(z)=\int_{\eta_1}^{z}\omega,
   \end{equation}
   which satisfies the jump conditions
   \begin{equation}
      \begin{aligned}
         J_+(z)+J_-(z)&=0,\ \   z\in\gamma_1; \qquad
         J_+(z)-J_-(z)=-\tau, \ \  z\in\tilde{\gamma}_1, \\
         J_+(z)+J_-(z)&=1, \ \  z\in \gamma^{\rm gap}_1\cup\bar{\gamma}^{\rm gap}_1,
      \end{aligned}
   \end{equation}
   and the symmetry condition
   \begin{equation}\label{sym12}
      J(z)+J(z^{-1})=\frac{1}{2},\quad z\in\mathbb{C}\setminus(\gamma_1\cup\bar{\gamma}_1\cup\gamma^{\rm gap}_1\cup\bar{\gamma}^{\rm gap}_1).
   \end{equation}
   Define the function 
   \begin{equation}
      h(z)=\left[(z-\alpha_1)(z-\overline{\alpha_1})(z-\eta_1)(z-\overline{\eta_1})\right]^{-\frac{1}{4}},
   \end{equation}
   so that it is analytic at infinity and behaves like $\frac{1}{z}$, and it satisfies the jump conditions
   \begin{equation}
      \begin{aligned}
         h_{+}(z)&=ih_{-}(z),\ \  z\in \gamma_1; \qquad
         h_{+}(z)=-h_{-}(z), \ \  z\in \tilde{\gamma}_1, \\
         h_{+}(z)&=-ih_{-}(z), \ \  z\in \gamma^{\rm gap}_1\cup\bar{\gamma}^{\rm gap}_1.
      \end{aligned}
   \end{equation} 
   Recall the $1$-dimensional Riemann theta function $\Theta^{[1]}(z;\tau)$  satisfies
   \begin{equation}\label{mn}
      \Theta^{[1]}(z+m+n\tau;\tau)=e^{-2\pi inz-\pi in^2\tau}\Theta^{[1]}(z;\tau),\quad m,n\in\mathbb{Z},
   \end{equation}
   and $\{\frac{1+\tau}{2}+m+n\tau\}_{m,n\in\mathbb{Z}}$ is the set of all zeros.  
 The Riemann theta function is related to the more standard Jacobi $\theta_3$ function by $\Theta^{[1]}(z;\tau) = \theta_3(\pi z;\tau)$, see \href{https://dlmf.nist.gov/20.2.E3}{DLMF 20.2.3}.

   Define the $2\times2$ matrix valued function $H(z)$ whose entries are given by
   \begin{align}
      H_{11}(z)&=h(z)\frac{\Theta^{[1]}(J(z)-J(\infty)-\frac{1+\tau}{2}+W_{\infty};\tau)\Theta^{[1]}(J(z)+J(0)-\frac{1+\tau}{2}+W_{0};\tau)}
      {\Theta^{[1]}(J(z)-J(\infty)-\frac{1+\tau}{2};\tau)\Theta^{[1]}(J(z)+J(0)-\frac{1+\tau}{2};\tau)},  \label{F11}
      \\
      H_{12}(z)&=-h(z)\frac{\Theta^{[1]}(-J(z)-J(\infty)-\frac{1+\tau}{2}+W_{\infty};\tau)\Theta^{[1]}(-J(z)+J(0)-\frac{1+\tau}{2}+W_{0};\tau)}
      {\Theta^{[1]}(-J(z)-J(\infty)-\frac{1+\tau}{2};\tau)\Theta^{[1]}(-J(z)+J(0)-\frac{1+\tau}{2};\tau)}, \label{F12}
      \\
      H_{21}(z)&=h(z)\frac{\Theta^{[1]}(J(z)+J(\infty)-\frac{1+\tau}{2}+W_{\infty};\tau)\Theta^{[1]}(J(z)-J(0)-\frac{1+\tau}{2}+W_{0};\tau)}
      {\Theta^{[1]}(J(z)+J(\infty)-\frac{1+\tau}{2};\tau)\Theta^{[1]}(J(z)-J(0)-\frac{1+\tau}{2};\tau)}, \label{F21}
      \\
      H_{22}(z)&=-h(z)\frac{\Theta^{[1]}(-J(z)+J(\infty)-\frac{1+\tau}{2}+W_{\infty};\tau)\Theta^{[1]}(-J(z)-J(0)-\frac{1+\tau}{2}+W_{0};\tau)}
      {\Theta^{[1]}(-J(z)+J(\infty)-\frac{1+\tau}{2};\tau)\Theta^{[1]}(-J(z)-J(0)-\frac{1+\tau}{2};\tau)}, \label{F22}
   \end{align}
   where
   \begin{equation}\label{WW}
      W_{\infty}=\frac{\lambda \Omega+\Delta-\pi}{4\pi}+\frac{\tau}{2},\quad W_{0}=\frac{\lambda \Omega+\Delta-\pi}{4\pi}-\frac{\tau}{2},
   \end{equation}
   which are uniquely determined by
$W_{\infty}-W_0=\tau$ and $W_{\infty}+W_0=\frac{\lambda\Omega+\Delta-\pi}{2\pi}.$ Here $\lambda$ is the large parameter: 
\begin{itemize}
\item for the {\bf long-time behaviour} $\lambda = t>0$ and 
\begin{align}
\label{tOmega}
     {\Omega}&=
     \frac{-2i\pi}{\omega_1}\left[\xi-(\re\eta_1+\re\alpha_1)\right],
\qquad 
{\Delta}=
   \frac {2}{i\omega_1}
   \int_{\eta_1}^{\alpha_1}\frac{\ln r(s)}{R_{1+}(s)}ds.
\end{align}
\item for the {\bf large-space behaviour} $x\to-\infty$ then $\lambda = x$ and
   \begin{equation} \label{Ome}
      \Omega=
       \frac{- 2\pi i}{\omega_1}
       ,  \qquad
      \Delta
      =
   \frac {2}{i\omega_1}   \int_{\eta_1}^{\eta_2}\frac{\ln r(s)}{R(s_+)}\d s
    \in\mathbb{R}.
   \end{equation}

\end{itemize}  
%

   It is remarkable that the above condition guarantees not only the symmetry $H(z)=z^{-1}H(z^{-1})\sigma_1$, but also that $H(z)$ satisfies the jump condition \eqref{JUmT}. To examine the solvability of the model RHP, let us investigate the asymptotic behaviours for $z=\infty$ and $z=0$: expanding $H(z)$ and noting that $J(0)+J(\infty)=\frac{1}{2}$ by (\ref{sym12}), we have
   \begin{equation}\label{inf1}
      \lim_{z\to\infty} H(z) =  \frac{2\omega_1 \Theta^{[1]}(-\frac{1+\tau}{2}+W_{\infty};\tau)\Theta^{[1]}(-\frac{\tau}{2}+W_0;\tau)}
      {\Theta^{[1]}(\frac{\tau}{2};\tau)\Theta^{[1]'}(\frac{1+\tau}{2};\tau)} \mathbb{I},
   \end{equation}
   \begin{equation}\label{01}
      \lim_{z\to 0 } z H(z)= \frac{2\omega_1\Theta^{[1]}(-\frac{1+\tau}{2}+W_{\infty};\tau)\Theta^{[1]}(-\frac{\tau}{2}+W_0;\tau)}
      {\Theta^{[1]}(\frac{\tau}{2};\tau)\Theta^{[1]'}(\frac{1+\tau}{2};\tau)} \sigma_1,
   \end{equation}
   which is consistent with the symmetry $H(z)=z^{-1}H(z^{-1})\sigma_1$. It is concluded that $H(z)$ is indeed a solution, since
   \begin{equation}
      \frac{2\omega_1\Theta^{[1]}(-\frac{1+\tau}{2}+W_{\infty};\tau)\Theta^{[1]}(-\frac{\tau}{2}+W_0;\tau)}
      {\Theta^{[1]}(\frac{\tau}{2};\tau)\Theta^{[1]'}(\frac{1+\tau}{2};\tau)}
      =
      \frac{2\omega_1\Theta^{[1]}(\frac{\Omega+\Delta-3\pi}{4\pi};\tau)\Theta^{[1]}(\frac{\Omega+\Delta-\pi}{4\pi}-\tau;\tau)}
      {\Theta^{[1]}(\frac{\tau}{2};\tau)\Theta^{[1]'}(\frac{1+\tau}{2};\tau)}
      \neq 0,
   \end{equation}
   since $\Omega$ and $\Delta$ are real. After normalization by the above nonzero constant, which keeps the jump conditions and symmetry, we arrive at the solution of RHP \ref{rh24}. Therefore, we obtain 
   \begin{equation} \label{Yz}
      Y(z)=\frac{\Theta^{[1]'}(\frac{1+\tau}{2};\tau) \Theta^{[1]}(\frac{\tau}{2};\tau)}
      {2\omega_1\Theta^{[1]}(-\frac{1+\tau}{2}+W_{\infty};\tau)\Theta^{[1]}(-\frac{\tau}{2}+W_0;\tau)}H(z).
   \end{equation}
   We are ready to prove Theorem \ref{thm12}.
   \begin{proof}
     As in the general case the proof proceeds by inverting all transformations
      \begin{equation}
         M(z)= \delta_{\infty}^{\sigma_3} e^{ig_{\infty}\sigma_3} E(z) P(z) e^{-ig(z)\sigma_3}\delta(z)^{-\sigma_3},  \\
      \end{equation}
      where $\delta(z)$, $g(z)$ $E(z)$ and $P(z)$ are defined by \eqref{deltas} (with $N=1$, $\alpha_1 = \eta_2$), \eqref{gdeflargex},  \eqref{Dee} and \eqref{Pdef} for $N = 1$, respectively. $g_{\infty} := g(\infty)$ and $\delta_{\infty} := \delta(\infty)$.
The reconstruction formula (\ref{rec}) requires to compute the large $z$ expansion of the solution of the RHP \ref{rh24}. We do not provide the details of the straightforward computations and only some interesting tidbits.

The differential $\d\varphi$ can be pulled back to the $k$-plane ($k=\frac 1 2(z+z^{-1})$) and yields
\be
\d \varphi = \frac {k \le(k-C\ri)}{\mathcal R(k)} \d k + \frac {A\d k}{\mathcal R(k)}, \ \ \ C = \frac{\left(\eta_{2} +\eta_{1} \right) \left(\eta_{1} \eta_{2} +1\right)}{4 \eta_{1} \eta_{2}},
\ee
where, as before
\be
\mathcal R(k) := \sqrt{ (k^2-1)(k-\eta_1^{\re})(k-\eta_2^{\re})}.
\ee
The constant $A$ is determined by the condition that $\int_1^{\eta_1^{\re}} \d\varphi=0$. In terms of the elliptic functions on the elliptic surface of $\mathcal R$ , $\d\varphi$ is a second kind differential with two double poles at the points above $k=\infty$. 
Then, adding an appropriate exact differential we can express $\d\varphi$ as  differential with a single double pole at $k=1$ as follows
\be
\d\varphi = \d \le(\frac{ \mathcal R(k) }{k-1}\ri)  
 +  \frac {H^2\d k}{(k-1) \mathcal R(k)} + \frac {\tilde A\d k}{\mathcal R(k)}, \qquad 
 H^2:=  \frac {(\eta_1-1)^2(\eta_2-1)^2}{2\eta_1\eta_2}.
 \label{A27}
\ee
Here the constant $H$ (and the determination in the root) is determined so as to remove the singularity at $k=1$ and it is such that $\mathcal R(k) =  {H}\sqrt{k-1}(1 + \mathcal O(k-1))$.
In terms of the Abel map
\be
 s= \int_1^k \frac{\d\kappa}{\mathcal R(\kappa)}=\frac 2 H \sqrt{k-1} (1+\mathcal O(k-1)),
\ee
  one thus concludes
\be
\varphi =  \le(\frac{ \mathcal R(k) }{k-1}\ri)   - 
\frac{\pi} {\omega_1}  \le(
\frac {\theta_1'(\frac {\pi s}{2\omega_1},2\tau)}{\theta_1(\frac {\pi s}{2\omega_1} ,2\tau)}
-\frac {\theta_1'(\tau,2\tau)}{\theta_1(\tau ,2\tau)}\ri)
 ,
\ee
where we have used that the modular parameter of the elliptic surface of $\mathcal R$ is $\tau^{[1]}  = 2\tau$ and that $\varphi(-1)=0$. Indeed, the reader may verify that  the $a$--period of the holomorphic differential $\frac {\d k}{\mathcal R(k)}$ is the same $2\omega_1$ while the $b$--period is doubled.
From this we find the expansion of $g(z) = \varphi -\frac 1 2 \le(z - \frac 1 z \ri)$ near $z=\infty $, i.e., $k=\infty$  and recalling that $s-s_\infty = -\frac 1 k + \mathcal O(k^{-2})$,  $k\simeq \frac  z 2 + \mathcal O(1)$ as $z\to\infty$ ,
\begin{align}
g(z) &= g_\infty +\frac {g^{(1)}} z + \mathcal O(z^{-2}),\\
g^{(1)} &:=   \frac {\pi^2}{\omega_1^2} \le(\frac {\theta_1'}{\theta_1}\ri)'\le(\frac{\pi s_\infty}{2\omega_1}; 2\tau\ri) 
+
\re \le  ( \eta_1+\eta_2+ \frac{(\eta_1-\eta_2)^2}8 + \frac {\eta_1}{4\eta_2} + \frac 7 4  \ri).
\end{align}
The Weierstrass uniformization $X = \wp(s;\omega_1;2\omega_3)$ corresponds to the birational map between the $k$--plane and $X$--plane given by 
\be
X = \frac {H^2}{4(k-1)} + X_\infty,
\ee
where $H$ is in \eqref{A27} and $X_\infty$ is determined by the condition  that $X(-1)+X(\eta_1^{\re}) + X(\eta_2^{\re})=0$, i.e., 
\begin{align}
\label{Xinfty}
X_\infty:= \frac{\eta_{1}^{2} \eta_{2}^{2}-6 \eta_{1}^{2} \eta_{2} -6 \eta_{1} \eta_{2}^{2}+\eta_{1}^{2}+20 \eta_{1} \eta_{2} +\eta_{2}^{2}-6 \eta_{1} -6 \eta_{2} +1}{48 \eta_{1} \eta_{2}},\\
e_1 = X(\eta_1^{\re}) , \ \ e_2= X(\eta_2^{\re}), \ \ e_3 = X(-1).
\end{align}

Using \href{https://dlmf.nist.gov/23.6.E14}{DLMF 23.6.14}, we have
\be
X = \wp\left(\frac {2\omega_1 v}{\pi} ;\omega_1,2\omega_3\right)=\left(\frac{\pi}{2\omega_{1}}\right)^{2}\left(\frac{\theta_{%
1}'''}{3\theta_{1}'}-\frac{{\mathrm{d}}^{2}}{{%
\mathrm{d}v}^{2}}\ln\theta_{1}\left(v;2\tau\right)\right),
\ee
where all the theta constants are for $\tau^{[1]}  = 2\tau$, then
one can rewrite $g^{(1)}$ as 
\be
g^{(1)} := - 4 X_\infty   +  \frac {\pi^2}{\omega_1^2} \frac {\theta_1'''}{3\theta_1'} + 
\re \le  ( \eta_1+\eta_2+ \frac{(\eta_1-\eta_2)^2}8 + \frac {\eta_1}{4\eta_2} + \frac 7 4  \ri).
\ee

\paragraph{Large space case.} 

Recall the amplitude of the dark soliton gas potential $q(x,t)$ can be reconstructed by the formulae
   \begin{equation}\label{recons2}
      |q(x,t)|^2=1-i \partial_x \left(\lim_{z\to\infty} z  M_{11}(z;x,t)\right).
   \end{equation}
Noting that $\delta^{(1)}, g^{(1)}$ are constants, we eventually find (recalling $\Theta^{[1]}(z) = \theta_3(\pi z)$) 
      \begin{equation} \label{qhalf}
         \begin{aligned}
            |q(x,0)|^2 &= 1 + i \frac{2\pi }{\omega_1\Omega} \left[ \partial_x^2 \ln\Theta^{[1]}\left(-\frac{1+\tau}{2}+W_{\infty};\tau\right) + \partial_x^2 \ln\Theta^{[1]}\left(-\frac{\tau}{2}+W_{0};\tau\right)  \right] - g^{(1)} + \mathcal{O}(x^{-1}) \\
            & = 1-g^{(1)}- \left[ 
            \partial_x^2 \ln\theta_1\left(\frac{x\Omega+\Delta}{4}+ \frac {\pi \tau}2 - \frac \pi 4\right) 
            +
            \partial_x^2 \ln\theta_1\left(\frac{x\Omega+\Delta}{4}+\frac { \pi\tau}2 + \frac{\pi }{4}\right)  \right]
         + \mathcal{O}(x^{-1})\cr
         &=  1-g^{(1)}- \left[ 
            \partial_x^2 \ln\theta_1\left(\frac{x\Omega+\Delta}{2}+{ \pi \tau} + \frac{\pi }{2};2\tau\right)  \right]
         + \mathcal{O}(x^{-1}),
         \end{aligned}
       \end{equation}
      where in the second equality we used \eqref{Ome}, \eqref{WW} and the duplication relations (Landen) for $\theta_j$'s.

Using \href{https://dlmf.nist.gov/23.6.E14}{DLMF 23.6.14} as above, and using \eqref{Ome},
we find (setting for convenience $S:=\omega_1 \frac{x\Omega+\Delta}{2\pi} + \omega_1 2\tau +  {\omega_1}$)
\begin{align}
|q(x,0)|^2 &= 1-g^{(1)}  
   +  \frac {\Omega^2}{4}   
\le(   \frac {2\omega_1}{\pi}\ri)^2
\wp(S;\omega_1;2\omega_3) -   \frac {\Omega^2}{4}   \frac {\theta_1'''}{3\theta_1'}
         + \mathcal{O}(x^{-1})
    \\
   &=1 +4 \le( X_\infty  - \wp(S;\omega_1;2\omega_3)\ri)         -\re \le  ( \eta_1+\eta_2+ \frac{(\eta_1-\eta_2)^2}8 + \frac {\eta_1}{4\eta_2} + \frac 3 4  \ri)
 + \mathcal{O}(x^{-1}).
 \label{A40}
\end{align}
We can convert $\wp$ into Jacobian elliptic functions $\theta$ by the following chain of equalities (all theta constant are for the modular parameter $\tau^{[1]}  = 2\tau$ and $v = \frac {\pi S}{2\omega_1}$  and $e_1 =  \wp(\omega_1;\omega_1;2\omega_3) = X(\eta_1^{\re})$)
\begin{align}
\label{wptodn}
\wp(S;\omega_1;2\omega_3) - e_1
\mathop{=}^{\text{\tiny  \href{https://dlmf.nist.gov/23.6.E5}{DLMF 23.6.5}  }}&
\le(\frac{\pi \theta_3\theta_4}{2\omega_1} \frac {\theta_2(v;2\tau)}{\theta_1(v;2\tau)}\ri)^2
\mathop{=}^{\text{\tiny  \href{https://dlmf.nist.gov/20.2.iii}{DLMF 20.2.iii}  }} 
-\le(\frac{\pi \theta_3^2}{2\omega_1}\ri)^2\le( \frac {\theta_4 \theta_3(v+\pi \tau;2\tau)}{\theta_3\theta_4(v+\pi \tau;2\tau)}\ri)^2\cr
\mathop{=}^{ { \genfrac{}{}{0pt}{3}
{\text{ \href{https://dlmf.nist.gov/23.6.E1}{DLMF 23.6.x}}}
{\text{ \href{https://dlmf.nist.gov/20.7.E1}{DLMF 20.7.1}}} 
   }}&   (e_3-e_1) \le( \frac {\theta_4 \theta_3(v+\pi \tau;2\tau)}{\theta_3\theta_4(v+\pi \tau;2\tau)}\ri)^2.
\end{align}
Finally, plugging the expressions \eqref{Xinfty}, $e_3 = X(-1), e_1 = X(\eta_1^{\re})$ into \eqref{A40} and simplifying we obtain
\be
|q(x,0)|^2 =  \frac 1 4\le(2 + \eta_1^{\re}+\eta_2^{\re}\ri) ^2  +
\le(\eta_1^{{\rm re}} + 1\ri)\le(\eta_2^{{\rm re}} - 1\ri)
\le( \frac {\theta_4 \theta_3\le(\frac{x \Omega+\Delta}4 +\frac\pi 2 ;2\tau\ri)}{\theta_3\theta_4\le(\frac{x \Omega+\Delta}4 +\frac\pi 2;2\tau\ri)}\ri)^2  + \mathcal{O}(x^{-1}).
\ee
The $\theta$ expression in the bracket is, up to a re-scaling of arguments, the Jacobi elliptic function $\dn$, see \href{https://dlmf.nist.gov/22.2.E6}{DLMF 22.2.6}, i.e.,
\be
 \frac {\theta_4 \theta_3\le(\frac{x \Omega+\Delta}4 +\frac\pi 2 ;2\tau\ri)}{\theta_3\theta_4\le(\frac{x \Omega+\Delta}4 +\frac\pi 2;2\tau\ri)}= \dn \le(  2K(m_1)\le(\frac{x \Omega+\Delta}{4\pi} +\frac 1 2\ri) , m_1 \ri).
\ee
Here, see \href{https://dlmf.nist.gov/22.2.E2}{DLMF 22.2.2}, \eqref{Ome}
\be
m_1 = \frac {\theta_2^2}{\theta_3^2}, \ \ \ K(m_1) = \frac \pi 2\theta_3^2 \Rightarrow  \frac{K\Omega}{2\pi}  =- i \frac K{\omega_1} = - i \frac {\pi\theta_3^2}{2\omega_1} =  \sqrt{e_1-e_3} = \sqrt{\le(\eta_1^{{\rm re}} + 1\ri)\le(1-\eta_2^{{\rm re}} \ri)}
\label{A42}
\ee 
and all $\theta$-constants are evaluated with the modular parameter $\tau^{[1]} = 2\tau$. \end{proof}

    \paragraph{Long time case.} Here we do not provide details because they are simply a specialization of the general case. 
   The solution to the model RHP for the time evolution part proceeds similarly to the process in Subsection \ref{subg1}, except that we need to incorporate the corresponding time evolution terms and the necessary parameter expressions in order to complete the formulation of Theorem \ref{thm14}.
   Finally, applying the same method as in Subsection \ref{sunun} to the case $N=1$ yields the result stated in Theorem \ref{thm15}. 
\section{Proof of (\ref{paxp})}\label{App2}
Inspired by Levermore's work in \cite{Levermore1988}, we provide a proof of the inequality (\ref{paxp}). For any fixed $\ell$, define
\begin{equation}
V(z,\bm{\eta}):=\frac{2Q_{\ell}(z)}{P_{\ell}(z)},\quad G(z):=2Q_{\ell}(z)-\xi P_{\ell}(z),
\end{equation}
where vector $\bm{\eta}=(\eta_1,\dots, \eta_{2 \ell - 1}, \alpha_{\ell},\alpha_{\ell}^{-1},\eta_{2 \ell - 1}^{-1}, \dots, \eta_1^{-1})$ and  $\xi=\frac{2Q_{\ell}(\alpha_\ell)}{P_{\ell}(\alpha_\ell)}$. By condition \eqref{Blet} and \eqref{Clet}, we have
\begin{equation}\label{sgs}
\int_{\gamma^{\rm gap}_{\alpha_{\ell}}}\frac{G(z)}{R_{\ell}(z)}dz=\int_{\bar{\gamma}^{\rm gap}_{\alpha_{\ell}}}\frac{G(z)}{R_{\ell}(z)}dz=0,\quad 
\int_{\gamma^{\rm gap}_j}\frac{G(z)}{R_{\ell}(z)}dz=\int_{\bar{\gamma}^{\rm gap}_j}\frac{G(z)}{R_{\ell}(z)}dz= 0, \quad j = 1, \dots, \ell - 1,
\end{equation}
where all integration contours are along counterclockwise around the unit circle. The definition of $G(z)$ gives $G(\alpha_{\ell}^{\pm 1})=0$. From equations \eqref{sgs} we see that $G(z)$ has other $2\ell + 2$ zeros, besides $\alpha_\ell^{\pm 1}$.
Now, using the symmetries and the fact that $G(\overline z) = \overline{G(z)}$, we see that the function 
\be
U(\theta):= \frac {z G(z)}{R_\ell(z)},\ \ \ z = {\rm e}^{i\theta}
\ee
is real--valued on all the gaps and vanishes at $\arg(\alpha_\ell)$. Given the vanishing of the integrals $\int_{\arg(\gamma^{gap}_\bullet)} U(\theta) \d \theta$ in \eqref{sgs},  we conclude that it must change sign within each gap and hence, by the pigeonhole principle, there is exactly one zero of $G(z)$ in each gap $\gamma^{\rm gap}_{0},\gamma^{\rm gap}_{1}, \dots,  \gamma^{\rm gap}_{\ell-1}, \gamma^{\rm gap}_{\alpha_\ell}$ together with one in each of the corresponding  conjugate gaps. The same argument applies to $\frac{z P_\ell(z)}{R_\ell(z)}$ so that both $G(z)$ and $P(z)$ have exactly one zero in each gap. Note that these zeros cannot be at the endpoints of the gaps. Thus we can write $V(z)$ as follows 

%
\begin{equation*}
V(z;{\boldsymbol \eta})- \xi = \frac{G(z)}{P_{\ell}(z)}=2\frac{(z - \alpha_{\ell}) (z - \alpha_{\ell}^{-1}) \prod_{j = 0}^{\ell } (z - s_j) (z - s_j^{-1})}{z \prod_{j=0}^{\ell} (z-p_j)(z-p_j^{-1})},
\end{equation*}
where $s_\ell, p_{\ell} \in \gamma^{\rm gap}_{\alpha_{\ell}}$ and  $s_j, p_j \in \gamma^{\rm gap}_j$, $ j = 0, \dots, \ell - 1$. This implies that on the interval $\gamma^{\rm gap}_{\alpha_{\ell}}$, $\gamma^{\rm gap}_j$, the function $\frac{G(z)}{P_{\ell}(z)}$ either changes sign twice( if $s_j\neq p_j$) or does not change sign at all (when $s_j=p_j$). It follows that the other sign changes occur at the points $\alpha_{\ell}^{\pm 1}$ due to $\frac{G(z)}{P_{\ell}(z)}\neq 0$ on the intervals $\bigcup_{j=1}^{\ell - 1} (\gamma_j \cup \bar{\gamma}_j)  \cup \gamma_{\alpha_{\ell}} \cup \bar{\gamma}_{\alpha_{\ell}}$. Furthermore
we have
\begin{equation}
\frac{G(-1)}{P_{\ell}(-1)}<0,\quad \frac{G(1)}{P_{\ell}(1)}>0.
\end{equation}
Therefore, we conclude that when $z$ crosses $\alpha_{\ell}$ clockwise on the unit circle, the sign of $\frac{G(z)}{P_{\ell}(z)}$ changes from negative to positive in a neighbourhood of $\alpha_{\ell}$, i.e.
\be
\label{guans}
\frac {\d}{\d \phi} \frac{ G({\rm e}^{i\phi})}{P_\ell({\rm e}^{i\phi})}\bigg|_{\phi = \phi_\ell} = i \alpha_\ell \frac {G'(\alpha_\ell)}{P_\ell(\alpha_\ell)}  >0,
\ee
where $\alpha_\ell = {\rm e}^{i\phi_\ell}$. 
We now turn to the estimate of $\partial_{\phi_\ell} \xi$: the differential $\omega(z;\alpha_\ell):= \frac  1t \d \varphi_\ell = \frac {G(z)}{R_\ell(z)} \d z$ is a real--normalized differential of the second kind with fixed singular behaviour near $z=0, \infty$:
\begin{gather}
 \frac {G(z)}{R_\ell(z)} \d z = \le( 2z - \xi + \mathcal O(z^{-2}\ri )\d z ,\ \ \ z\to\infty;\\
 \frac {G(z)}{R_\ell(z)} \d z = \le( \frac {2}{z^3}  - \frac {\xi}{z^2} + \mathcal O(1)\ri)\d z ,\ \ \ z\to 0,
\end{gather}
on the main sheet of $R_\ell$. Denoting by $\dot{}$ the derivative with respect to $\phi_\ell= \arg(\alpha_\ell)$ we have that $\dot \omega$ must also be a real--normalized differential of the second kind, with singular behaviour

\begin{gather}
 \dot \omega(z;\alpha_\ell)  
= \le( - \dot \xi + \mathcal O(z^{-2}\ri )\d z ,\ \ \ z\to\infty;\\
 \dot \omega(z;\alpha_\ell)= \le( \- \frac {\dot \xi}{z^2} + \mathcal O(1)\ri)\d z ,\ \ \ z\to 0.
\end{gather}
But this is then actually proportional to $\d \varphi_1^{(\ell)} = \frac {P_\ell}{R_\ell} \d z$ since a real--normalized differential is uniquely determined by its singular part near the poles,  namely 
\be
\dot \omega =-\dot \xi  \frac {P_\ell}{R_\ell} \d z.
\ee

On the other hand, denoting
\be
G(z) = (z-\alpha_\ell)(z-\alpha_\ell^{-1}) G_0(z), \ \ \ G_0(z) := \prod_{j=0}^{\ell} (z-s_j)(z-s_j^{-1}),
\ee
we can differentiate with respect to $\phi_\ell$ to get 
\be
\dot \omega = \dot {\le(\frac {G(z)}{R_\ell}\ri)} = \le[\frac {\dot G_0(z)}{G_0(z)} - \frac {i\alpha_\ell}{2(z-\alpha_\ell)} -   \frac {i\alpha_\ell}{2(z-\alpha_\ell^{-1})}\ri] \frac {G}{R_\ell(z) } = -\frac 1 2 \partial_{\phi} \frac {G({\rm e}^{i\phi})}{P_{\ell}({\rm e}^{i\phi})} \bigg|_{\phi=\phi_\ell}.
\ee
In other words, it is obtained that
\be
 -\frac {i\alpha_\ell}2  =\mathop{ \res}_{z=\alpha_\ell} \frac {\dot \omega}{\omega}\d z =  \dot \xi \mathop{\res}_{z=\alpha_\ell} \frac {P_\ell(z)\d z}{ G(z)} = \dot \xi\frac {P_\ell(\alpha_\ell)}{ G'(\alpha_\ell)} \ \ \Rightarrow \ \ 
 \dot \xi = - \frac {i\alpha_\ell}2\frac { G'(\alpha_\ell)}{P_\ell(\alpha_\ell)}.
\ee
Using now  \eqref{guans}  we see  that $\dot \xi = \partial_{\phi_\ell} \xi<0$. 

\section{Error estimates on large-space and long-time behaviours.}
      \label{sectionerrors}
      In this section, we will present the error estimations on the asymptotic behaviours of the initial and the time-evolution potentials. To be specific, we will discuss the local parametrix around the end points of the jump contours of the considering RHPs by using modified Bessel function and Airy function.
      \subsection{Large-space behaviour: local parametrix around $z=\eta_j^{\pm 1}$, $j = 1, \dots, 2N$.}\label{APP1}
      In this section, to derive the error term $\mathcal{O}(x^{-1})$ in \eqref{qx0minus}, we will construct the local parametrix near $z=\eta_1$ via modified Bessel function \cite{Kamv04} in detail, and that near $z=\eta_1^{-1}$ can be derived by utilizing the symmetry property. The other cases near $ z = \eta_j^{\pm 1} $ for $ j = 2, 3,
      \dots, 2N $ can be obtained by a similar process. The method is standard; see, for example, \cite{CMP2021gas,CPAM2023gas}.
      \par 
      To begin with, define the open disk centered at $z=\eta_1$: $D_{\rho}^{\eta_1}:=\{z\in\mathbb{C}||z-\eta_1|<\rho\}$, where $\rho\in\mathbb{R}^+$ is small enough such that $\eta_1^{-1}, \eta_j^{\pm 1} (j = 2, 3, \dots , 2N) \notin D_{\rho}^{\eta_1}$. Then introduce the local conformal map
      \begin{equation}\label{confm}
         \zeta = - \frac{1}{x^2} \varphi^2(z).
      \end{equation}
      For $z \to \eta_1$ and restricting $z \in \{z \in \mathbb{C} | |z| = 1 \}$ we see that
      \begin{equation}\label{zea1}
         \begin{aligned}
            \zeta =& - (z - \eta_1) \frac{( \sum_{j=0}^{N} A_j \eta_1^{2N+2-j} + \sum_{j=0}^{N} A_j \eta_1^j  + A_{N+1}\eta_1^{N+1} )^2}{\eta_1^4 (\eta_1 - \eta_1^{-1}) \prod_{j = 2}^{2N}  (\eta_1 - \eta_j) (\eta_1 - \eta_j^{-1})} + o(z - \eta_1)  \\
            =& \re \left[(z - \eta_1) \frac{( \sum_{j=0}^{N} A_j \eta_1^{2N+2-j} + \sum_{j=0}^{N} A_j \eta_1^j  + A_{N+1}\eta_1^{N+1} )^2}{\eta_1^4 (\eta_1 - \eta_1^{-1}) \prod_{j = 2}^{2N}  (\eta_1 - \eta_j) (\eta_1 - \eta_j^{-1})}\right] + o(z - \eta_1),
         \end{aligned}
      \end{equation}
      where the last equality is obtained from  $ \eta_1^2 z^{-1} - z = -2 (z - \eta_1) + \mathcal{O}\left((z-\eta_1)^2\right)$ which can be used to derive
      \begin{equation}
         \zeta + \bar{\zeta} = - 2 (z - \eta_1) \frac{( \sum_{j=0}^{N} A_j \eta_1^{2N+2-j} + \sum_{j=0}^{N} A_j \eta_1^j  + A_{N+1}\eta_1^{N+1} )^2}{\eta_1^4 (\eta_1 - \eta_1^{-1}) \prod_{j = 2}^{2N}  (\eta_1 - \eta_j) (\eta_1 - \eta_j^{-1})} + \mathcal{O}\left((z-\eta_1)^2\right).
      \end{equation}
      Expression \eqref{zea1} implies that the arc $\{z\in\mathbb{C} | z \in D_{\rho}^{\eta_1} \cap \{z \in \mathbb{C} | |z| = 1, \arg z > \eta_1 \}\}$ is nearly mapped to $(-\infty , 0)$ in the $\zeta$ plane. It follows that the branch cut of $\zeta^{\frac{1}{2}}$ is $(-\infty , 0)$.
      \par 
      Next, we define the transformation 
      \begin{equation}
         P^{(1)}(\zeta) = T(z) {\rm e}^{-i \varphi(z) \sigma_3} \left[ \delta(z) \sqrt{ r (z)} \right]^{-\sigma_3} {\rm e}^{\frac{-i \pi}{4} \sigma_3}
         \begin{pmatrix}
            0 & 1 \\
            1 & 0
         \end{pmatrix}, \quad z \in D_{\rho}^{\eta_1}.
      \end{equation}
      Then the jump matrices of $P^{(1)}(\zeta)$ is presented in Figure \ref{Bess}.
      \begin{rmk}
         Note that the conformal map \eqref{confm} locally maps the contours in the region $D_{\rho}^{\eta_1}$ to the symmetry-contours corresponding to the real line in the $\zeta$ plane (Figure \ref{Bess}). Indeed, we observe that $\overline{\varphi(\bar{z}^{-1})}^2 = \varphi^2(z)$. It follows that $\zeta(\Gamma_{+}) = \overline{\zeta(\Gamma_{-})}$.
      \end{rmk}
      \begin{figure}
         \centering  
         \subfigbottomskip=10pt 
         \subfigcapskip=-5pt 
         \subfigure[The jump matrices of $P^{(1)}(\zeta)$ around $\zeta=0$.]{
            \includegraphics[height=5cm]{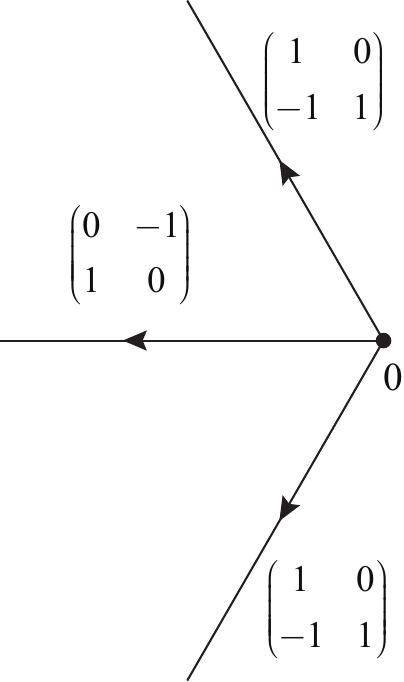}\label{Bess}}
         \quad\quad\quad\quad\quad\quad\quad\quad
         \subfigure[The jump matrices of $P^{(2)}(\zeta)$ around $\zeta=0$.]{
            \includegraphics[height=5cm]{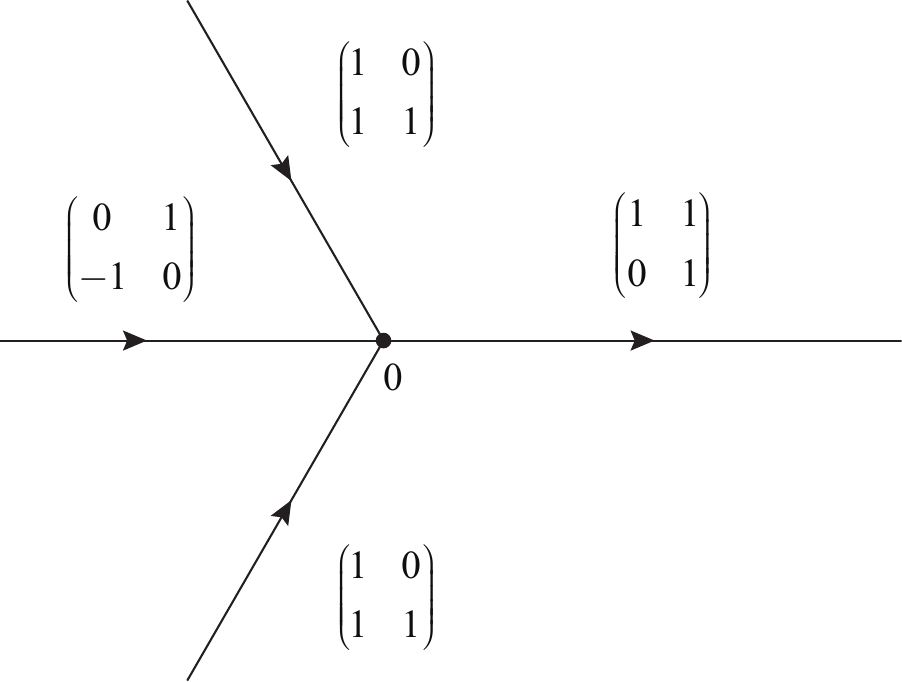}\label{Airy}}
         \caption{Jump matrices of $P^{(1)}(\zeta)$ and $P^{(2)}(\zeta)$ on their respective $\zeta$-planes around $\zeta=0$.}
         \label{Ai-Be}
      \end{figure}
      \par 
      Now we introduce the model RHP for the Bessel function \cite{Kamv04}.
      Then we introduce the RHP of the standard model parametrix $M_{\rm Be}$ \cite{Kamv04} as following:
      \begin{RHP}\label{BeRH}\
         Find a matrix-valued function $M_{\rm Be}(\zeta)$, satisfying the following properties
         \begin{enumerate}
            \item $M_{\rm Be}(\zeta)$ is analytic for $\zeta\in\mathbb{C}\setminus\gamma_{\rm Be}$,
            where $\gamma_{\rm Be}=\gamma_{{\rm Be}\pm}\cup\gamma_{{\rm Be}0}$. Here $\gamma_{{\rm Be}\pm}:=\{\arg\zeta=\pm\frac{2\pi}{3}\}$ and $\gamma_{{\rm Be}0}=\{\arg\zeta=\pi\}$.
            \item Jump matrices: 
            \begin{equation}
               M_{{\rm Be}+}(\zeta)=M_{{\rm Be}-}(\zeta)\begin{cases}
                  \begin{pmatrix}
                     1 & 0 \\
                     1 & 1
                  \end{pmatrix}, &\zeta\in\gamma_{{\rm Be}+}\cup\gamma_{{\rm Be}-},  \\
                  \begin{pmatrix}
                     0 & 1 \\
                     -1 & 0
                  \end{pmatrix}, &\zeta\in\gamma_{{\rm Be}0}.
               \end{cases}
            \end{equation}
            \item Asymptotic condition:
            \begin{equation}
               M_{\rm Be}(\zeta)=(2\pi\zeta^{\frac{1}{2}})^{\frac{-1}{2}\sigma_3}\frac{1}{\sqrt{2}}\begin{pmatrix}
                  1 & i \\
                  i & 1
               \end{pmatrix}
               \left[I+\mathcal{O}(\zeta^{\frac{-1}{2}})\right]
               {\rm e}^{2\zeta^{\frac{1}{2}}\sigma_3}\quad \zeta\to\infty.
            \end{equation}
            \item The singularity at $\zeta=0$:
            \begin{equation}
               M_{\rm Be}(\zeta)=\begin{pmatrix}
                  \mathcal{O}(\ln |\zeta|) & \mathcal{O}(\ln |\zeta|) \\
                  \mathcal{O}(\ln |\zeta|) & \mathcal{O}(\ln |\zeta|)
               \end{pmatrix}.
            \end{equation}
         \end{enumerate}\hfill $\triangle$
      \end{RHP}
      The solution of the RHP \ref{BeRH} is given by
      \begin{equation}
         M_{\rm Be}(\zeta)=\begin{cases}
            \begin{pmatrix}
               I_1(2\zeta^{1/2}) & \frac{i}{\pi}K_2(2\zeta^{1/2}) \\
               2\pi i\zeta^{1/2}I_1'(2\zeta^{1/2}) & -2\zeta^{1/2}K_2(2\zeta^{1/2}) 
            \end{pmatrix}, & |\arg\zeta|<\frac{2\pi }{3},
            \\
            \begin{pmatrix}
               \frac{1}{2} H_1(2i\zeta^{1/2}) & \frac{1}{2} H_2(2i\zeta^{1/2}) \\
               \pi \zeta^{1/2}H_1'(2i\zeta^{1/2}) & \pi \zeta^{1/2}H_2'(2i\zeta^{1/2}) 
            \end{pmatrix}, & \frac{2\pi}{3}<\arg\zeta<\pi,
            \\
            \begin{pmatrix}
               \frac{1}{2} H_2(2i\zeta^{1/2}) & \frac{-1}{2}H_2(2i\zeta^{1/2}) \\
               -\pi \zeta^{1/2}H_2'(2i\zeta^{1/2}) & \pi \zeta^{1/2}H_1'(2i\zeta^{1/2}) 
            \end{pmatrix}, & -\pi<\arg\zeta<-\frac{2\pi}{3},
         \end{cases}
      \end{equation}
      where $I_1(\zeta)$ and $K_2(\zeta)$ represent modified Bessel functions of the first and second kind, respectively, while $H_1(\zeta)$ and $H_2(\zeta)$ is Hankel function.
      \par 
      It follows that the definition of the local parametrix near $z = \eta_1$ is given by
      \begin{equation}
         P^{\eta_1}(z) = N^{\eta_1}_{\rm Be}(z) M_{\rm Be}(\zeta(z)) \begin{pmatrix}
            0 & 1 \\
            1 & 0
         \end{pmatrix}
         {\rm e}^{i \varphi(z) \sigma_3} \left[ \delta(z) \sqrt{ r (z)} \right]^{\sigma_3} {\rm e}^{\frac{i \pi}{4} \sigma_3},
      \end{equation}
      where $N^{\eta_1}_{\rm Be}(z)$ is an analytic function in $z \in D_{\rho}^{\eta_1}$ defined by letting 
      \begin{equation}\label{imposi}
         P^{\eta_1}(z) \left[Y(z)\right]^{-1}=\mathbb{I} + \mathcal{O}(|x|^{-1}), \quad x \to -\infty, \quad z \in \partial D_{\rho}^{\eta_1} \setminus \zeta^{-1}(\gamma_{\rm Be}).
      \end{equation}
      \par 
      So far, we have constructed the local parametrix near $z=\eta_1$. The local parametrix at $z = \eta_1^{-1}$ can be obtained by defining $P^{\eta_1^{-1}}(z) = \sigma_1 \overline{P^{\eta_1}(\bar{z})} \sigma_1$. Similarly, the local parametrix at $z = \eta_j^{\pm 1}$ defined by $P^{\eta_j^{\pm 1}}(z)$, $ j = 2, 3,
      \dots ,2N $ can also be constructed by using modified Bessel function.
      \par 
      Now we are ready to define the error matrix as follows
      \begin{equation}\label{Dee}
         E(z) = T(z) [P(z)]^{-1},
      \end{equation}
      where 
      \begin{equation}\label{Pdef}
         P(z) = \begin{cases}
            Y(z) , & z \in \mathbb{C} \setminus \cup_{j=1}^{2N} D_{\rho}^{\eta_j^{\pm 1}}, \\
            P^{\eta_j^{\pm 1}}(z), & z \in D_{\rho}^{\eta_j^{\pm 1}}, \quad j = 1,\dots, 2N,
         \end{cases}
      \end{equation}
      where $D_{\rho}^{\eta_j^{\pm 1}}:= \{z \in \mathbb{C} | |z-\eta_j^{\pm 1}| < \rho\}$, $j = 1,\dots, 2N$. Note that the definition of $P(z)$ \eqref{Pdef} kills some jump matrices of $T(z)$ due to the definition \eqref{Dee}. Observing the exponential decaying term of the jump matrices of $T(z)$ and combining it with the restriction \eqref{imposi}, we derive the order of approximation of the jump matrices of $E(z)$ to $\mathbb{I}$ as $|x| \to +\infty$ (illustrated in Figure \ref{open-err}).
      \begin{figure}
         \centering  
         \subfigbottomskip=10pt 
         \subfigcapskip=-5pt 
         \subfigure[The case of large $|x|$ for $ j = 1,\dots, N$.]{
            \includegraphics[height=5cm]{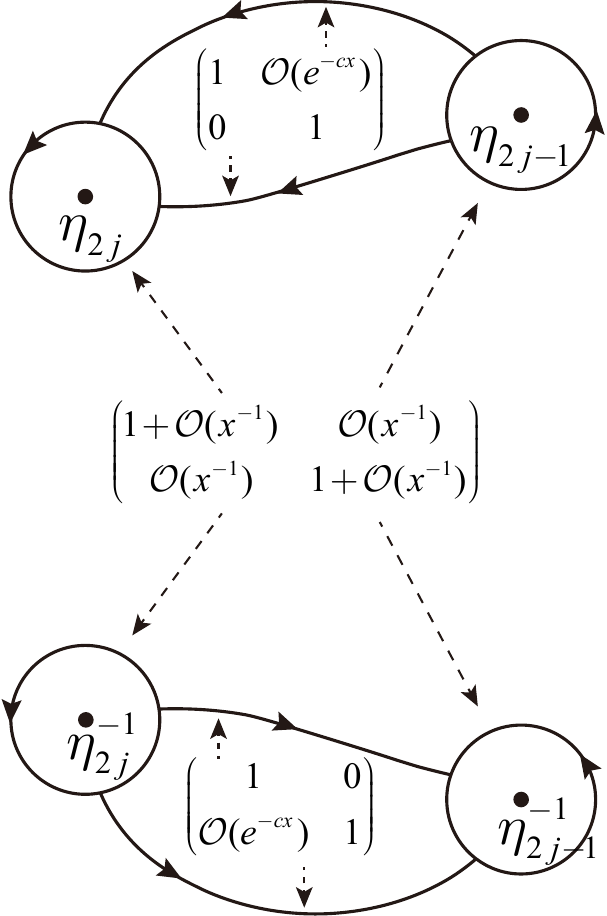}\label{open-err}}
         \quad\quad\quad\quad\quad\quad\quad\quad
         \subfigure[The case of long time.]{
            \includegraphics[height=5cm]{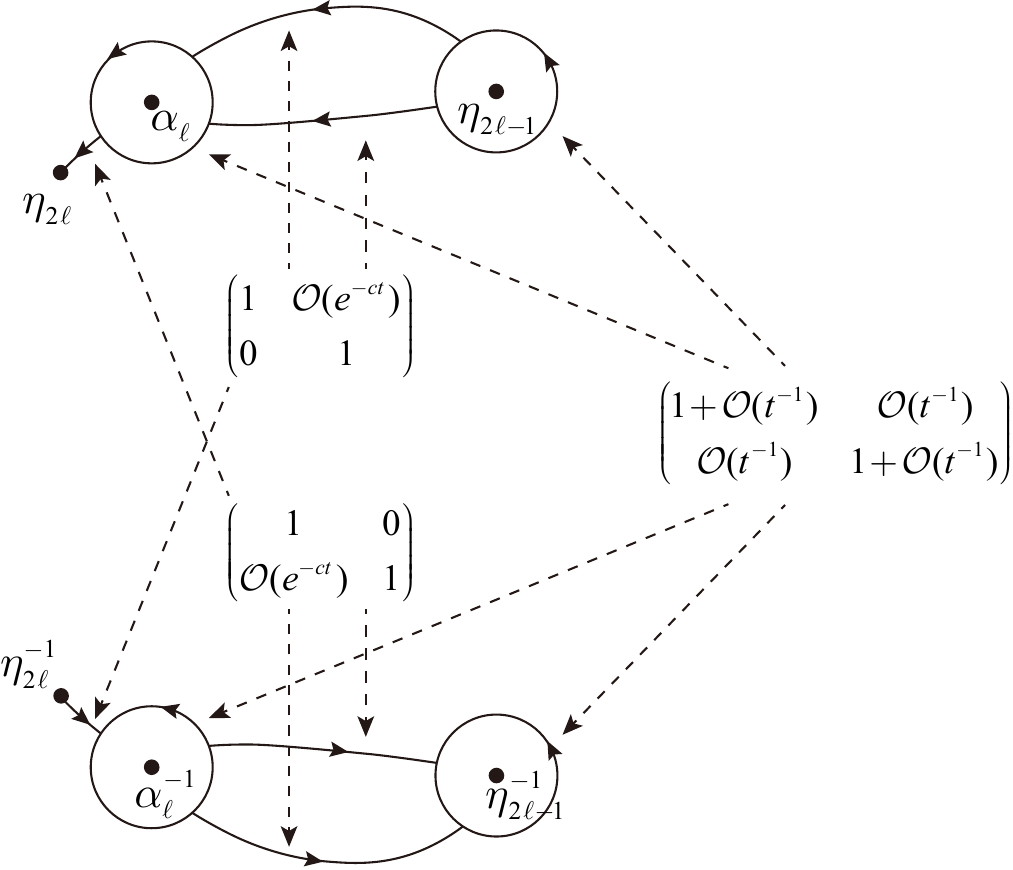}\label{open1-err}}
         \caption{The rate at which the jump matrices of $E(z)$ and $\tilde{E}(z)$ approach $\mathbb{I}$ as large-$|x|$ and long-time behaviours, where $c\in\mathbb{R}^+$ is a fixed constant.}
         \label{erro}
      \end{figure}
      \par 
      By the standard small norm theory \cite{Its11}, for $z \to \infty$, we obtain
      \begin{equation}\label{E1s}
         E(z) = \mathbb{I} + \frac{E^{(1)}}{z} + \mathcal{O}(z^{-2}), 
      \end{equation}
      where $E^{(1)} = \mathcal{O}(x^{-1})$.
      \subsection{Long-time behaviour: local parametrix around $z=\alpha_{\ell}^{\pm1}$ and $z=\eta_j^{\pm1}$, $j = 2, \dots, 2\ell - 1$.}
      The local parametrix near $z=\alpha^{\pm1}_{\ell}$ can be constructed by using the Airy function. In this section, we will consider the local parametrix at $z = \alpha_{\ell}$ in detail. Firstly, define the vicinity of $\alpha_{\ell}$: $D_{\rho}^{\alpha_{\ell}}:= \{z\in\mathbb{C} | |z-\alpha_{\ell}|<\rho\}$. Then, introduce the local conformal map as follows
      \begin{equation}
         \zeta^{\frac{3}{2}} = -\frac{3 i}{2 t} \left[ \varphi_{\ell}(z) -\frac{\tilde{\Omega}_{\ell}}{2} \right].
      \end{equation}
      Restricting $z\in \{z\in\mathbb{C} | |z| = 1\}$ and recalling the equation \eqref{abx}, for $z \to \alpha_{\ell}$, we obtain
      \begin{equation}
         \zeta = \re \left[(\alpha_{\ell}-z) \left(\frac{ \sqrt{\alpha_{\ell} - \alpha_{\ell}^{-1}} \prod_{j = 0}^{\ell } (\alpha_{\ell} - s_j) (\alpha_{\ell} - s_j^{-1})}
         {\alpha_{\ell}^3 \prod_{j = 1}^{ 2 \ell - 1 } \sqrt{\alpha_{\ell} - \eta_j} \sqrt{\alpha_{\ell} - \eta_j^{-1}}}\right)^{\frac{2}{3}} \right] + o(z - \alpha_{\ell}).
      \end{equation}
      Therefore the branch cut of $\zeta^{\frac{3}{2}}$ is $(-\infty, 0)$, corresponding to $z \in \gamma_{\alpha_{\ell}}$.
      \par 
      Now introduce the following transformation
      \begin{equation}
         P^{(2)}(\zeta) = \tilde{T}(z) {\rm e}^{-i \varphi_{\ell}(z) \sigma_3} \left[\delta_{\ell}(z) \sqrt{ r (z)}\right]^{-\sigma_3}
         {\rm e}^{\frac{i \pi }{4} \sigma_3} \begin{pmatrix}
            0 & 1 \\
            1 & 0
         \end{pmatrix},
      \end{equation}
      whose jump matrices are illustrated in Figure \ref{Airy}. 
      \par
      Next, we introduce the following model RHP corresponding to the Airy function \cite{Dei99,Dkm+}:
      \begin{RHP}\label{AirRH}\
         Find a $2\times2$ matrix-valued function $M_{\rm Ai}(\zeta)$, satisfying the following properties
         \begin{enumerate}
            \item $M_{\rm Ai}(\zeta)$ is analytic for $\zeta\in\mathbb{C}\setminus\gamma_{\rm Ai}$,
            where $\gamma_{\rm Ai}=\gamma_{{\rm Ai}\pm}\cup\gamma_{{\rm Ai}0-}\cup\gamma_{{\rm Ai}0+}$. Here $\gamma_{{\rm Ai}\pm}:=\{\arg\zeta=\pm\frac{2\pi}{3}\}$, $\gamma_{{\rm Ai}0+}=\{\arg\zeta=0\}$ and $\gamma_{{\rm Ai}0-}=\{\arg\zeta=\pi\}$.
            \item Jump matrices: 
            \begin{equation}
               M_{{\rm Ai}+}(\zeta)=M_{{\rm Ai}-}(\zeta)\begin{cases}
                  \begin{pmatrix}
                     1 & 0 \\
                     1 & 1
                  \end{pmatrix}, &\zeta\in\gamma_{{\rm Ai}+}\cup\gamma_{{\rm Ai}-},  \\
                  \begin{pmatrix}
                     0 & 1 \\
                     -1 & 0
                  \end{pmatrix}, &\zeta\in\gamma_{{\rm Ai}0-},   \\
                  \begin{pmatrix}
                     1 & 1 \\
                     0 & 1
                  \end{pmatrix}, &\zeta\in\gamma_{{\rm Ai}0+}.
               \end{cases}
            \end{equation}
            \item Asymptotic condition:
            \begin{equation}\label{airm}
               M_{\rm Ai}(\zeta)=\zeta^{\frac{-1}{4}\sigma_3}\frac{1}{\sqrt{2}}\begin{pmatrix}
                  1 & i \\
                  i & 1
               \end{pmatrix}
               \left[I+\mathcal{O}(\zeta^{\frac{-3}{2}})\right]
               {\rm e}^{\frac{-2}{3}\zeta^{\frac{3}{2}}\sigma_3},\quad \zeta\to\infty.
            \end{equation}
            \item $M_{\rm Ai}$ is bounded near $\zeta=0$, $\zeta\in\mathbb{C}\setminus\gamma_{{\rm Ai}}$.
         \end{enumerate}\hfill $\triangle$
      \end{RHP}
      The solution of the RHP \ref{AirRH} is written as
      \begin{align}
         M_{\rm Ai}(\zeta)=&\sqrt{2\pi}\begin{pmatrix}
            {\rm Ai}(\zeta) & -{\rm e}^{\frac{4\pi i}{3}}{\rm Ai}({\rm e}^{\frac{4\pi i}{3}}\zeta)  \\
            -i{\rm Ai}'(\zeta) & i{\rm e}^{\frac{2\pi i}{3}}{\rm Ai}'({\rm e}^{\frac{4\pi i}{3}}\zeta)
         \end{pmatrix}, && 0<\arg \zeta<\frac{2\pi }{3},
         \\
         M_{\rm Ai}(\zeta)=&\sqrt{2\pi}\begin{pmatrix}
            -{\rm e}^{\frac{2\pi i}{3}}{\rm Ai}({\rm e}^{\frac{2\pi i}{3}}\zeta) & -{\rm e}^{\frac{4\pi i}{3}}{\rm Ai}({\rm e}^{\frac{4\pi i}{3}}\zeta)  \\
            i{\rm e}^{\frac{4\pi i}{3}}{\rm Ai}'(\zeta) & i{\rm e}^{\frac{2\pi i}{3}}{\rm Ai}'({\rm e}^{\frac{4\pi i}{3}}\zeta)
         \end{pmatrix}, && \frac{2\pi}{3}<\arg \zeta<\pi,
         \\
         M_{\rm Ai}(\zeta)=&\sqrt{2\pi}\begin{pmatrix}
            -{\rm e}^{\frac{4\pi i}{3}}{\rm Ai}({\rm e}^{\frac{4\pi i}{3}}\zeta) & {\rm e}^{\frac{2\pi i}{3}}{\rm Ai}({\rm e}^{\frac{2\pi i}{3}}\zeta)  \\
            i{\rm e}^{\frac{2\pi i}{3}}{\rm Ai}'({\rm e}^{\frac{4\pi i}{3}}\zeta) & -i{\rm e}^{\frac{4\pi i}{3}}{\rm Ai}'(\zeta)
         \end{pmatrix}, && -\pi<\arg \zeta<-\frac{2\pi}{3},
         \\
         M_{\rm Ai}(\zeta)=&\sqrt{2\pi}\begin{pmatrix}
            {\rm Ai}(\zeta) & {\rm e}^{\frac{2\pi i}{3}}{\rm Ai}({\rm e}^{\frac{2\pi i}{3}}\zeta)  \\
            -i{\rm Ai}'(\zeta) & i{\rm e}^{\frac{4\pi i}{3}}{\rm Ai}'(\zeta)
         \end{pmatrix}, && -\frac{2\pi}{3}<\arg \zeta<0,
      \end{align}
      where ${\rm Ai}(\zeta)$ is the Airy function.
      \par 
      Therefore, the local parametrix near $z = \alpha_{\ell}$ is defined by
      \begin{equation}
         \tilde{P}^{\alpha_{\ell}}(z) = N_{\rm Ai}^{\alpha_{\ell}}(z) M_{\rm Ai}(\zeta(z)) \begin{pmatrix}
            0 & 1 \\
            1 & 0
         \end{pmatrix}
         {\rm e}^{i \varphi_{\ell}(z) \sigma_3} \left[\delta_{\ell}(z) \sqrt{ r (z)}\right]^{\sigma_3} {\rm e}^{\frac{-i \pi }{4} \sigma_3},
      \end{equation}
      where $N_{\rm Ai}^{\alpha_{\ell}}(z)$ is an analytic function in $D_{\rho}^{\alpha_{\ell}}$ defined by
      \begin{equation}\label{Nai}
         N_{\rm Ai}^{\alpha_{\ell}}(z) \left[\tilde{Y}(z)\right]^{-1} = \mathbb{I} + \mathcal{O}(t^{-1}), \quad t \to +\infty, \quad z \in \partial D_{\rho}^{\alpha_{\ell}} \setminus \zeta^{-1}(\gamma_{\rm Ai}).
      \end{equation}
      \par 
      Note that the local parametrix near $z = \eta_j^{\pm 1}$ (written as $\tilde{P}^{ \eta_j^{\pm 1}}(z)$ ), $ j = 1, \dots, 2 \ell - 1 $, can be constructed via the modified Bessel function similar to Appendix \ref{APP1}, and thus we omit the details. The local parametrix near $z = \alpha_{\ell}^{-1}$, $\tilde{P}^{\alpha_{\ell}^{-1}}(z)$, can be obtained by using the symmetry: $\tilde{P}^{\alpha_{\ell}^{-1}}(z) = \sigma_1 \overline{\tilde{P}^{\alpha_{\ell}}(\bar{z})} \sigma_1$.
      \par 
      Now we are ready to define the error matrix
      \begin{equation}\label{te}
         \tilde{E}(z) = \tilde{T}(z) \left[\tilde{P}(z)\right]^{-1}, 
      \end{equation}
      where $\tilde{P}(z)$ is defined by
      \begin{equation}\label{tp}
         \tilde{P}(z) = \begin{cases}
            \tilde{T}(z), & z \in \mathbb{C} \setminus \left(  \bigcup_{j=1}^{2 \ell - 1} D_{\rho}^{\eta_1^{\pm 1}} \cup D_{\rho}^{\alpha_{\ell}^{\pm 1}}\right), \\
            \tilde{P}^{\alpha_{\ell}^{\pm 1}}(z) , & z \in D_{\rho}^{\alpha_{\ell}^{\pm 1}}, \\
            \tilde{P}^{\eta_1^{\pm 1}}(z) , & z \in D_{\rho}^{\eta_j^{\pm 1}}, \quad j = 1, \dots, 2 \ell - 1,
         \end{cases}
      \end{equation}
      where $D_{\rho}^{\alpha_{\ell}^{-1}}:=  \{z\in\mathbb{C} | |z-\alpha_{\ell}^{-1}|<\rho\}$.
      \par 
      By the small norm theory and combining with \eqref{Nai}, we obtain the expansion of $\tilde{E}(z)$ for $z \to \infty$:
      \begin{equation}\label{tE1}
         \tilde{E}(z) = \mathbb{I} + \frac{\tilde{E}^{(1)}}{z} + \mathcal{O}(z^{-2}),
      \end{equation}
      where $\tilde{E}^{(1)} = \mathcal{O}(t^{-1})$ and $\partial_x \tilde{E}^{(1)}$ is bounded.
%
%


\bibliographystyle{amsplain}

\end{document}